\documentclass[11pt,a4paper]{article}
\usepackage{jheppub}

\usepackage{multirow, graphicx,amssymb,url,mathrsfs,amsmath}
\usepackage{wrapfig,boxedminipage,setspace,subfigure,epsfig}
\usepackage{amsxtra,amstext,latexsym,dsfont,amsfonts}
\usepackage{color,eucal}
\usepackage[dvipsnames]{xcolor}
\usepackage{float}
\usepackage{amsthm}
\usepackage{slashed,comment}
\usepackage{kotex}
\usepackage{tensor}
\usepackage{mathtools}
\usepackage{indentfirst}
\usepackage{graphicx}
\usepackage{grffile} % 允许空格和多个点

\usepackage{array} 
\usepackage[section]{placeins}

\newcommand{\bra}[1]{\mbox{$\langle #1 |$}}
\newcommand{\ket}[1]{\mbox{$| #1 \rangle$}}

\newcommand{\ie}{{\it i.e.}}

\newtheorem{theorem}{Theorem}  
\newtheorem{lemma}{Lemma}

\makeatletter

\renewcommand{\maketag@@@}[1]{\hbox{\m@th\normalsize\normalfont#1}}%

\makeatother

\title{Holographic multipartite entanglement {structures in} IR modified geometries}

\author[a]{Xin-Xiang Ju,}
\author[a]{Bo-Hao Liu,}
\author[a,b]{Ya-Wen Sun,}
\author[a]{Bo-Yu Xu} 
\author[a]{and Yang Zhao}

\emailAdd{juxinxiang21@mails.ucas.ac.cn}
\emailAdd{liubohao16@mails.ucas.ac.cn}
\emailAdd{yawen.sun@ucas.ac.cn}
\emailAdd{xuboyu25@mails.ucas.ac.cn}
\emailAdd{zhaoyang20a@mails.ucas.ac.cn}

\affiliation[a]{School of Physical Sciences, University of Chinese Academy of Sciences, Beijing 100190, China}
\affiliation[b]{Kavli Institute for Theoretical Sciences, University of Chinese Academy of Sciences, Beijing
100049, China}

\abstract{We investigate how IR modifications of the bulk geometry reshape long-range multipartite entanglement on the boundary in holography. We modify the IR geometries in two opposite directions: spherical modifications that enhance long-range entanglement and hyperbolic modifications that suppress them. We utilize various multipartite entanglement measures/signals to analyze the multipartite entanglement structures. These measures/signals are combinations of entanglement entropy, multi-entropy, entanglement wedge cross sections (EWCS) and multi-EWCS. Our results reveal that in the extremal limits of these two geometric modifications, the multipartite entanglement structures exhibit starkly contrasting behaviors: various measures saturate either their theoretical upper or lower bounds in the respective geometries. This demonstrates that IR deformations provide a practical holographic framework for realizing extremal entanglement regimes. Moreover, it serves as an effective tool for studying quantum marginal problems in holography. Finally, by observing how different measures respond to these engineered geometries, we gain clarifying insights into the specific types of multipartite entanglement that each measure/signal is particularly sensitive to.
}

\arxivnumber{}

\begin{document}
 \maketitle
\flushbottom

%%%%%%%%%%%%%%%%%%%%%%%%%%%%%%%%
%    section: Introduction
%%%%%%%%%%%%%%%%%%%%%%%%%%%%%%%%
	
\section{Introduction}\label{sec1}

\noindent The holographic principle    \cite{Maldacena:1997re} states that a gravitational theory in the bulk is completely described by a quantum field theory living on its boundary. A key insight of this duality is the emergence of spacetime geometry from patterns of entanglement in the boundary state    \cite{Rangamani:2016dms,Swingle:2009bg,VanRaamsdonk:2009ar,VanRaamsdonk:2010pw,Maldacena:2013xja}. In particular, the seminal work of Ryu and Takayanagi    \cite{Ryu_2006} relates the entanglement entropy of a boundary region to the area of a minimal surface in the dual bulk spacetime, thus establishing a precise dictionary between geometric quantities and boundary quantum correlations. This observation suggests that the very fabric of bulk geometry is woven from entanglement: changing the entanglement structure reshapes the geometry, and vice-versa.

Building on this insight, \cite{Ju_2024} has turned to a more refined probe of entanglement, the conditional mutual information (CMI), to uncover how the entanglement at different distance scales is encoded in the radial profile of the bulk. For two infinitesimal boundary subregions separated by a finite distance $l$, the CMI conditioned on the interval between them is entirely determined by the bulk geometry at a specific radial depth $z_* \sim l $ in the bulk \cite{Ji:2025vks}. This relation provides a real-space measure of long scale boundary entanglement that is directly sensitive to the IR geometric data in the bulk. In this sense, the IR geometry determines the long range entanglement structures at the boundary while the UV geometry determines the short range ones. However, it should be noted that this “real space entanglement/bulk radial scale” correspondence does not conflict with the familiar UV/IR 
relation of holography: the apparent locality with respect to boundary distance arises from a particular radial gauge, while the physical observable remains fully covariant.

Within this picture, variations of the bulk IR geometry translate directly into changes in the long-scale entanglement structures of the boundary theory. In \cite{Ju_2024}, we have considered two types of opposite IR geometries, leading to two contrasting patterns of distribution of the conditional mutual information across different length scales. Subsequent works    \cite{Ju:2024kuc,Ju:2024hba,Ju:2025tgg} further showed that, in certain configurations {where CMI approaches its upper bound value}, the CMI not only measures correlations between pairs of regions but quantitatively captures the tripartite entanglement among the three relevant subregions. 

{Unlike the bipartite case, the classification and quantification of multipartite entanglement are considerably more intricate, and no universally accepted measure exists \cite{Xie2021,Ma_2024, Gadde_2023}. This complexity reflects the fact that multipartite entanglement admits qualitatively different structures and cannot be fully captured by a single quantity. In recent years, a variety of multipartite entanglement measures have been proposed in holography. These measures are constructed with quantities that admit a holographic dual such as the entanglement wedge cross section (EWCS) \cite{Umemoto_2018, Basak2025, ahn2025}, its multipartite generalizations (multi-EWCS) \cite{Umemoto_2018multiEWCS, Chu:2019etd, Bao_2019, Bao_2019multiSR, Bao:2018fso}, and the multi-entropy \cite{Gadde_2022,Gadde:2024taa,Gadde:2025csh,iizuka2025GM, iizuka2025MGM, iizuka2025}. Besides, owing to the difficulty of finding a faithful entanglement measure, part of the recent effort has shifted toward the identification of multipartite entanglement signals  \cite{balasubramanian2024, bao2025}\textemdash quantities that may not satisfy the requirements of proper entanglement measures proposed in   \cite{Xie2021, Ma_2024}, but whose nonvanishing nonetheless provides a robust indication of the presence of multipartite entanglement.} Other related discussions on multipartite entanglement in holography could be found in \cite{Bao:2015bfa, DeWolfe:2020vjp, Harper:2021uuq, Harper:2022sky, Zou:2022nuj, Hernandez-Cuenca:2023iqh, Jiang:2025iet, balasubramanian2025, Akella:2025owv, Balasubramanian:2025jhq, Yuan:2025dgx}.

Based on these insights and developments, in this work we  systematically investigate how changes in the bulk IR geometry affect the long-distance multipartite entanglement on the boundary. This question gains particular relevance when studying the entanglement distribution among all subregions of the full boundary pure state, since all subregions have to span large spatial scales. More explicitly, we explore this interplay between bulk geometry and boundary multipartite entanglement through a focused analysis of various multipartite entanglement measures in  IR modified holographic geometries. These measures can be organized into three categories: i) measures derived from the bulk entanglement wedge cross sections (EWCS) e.g the Markov gap   \cite{Akers_2020,Hayden_2021} and the L-entropy   \cite{Basak2025}; ii) signals obtained from multi-EWCS; iii) measures based on the multi-entropy \cite{Gadde_2022}, e.g. $\kappa$   \cite{liu2024}. By monitoring how each measure responds to the two opposite types of IR geometry deformations, we achieve two goals in one stroke. First, this would help us gain further insights into where the long range multipartite entanglement structure lies in the bulk geometry. Second, at the same time this analysis provides more information on what type of multipartite entanglement structure these measures could detect from the change of the behavior of these measures in IR modified geometries. 

{Another important motivation of this paper is to investigate the quantum marginal problem   \cite{schilling2014,schilling2015} in a holographic setting. The quantum marginal problem concerns whether a consistent global quantum state exists given a set of its reduced density matrices (quantum marginals). In our holographic framework, the reduced density matrices of relatively small boundary subregions play the role of these marginals. Their entanglement wedges lie in the UV region, which is unaffected when we modify the IR geometry; consequently, their density matrices are also unchanged by subregion–subregion duality   \cite{Czech:2012bh,Wall:2012uf,Headrick:2014cta}. In this sense, modifying the IR geometry corresponds to constructing the full boundary state from fixed marginals, which is precisely the task addressed by the quantum marginal problem. Studying the extremal values of multipartite entanglement measures/signals in the resulting global states then provides holographic constraints on such constructions.}

The rest of this paper is organized as follows. In Section \ref{sec2}, we review the two types of IR modified geometries that we introduced in   \cite{Ju_2024}: the spherical and hyperbolic extremal geometries. We then describe the corresponding changes in the boundary behavior of the conditional mutual information in the two opposite IR modified geometries, revealing two opposite directions in the redistribution of quantum entanglement across different length scales. In Section \ref{sec3}, we show how the EWCS varies in the two types of IR modified geometries and then give the results for the change in two multipartite entanglement measures: the Markov gap and L-entropy, which are consistent with the expectation from the behavior of CMI on these two types of geometries. {It will be shown that the two types of IR modified geometries correspond to upper and lower bound values of the measures, respectively.} In Section \ref{section multi EWCS}, we investigate the modifications of the multi-EWCS on these two types of IR modified geometries and the behaviors of multipartite entanglement measures constructed from the multi-EWCS, including a new signal that we define which vanishes in the hyperbolic IR modified geometry. In Section \ref{multi entropy}, we consider the change in how the multi-entropy and related measures behave in the two IR modified geometries. Finally, we conclude and discuss our results in Section \ref{sec7}.

\section{Review of IR modified geometries and the corresponding long scale entanglement structures}\label{sec2}

\noindent In holography, the geometry of the bulk spacetime encodes entanglement structures of the boundary quantum field theory. A particularly illuminating probe of this relationship is the conditional mutual information (CMI) of boundary subregions. In   \cite{Ju_2024}, it was demonstrated that for two infinitesimal boundary subregions separated by a finite distance \( l \), their CMI—under the condition of the region between them—is holographically determined by the bulk geometry at a specific radial scale corresponding to \( l \). This quantity can be naturally interpreted as a real-space measure of long-distance entanglement in the boundary theory. The validity of this geometric interpretation was rigorously established in   \cite{Ji:2025vks}. Notably, this real-space/radial scale correspondence does not contradict the conventional UV/IR relation in holography, as the seemingly local real-space correspondence explicitly depends on the choice of gauge in the radial coordinate. Within this framework, modifications to the IR geometry of the bulk correspond directly to changes in the large-scale entanglement structure of the boundary theory. Furthermore,   \cite{Ju:2025tgg} revealed that, for certain configurations, the CMI computed in this setting quantitatively captures tripartite entanglement among the three relevant boundary subregions. 

These insights motivate a broader investigation into how changes of the bulk IR geometry affect the multipartite entanglement structures at the boundary. This is of particular interest when considering the entanglement distribution across the entire boundary pure state, where the global multipartite entanglement necessarily involves correlations over long spatial scales. In this section, we review the specific modifications to the IR geometry introduced in \cite{Ju_2024, Ju:2023bjl} and summarize the resulting behavior of boundary CMI. In subsequent sections, we will explore how these IR geometric deformations influence the structure of multipartite entanglement in the boundary theory, thus providing a holographic perspective on the interplay between geometry and multipartite entanglement across scales.

In \cite{Ju_2024}, we proposed two classes of toy-model IR geometries dual to two opposite types of long distance entanglement structures.
To accomplish this, the IR geometry was modified into two distinct forms\textemdash the spherical and the hyperbolic  geometries\textemdash and their corresponding entanglement behaviors were analyzed. In the extremal limit of the former case, the IR region becomes an entanglement shadow impenetrable to any RT surface, while in the extremal case of the latter, the boundary of the IR region is equivalent to an end-of-the-world (EoW) brane. The evaluation of boundary CMI in the two geometries shows that such modifications result in a redistribution of entanglement across different distance scales.

\subsection{Two opposite ways to modify the IR geometry }

\begin{figure} [H]
    \centering 
   \includegraphics[width=0.8\textwidth]{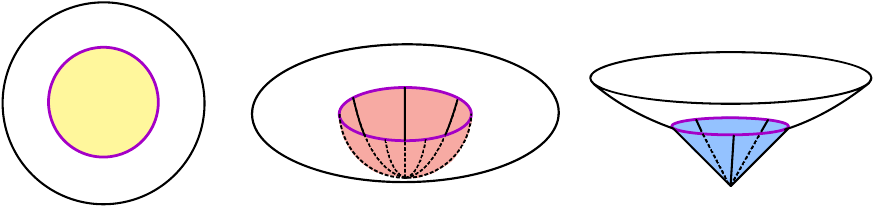} 
   \caption{A graphical summary of the geometries in   \cite{Ju_2024} with modified IR regions. The left figure depicts the general case where the geometry of the IR region (shown in yellow) is modified in a global AdS$_3$ spacetime. Displayed in purple is the edge of this IR region, where matter fields reside and spatial connection conditions are imposed. The middle and the right figures depict the diametrically opposite toy-model geometries obtained through such modification. The middle figure shows the spherical extremal case, where the IR region with large positive curvature (shown in red) can be viewed as a hemisphere embedded in an imaginary Euclidean space. The right figure, on the other hand, shows the hyperbolic extremal case, where the IR region with extremely negative curvature (shown in blue) infinitely approaches a light cone embedded in an imaginary Minkowski spacetime. }
    \label{modIR} 
\end{figure}

\noindent {As shown in the left figure of figure \ref{modIR}, in a global AdS$_3$ spacetime, we could pick an IR region (the yellow region) with an edge at $r=r_{IR}$ and replace this IR region with other geometries while the outside UV geometry is kept fixed.}
This is achieved by adding matter fields on the IR boundary which form a matter brane. The spatial connection condition on the IR edge demands that the geometries on both sides meet with the same induced metric on the edge. We also demand that the Cauchy slice always have zero extrinsic curvature, so we can use RT formula on it to obtain the entanglement entropy. 

Further constraints arise from the Gauss-Bonnet theorem, which implies that once the integration of the curvature scalar is fixed, any modification of the IR geometry amounts to a redistribution of curvature within the IR region. There exist two directions of this redistribution, specifically, to increase or decrease the curvature at the center of the IR region.
The extremal outcomes of the directions are two diametrically opposed geometries, as shown in the middle and the right subfigures of figure \ref{modIR}: 

\begin{itemize}
    \item Spherical extremal case: The curvature at the center of the IR region greatly increases, while near the edge it decreases.\footnote{``Near" refers to a thin shell with negligible thickness near the edge.} This extremely large central curvature prevents geodesics from penetrating the IR region, rendering the IR region an entanglement shadow.
    \item Hyperbolic extremal case: The curvature at the center of the IR region greatly decreases, while near the edge it increases. As the curvature inside the IR region approaches negative infinity, geodesic lengths in the IR region tend to zero. The edge of the IR region is equivalent to an EoW brane in the sense of holographic entanglement entropy.
\end{itemize}

With the Cauchy slice prepared in either extremal configuration, we can evolve it forward and backward in time according to Einstein’s equations. For matter fields within the IR region, the null energy condition (NEC) is imposed, while violations of the weak energy condition (WEC) are allowed. 

For both extremal cases introduced above, we study their geometrical structures and resulting RT surface behaviors as the foundation of further entanglement discussion. In the spherical extremal case, we can view the IR geometry with increased positive curvature as a hemisphere, with the boundary of the IR region identified with the equator, as shown in figure \ref{CMI}. The ansatz of the bulk spacetime metric for modified IR geometry is
\begin{equation}\label{ansatz}
ds^2=-f(r)g(r) dt^2+\frac{1}{g(r)}dr^2+r^2d\theta^2.
\end{equation} 
On the $t=0$ Cauchy slice with zero extrinsic curvature, for general spherical cases\footnote{Here, ``general" means that the curvature within the IR region is increased, but the geometry is not necessarily extremal.}
\begin{equation}\label{g(r) 1}
g(r)=\begin{cases}
(l^2-r^2)/l^2, & \text{for } r<r_{IR},\\
(l_{AdS}^2+r^2)/(l_{AdS}^2), & \text{for } r>r_{IR}.
\end{cases}
\end{equation}
Here $r_{IR}$ denotes the radial location of the gluing edge of the IR region, {and $l_{AdS}$ is the AdS radius. For general spherical cases, the IR geometry can be regarded as a spherical crown, with $l$ being its radius.} {Note that in the extremal limit, the spherical crown becomes a hemisphere, and $l=r_{IR}$. {In this limit, the $g_{rr}$ metric component becomes infinity at the edge of the IR region so that the whole IR region becomes an entanglement shadow where no geodesics from the boundary could enter}.}
 
Within this framework, we focus on the shapes of RT surfaces which undergo four phases. 

I. For a sufficiently small boundary subregion, its RT surface lies entirely in the unmodified exterior region, identical to that in pure AdS. As the subregion size $L$ increases, it will reach a critical value $L_c$, with the RT surface being tangent to the edge of the IR region.

II. For $L_c<L<\pi l_{S^1}$\footnote{In both extremal cases, $2\pi l_{S^1}$ denotes the spatial size of the boundary system.}, the RT surface wraps around the IR region without penetrating it, as shown in the upper left figure of figure \ref{CMI}. It consists of three smoothly joined geodesic arcs: two in the unmodified exterior tangent to the edge of the IR region, and one along the edge.

III. For a subregion with $\pi l_{S^1}<L<2\pi l_{S^1}-L_c$, its RT surface is the same as that of its complement in phase II. Its entanglement wedge now includes the IR region. 

IV. For $L>2\pi l_{S^1}-L_c$, the RT surface detaches from the edge of the IR region, and reduces to its pure AdS shape.

\begin{figure} [H]
    \centering 
   \includegraphics[width=0.8\textwidth]{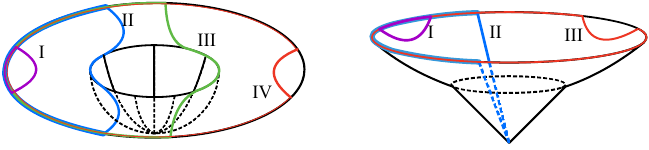} 
   \caption{The RT surfaces in the spherical (left) and hyperbolic (right) extremal cases. As in figure \ref{modIR}, the Cauchy slices of the spherical and hyperbolic extremal geometries are embedded respectively in higher-dimensional imaginary Euclidean and Minkowski backgrounds. 
   The left and right figure depict respectively the four RT surface phases in the spherical extremal case and the three phases in the hyperbolic extremal case, along with their corresponding boundary regions.}
    \label{CMI} 
\end{figure}

On the other hand, applying the same ansatz (\ref{ansatz}) to general hyperbolic cases where the curvature within the IR region decreases, we have
\begin{equation}\label{Hy}
g(r)=\begin{cases}
(l^2+r^2)/l^2, & \text{for } r<r_{IR},\\
(l_{AdS}^2+r^2)/l_{AdS}^2, & \text{for } r>r_{IR}.
\end{cases}
\end{equation}
The curvature inside the IR region $R\propto -1/l^2\to -\infty$ as the hyperbolic radius $l\to 0$ in the extremal limit. In this limit, the IR region tends to a ``light cone" in which any geodesic segment has vanishing length, and in the sense of holographic entanglement entropy, its boundary can be regarded as an EoW brane. Meanwhile, to minimize area, optimal RT surfaces plunge into the IR region immediately, so any geodesic intersecting the boundary of the IR region is orthogonal to it. The RT surfaces in this case exhibit three phases.

I. As in the spherical extremal case, for any small enough boundary region, its RT surface remains the same as in the original AdS spacetime, residing outside the IR region.

II. When the size of the boundary region $L$ increases to a critical value $L_c$, a phase transition occurs before the phase-I surface touches the boundary of the IR region. Beyond this point, the RT surface enters the IR region, leaving two UV segments orthogonal to the IR boundary.

III. When $L > 2\pi l_{S^1} - L_c$, the RT surface degenerates to its original form in vacuum AdS spacetime, off the boundary of the IR region. The corresponding entanglement wedge contains the IR region. 

\subsection{Conditional mutual information and long scale entanglement in IR modified geometries}

\noindent So far we have studied the entanglement structure of the modified geometries through entanglement entropy (RT surfaces). However, it is insufficient to fully capture the connection between radial geometry and boundary correlations across different length scales. To address this we further utilize the conditional mutual information, which, unlike the entanglement entropy, distinguishes entanglement at different length scales as a probe of entanglement structure. 

The conditional mutual information $I(A:B|E)$ is defined by 
\begin{equation}\label{cmi}
    I(A:B|E)=S_{AE}+S_{BE}-S_{ABE}-S_{E}.\end{equation} 
In our discussion, $E$ is chosen as the boundary interval between $A$ and $B$. In pure AdS, such CMI never vanishes. In the spherical extremal case, however, the RT surfaces of $AE$, $BE$, $ABE$ and $E$ coincide when and only when the length of $E$ is longer than $L_c$, and the length of $ABE$ is no longer than $\pi l_{S^1}$. In this regime the CMI $I(A:B|E)$ vanishes, which equivalently requires that the distance $l_{ab}$ between any two points $a\in A$ and $b\in B$ satisfies $L_c\leq l_{ab}<\pi l_{S^1}$. Such vanishing CMI signals the absence of long-range entanglement ($L_c < L < \pi l_{S^1}$) between boundary subregions $A$ and $B$ and the presence of the longest-range entanglement at $L = \pi l_{S^1}$.

\begin{figure} [H]
    \centering 
   \includegraphics[width=0.8\textwidth]{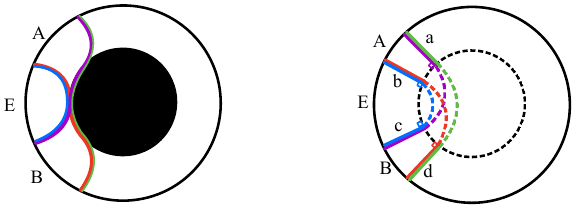} 
   \caption{The vanishing-CMI configurations for the spherical (left) and hyperbolic (right) extremal cases. Specifically, the figures plot the Cauchy slices of the spherical and hyperbolic extremal geometries in figure \ref{CMI} in stereographic projection. 
   The left figure represents the vanishing CMI for the spherical extremal case: the RT surfaces for $AE$ (purple), $BE$ (red), $E$ (blue), and $ABE$ (green) are shown, and their contributions cancel exactly in the CMI combination. An analogous cancellation occurs in the hyperbolic extremal case, as illustrated in the right figure.}
    \label{CMI2} 
\end{figure}

In the hyperbolic extremal case, the condition for vanishing CMI is looser. Specifically, as shown in the lower right of figure \ref{CMI2}, with the boundary length of $E$ no shorter than the critical length $L_c$, 
    \begin{equation}
    I(A:B|E)=S_{AE}+S_{BE}-S_{ABE}-S_{E}=a+c+b+d-a-d-b-c=0,
\end{equation} reflecting the absence of entanglement between $A$ and $B$. In conclusion, all quantum entanglement between boundary subregions with a distance longer than $L_c$ is eliminated for the hyperbolic extremal case.

To analyze the behavior of more refined bipartite entanglement, in   \cite{Ju_2024} we further provide a detailed calculation of the CMI between two infinitesimal subregions separated by different distances on the boundary. Eventually, as illustrated in figure \ref{IRES}, we obtain the following physical picture:

Modifying the IR geometry can be understood as a ``redistribution of the entanglement structure across different length scales". In the spherical extremal case, long-scale (longer than the critical length) entanglement is transferred to the longest-scale, whereas in the hyperbolic extremal case, it is transferred to the critical-scale (the shortest scale subject to modification). It can be concluded that such qualitative difference in entanglement structures originates from the diametrically opposed geometric properties of the spherical and hyperbolic extremal case.

\begin{figure}[H]
    \centering \includegraphics[width=0.9\textwidth]{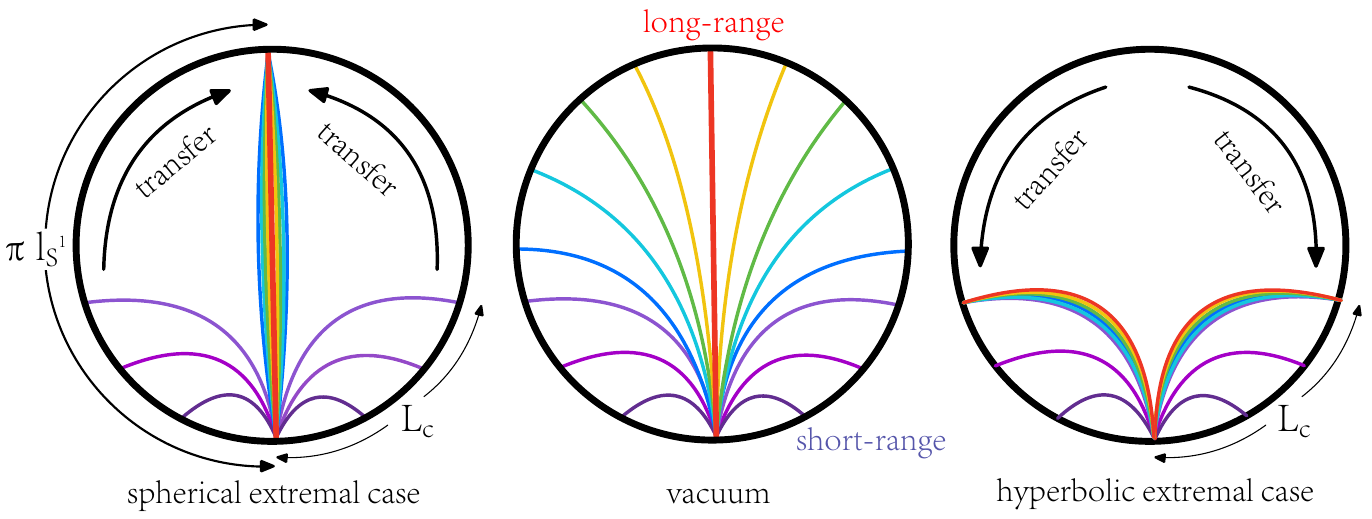} 
    \caption{Entanglement structures are depicted using threads representing the entanglement (CMI between two infinitesimal subregions) between the two points they connect. The middle figure displays the entanglement in vacuum AdS, with entanglement at all length scales. On the left the spherical extremal case is shown with all $L > L_c$ long-scale entanglement eliminated and transferred to the longest scale $L = \pi l_{S^1}$. On the right figure, in the hyperbolic extremal case, all $L > L_c$ long-scale entanglement is eliminated and transferred to the critical length $L = L_c$.}\label{IRES} 
\end{figure}

\subsection{Holographic entropy inequalities in IR modified geometries}

\noindent In addition to CMI, another perspective on the relationship between entanglement and geometry arises from holographic inequalities associated with entanglement entropy, EWCS and multi EWCS   \cite{Bao:2015bfa}. We have the following theorem: 

\begin{theorem}
Holographic inequalities remain valid in the spherical and hyperbolic extremal cases.
\end{theorem}
\begin{proof}
We first introduce the notion of non-extremal geometries. These are geometries with similar IR modifications as our extremal cases, but the curvature deformation does not reach the extremal limit. Since the proofs of these inequalities do not rely on specific geometric conditions, their validity also extends to non-extremal geometries. It should be noted that both types of extremal geometries can be regarded as continuously transformed from non-extremal ones. According to the intermediate value theorem, if an inequality is violated in an extremal geometry by a finite amount $\epsilon$, there must exist a non-extremal geometry, in which the inequality is violated by some finite $\epsilon'$, with $|\epsilon'|< |\epsilon|$. This contradicts with the general validity of these inequalities in non-extremal geometries. Therefore, all holographic inequalities should remain valid in the spherical and hyperbolic extremal cases.
\end{proof}

To conclude, in Section \ref{sec2} we have set our playground of the spherical and hyperbolic extremal cases and discussed mainly how the modifications of IR geometry reshape their bipartite entanglement. It should also be noted that, despite the proposal of the two toy models,   \cite{Ju_2024} provided additional insights into the entanglement structures related to differential entropy   \cite{Balasubramanian:2013lsa} and brane-world holography   \cite{Fujita:2011fp,Takayanagi:2011zk}. For the spherical extremal case, integrating the divergent 2-point CMI between infinitesimal subregions over the boundary yields the area of the modified region. This shows that the carrying capacity of total entanglement for the IR region is determined by its surface area, strengthening the connection between entanglement and geometry. For the hyperbolic extremal case, the shape of the IR region can be generalized so that its edge extends to the boundary, coinciding with an EoW brane in brane-world holography. The elimination of long-scale entanglement reveals the fine entanglement structure that the degrees of freedom extremely close to the boundary of the BCFT can never be entangled with other degrees of freedom at any finite distance. 

In the following sections, our present work further establishes the connection between multipartite entanglement and geometry. We employ a broader set of entanglement measures, including EWCS, multi EWCS, and multi-entropy, to probe more refined classes of boundary entanglement structure beyond the bipartite discussion in   \cite{Ju_2024}. Moreover, new insights about these entanglement measures are gained during such broader multipartite discussion.

\section{Entanglement measures from EWCS for modified IR geometry}\label{sec3}

\noindent {In this section, we investigate the change of the multipartite entanglement structure in the two types of IR modified geometries introduced in the previous section, utilizing the EWCS related multipartite entanglement measures. As explained in the previous section, the multipartite entanglement is expected to increase at longer distance scales in the {spherically IR modified geometries}, while decrease in the hyperbolic case. We begin with a brief review of EWCS and its  boundary dual quantities. Then we analyze how EWCS and the associated multipartite entanglement measures behave in the two types of IR modified geometries.}

{The entanglement wedge cross section (EWCS) is an important geometric object in holography, first proposed in   \cite{Umemoto_2018}. The EWCS of two boundary subregions $A$ and $B$, denoted by $E_W(A:B)$, is defined to be proportional to the minimal cross section that divides $A$ and $B$ in the bulk entanglement wedge $M_{AB}$. In the following discussion, we denote such a minimal cross section by $\gamma'_{A,B}$, and the corresponding $E_W(A:B)$ is defined as  }
\begin{equation}
    E_W(A:B)\equiv \frac{\text{Area}(\gamma'_{A,B})}{4G_N}.
\end{equation}
The EWCS is conjectured to be the holographic dual of {two} information-theoretical quantities: the entanglement of purification  $E_P(A:B)$   \cite{Umemoto_2018} and the reflected entropy $S_R(A:B)$   \cite{dutta2019}. 

{Several measures of multipartite entanglement have been constructed utilizing the entanglement of purification or the reflected entropy, giving rise to EWCS-related multipartite entanglement measures in holography. }
 For instance, the definitions of the Markov gap   \cite{Hayden_2021} and the L-entropy   \cite{Basak2025} rely on the reflected entropy. For a tri-partitioned pure state, the Markov gap $h_{AB}$ is defined as 
\begin{equation}
    h_{AB} \equiv S_R(A:B)-I(A:B).
\end{equation}
Meanwhile, the L-entropy $l_{ABC}$ is defined by
\begin{equation}
    l_{ABC} \equiv [l_{AB}l_{BC}l_{AC}]^{1/3},
\end{equation}
where the two party L-entropy $l_{AB} \equiv 2\min\{S_A, S_B\} - S_R(A:B)$, and $l_{BC}$ and $l_{AC}$ are defined analogously. {The L-entropy satisfies the requirements of a genuine multipartite entanglement (GME) measure   \cite{Xie2021,Ma_2024} and is therefore expected to count the genuine multipartite entanglement among $A$, $B$, and $C$. In particular, a separable state has vanishing L-entropy and for the tripartite GHZ state the L-entropy attains its maximal value in three qubit systems. In principle, the definition of the tripartite L-entropy could be generalized to the multipartite systems by taking the geometric mean of all possible two party L-entropies   \cite{Basak2025}. This construction satisfies the requirements of a GME measure only for systems with up to $n=5$ parties. A refined version of generalized L-entropy was proposed in   \cite{ahn2025}, which satisfies the criteria for a valid GME measure for all $n \geq 3$.}

In analogy with the definition of the L-entropy, one can also define the generalized tripartite Markov gap {with a permutation symmetry in $A, B$ and $C$} as
\begin{equation}
    h_{ABC} \equiv [h_{AB}h_{BC}h_{AC}]^{1/3}.
\end{equation}
As discussed in   \cite{iizuka2025}, instead of using the geometric mean of $h_{AB}$, $h_{AC}$, and $h_{BC}$, the Markov gap can be extended to the multipartite case by constructing a suitable linear combination of the reflected multi-entropy   \cite{Yuan_2025} and the relevant entanglement entropies.

Moreover, the quantity $g(A:B) = 2E_P(A:B) - I(A:B)$ was also defined in   \cite{Zou_2021}, and it was proven that the necessary and sufficient condition for $g(A:B) = 0$ is that the state $\ket{\psi}_{ABC}$ can be written as a \textit{triangle state} up to a local unitary transformation. By triangle state we mean that $\ket{\psi}_{ABC}$ can be decomposed as \footnote{{A triangle state $\ket{\psi}_{ABC}$ lacks non-trivial tripartite entanglement if we divide the subsystems into smaller subsystems. However, when the subsystems in question are merely $A$, $B$, and $C$, according to the definition of genuine multipartite entanglement, $A$, $B$, and $C$ still shares genuine multipartite entanglement as $\ket{\psi}_{ABC}$ is not separable unless factorizing $A$, $B$, and $C$ are factorized into smaller subsystems.}}
\begin{equation}\label{triangle state}
    \ket{\psi}_{ABC} = \ket{\psi}_{A_LB_R}\ket{\psi}_{B_LC_R}\ket{\psi}_{A_RC_L}
\end{equation}
for some appropriate bipartition $\mathcal{H}_{\alpha}=\mathcal{H}_{\alpha_L}\otimes\mathcal{H}_{\alpha_R}$ ($\alpha=A,B,C$)   \cite{Zou_2021}.  In Section \ref{quadrangle state}, we will generalize this conclusion to the multipartite case. It should be emphasized that $h_{AB}$ and $g(A:B)$ are not exactly equivalent in general quantum systems although they have the same holographic dual, namely $2E_W(A:B)-I(A:B)$. 

For a general quantum state, $g(A:B)=0$ is strictly stronger than $h_{AB}=0$ since $g(A:B)\geq h_{AB}\geq=0$. The latter merely implies that $\ket{\psi}_{ABC}$ can be written as a sum of triangle states (SOTS)   \cite{Zou_2021} up to local unitary transformations, which takes the form
\begin{equation}\label{SOTS}
\ket{\psi}_{ABC}=\sum_{j}\sqrt{p_j}\ket{\psi_j}_{A^j_LB^j_R}\ket{\psi_j}_{B^j_LC^j_R}\ket{\psi_j}_{A^j_RC^j_L},
\end{equation}
where $\sum_jp_j=1$. For instance, the tripartite GHZ state has vanishing $h_{AB}$ but non-vanishing $g(A:B)$, and is therefore a SOTS but not a triangle state. Conversely, the condition $h_{AB}\neq 0$ is stronger than the condition $g(A:B)\neq 0$. A nonvanishing $h_{AB}$ indicates the presence of tripartite entanglement structures that are neither triangle states nor sums of triangle states. Nevertheless, since both quantities have the same holographic dual $2E_W(A:B)-I(A:B)$, it is therefore consistent in the holographic context to always adopt the stronger conditions in the two cases $2E_W(A:B)-I(A:B)=0$ and $2E_W(A:B)-I(A:B)\neq 0$.

In this work, we especially focus on the behavior of the two EWCS related measures: the Markov gap   \cite{Hayden_2021} (or equivalently $g(A:B)$  \cite{Zou_2021}) and the L-entropy   \cite{Basak2025} in IR modified geometries. This helps us learn more about multipartite entanglement structures in holography, and in particular about the relation between multipartite entanglement structures and the bulk geometry. In Section \ref{EWCS circular}, we examine the behavior of EWCS-related measures when the IR region is of a circular shape. However, it should be emphasized that the IR region is not restricted to being circular and can, in principle, take an arbitrary shape since gluing two manifolds together only requires that their induced metrics match on the boundary under null energy conditions. In the subsequent sections \ref{EWCS upper bound}  and \ref{section EWCShyperbolic}, {we mainly focus on the maximally possible modifications of the geometry that keep the entanglement wedges of certain specified boundary subregions invariant, resulting in a non-circular IR region}, and we investigate the extremal values that the EWCS can attain under such modifications. In particular, we find that, in the modified hyperbolic geometry, $g(A:B)$ of the boundary quantum state $\ket{\psi}_{ABC}$ can be made to vanish, rendering the boundary state a triangle state.

\subsection{EWCS for geometries with circular shaped modified IR regions}
\label{EWCS circular}

\noindent In this subsection, we examine how the EWCS and EWCS-related measures change in the two types of IR modified geometries with a circular IR region, under the same setup as in Section \ref{sec2}. We analyze the behavior of {$E_W(A:C)$—with the boundary pure state being partitioned into three equal-size subregions $A, B, C$}—in the extremal cases of the modified spherical and hyperbolic geometries, and compare it with that in pure AdS. Furthermore, we investigate how two EWCS-related measures, the Markov gap and the L-entropy, respond to such geometric deformations.
\begin{figure}[h] 
\centering 
\includegraphics[width=0.95\textwidth]{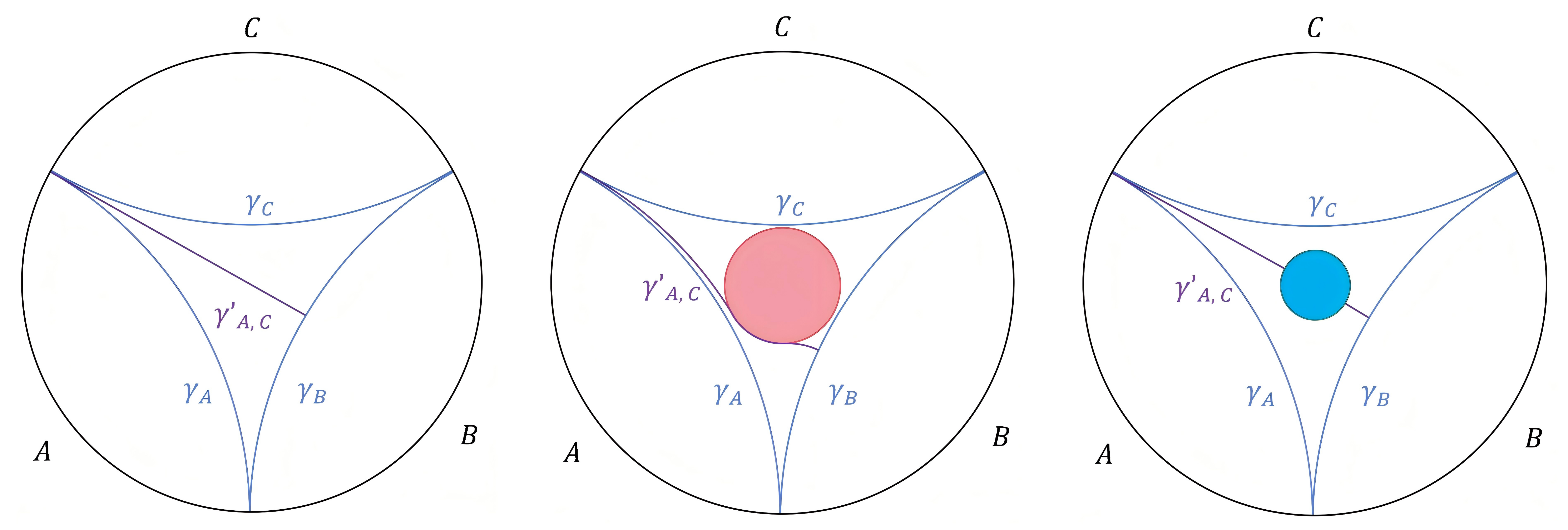} 
\caption{{The calculation of }EWCS in the pure AdS and the modified geometries. The entire conformal boundary is partitioned into three equal subregions $A$, $B$, and $C$. In the left, middle, and right figures, the purple curves represent $\gamma'_{A,C}$ in the pure AdS, the modified spherical geometry, and the modified hyperbolic geometry, respectively. The three blue curves denote the minimal surfaces homologous to $A$, $B$, and $C$. The modified IR region is sufficiently small so that the entanglement wedges of $A$, $B$, and $C$ remain unchanged after the geometric modification.} 
\label{circular IR region} 
\end{figure}

As shown in figure \ref{circular IR region}, the three subregions $A$, $B$, and $C$ are equal in size. We keep the modified IR region sufficiently small so that the entanglement wedges of $A$, $B$, and $C$ remain unchanged after the modification, thereby preserving the local short-range entanglement structure {within each subregion}. In the {first type of IR modified geometry, i.e.} the extremal case of the IR spherically modified  geometry, $\gamma'_{A,C}$ is pushed out to the boundary of the infrared region and thus acquires a larger length. In contrast, in the second type of IR modified geometry, i.e. the extremal case of the hyperbolic IR modified  geometry, the IR region tends to a ``light cone" so that {the length of the geodesic inside the IR region approaches 0}, and $\gamma'_{A,C}$ becomes shorter than in the pure AdS geometry. {Therefore, $E_W(A:C)$ increases in the spherically IR modified geometry while decreases in the hyperbolic case, compared to the original pure AdS result.}

From the definitions of the Markov gap and the L-entropy, it follows that $h_{ABC}$, which is equivalent to $h_{AB}$ since $A$, $B$, and $C$ have equal sizes, increases in the modified spherical geometry and decreases in the hyperbolic case, while $l_{ABC}$ behaves oppositely, {due to the opposite signs of the EWCS term in the two measures and the unchanged entanglement entropies of $A$, $B$, and $C$ before and after the modification.} This strikingly different behavior suggests that they detect different types of entanglement. {Since the Markov gap vanishes for triangle states and SOTS{\footnote{{Since entanglement structures should be invariant under local unitary (LU) transformations of each subsystem, entanglement measures are defined to be LU invariant. Accordingly, when we refer to triangle states and SOTS, we mean states that can be written in the forms of equations (\ref{triangle state}) and (\ref{SOTS}), respectively, up to local unitary transformations of each subregion.}}}, we expect the Markov gap to detect certain non-SOTS-type entanglement.  On the other hand, the L-entropy could be nonzero for triangle states and SOTS-type entanglement, e.g. the GHZ type entanglement, so we expect that L-entropy could at least detect SOTS-type entanglement. This implies that, in the spherical extremal case, the amount of non-SOTS  entanglement increases whereas the SOTS type or other types of entanglement that could be detected by the L-entropy decreases; in contrast, the situation is reversed in the hyperbolic case. } In Section \ref{L and M} we will make more detailed comparison between these two measures.

\subsection{EWCS for maximally IR modified geometries with certain entanglement wedges fixed}\label{section EWCSspherical}

\noindent {Unlike previous work   \cite{Ju_2024} and the preceding subsection, in this subsection, we consider a different shape for the IR modified regions—namely, the maximal extent to which the IR region can be deformed while keeping the density matrices of $A$, $B$, and $C$ unchanged. In general, the shape of the IR region is no longer circular.} One motivation for considering this maximal modification of the IR region is that it allows us to alter the multipartite entanglement structure among the subsystems as much as possible, while preserving the entanglement structure within each individual subsystem. Moreover, this is also related to the quantum marginal problem   \cite{schilling2014,schilling2015}.

The quantum marginal problem concerns the reconstruction of the full density matrix of a quantum system given the reduced density matrices of its subsystems. It is the problem of determining the mathematical conditions on density matrices belonging to different subsystems of interest ensuring that all are belonging via partial trace to the same quantum state of the total system. {The quantum marginal problem has important applications in quantum information theory. For example,} in quantum information theory, many entanglement measures are defined as the extrema of specific combinations of entanglement entropies. Representative examples include the squashed entanglement   \cite{Christandl2004}, the conditional entanglement of mutual information (CEMI)   \cite{Yang2008}, and the entanglement of purification introduced earlier. These quantities mostly exhibit desirable properties—such as convexity and faithfulness—yet they are typically difficult to evaluate in practice, and this difficulty is closely connected to the solution of the quantum marginal problem in special cases.

Although the quantum marginal problem is well-defined, it is in general highly nontrivial to analyze. In holography, this problem can be reformulated within the framework of modifying IR geometry   \cite{Ju_2024,Ju:2025mvz}. Fixing the reduced density matrices while deforming the bulk geometry outside these entanglement wedges is equivalent to constructing various possible density matrices of the full system. Our focus is to study the extremal features of these constructed holographic density matrices under such deformations. These extremal properties serve as constraints bounding the full-system density matrix, indicating that any holographic physically realizable global state must satisfy these bounds.

To ensure that the density matrices of certain boundary subsystems remain unchanged after modifying the IR geometry, the IR region whose geometry could be modified can only  stay outside the entanglement wedges of these boundary subsystems. In the spherical extremal case, the modified IR region could be the whole region outside those entanglement wedges so that the geometric structure inside the entanglement wedges of these boundary subsystems remain unchanged, thereby preserving their reduced density matrices. However, it should be noted that geodesics in the hyperbolic IR modified geometry are shorter than those in the pure AdS. As a result, even if the geometric deformation is restricted to regions outside the entanglement wedges, the corresponding minimal surfaces may still be altered. Therefore, the above arguments no longer hold in the case of the modified hyperbolic geometry and the IR region that could be modified is smaller than the whole region outside those entanglement wedges of the boundary subsystems. In Section \ref{section EWCShyperbolic}, we introduce the alternative geometric construction that preserves the entanglement wedges of given boundary subsystems in the hyperbolic case.

\subsubsection{The spherical case: EWCS upper bound for the quantum marginal problem}\label{EWCS upper bound}
\noindent We start from the spherically IR modified geometries. We will first consider the case in which the entanglement wedges of three boundary subregions are preserved. Subsequently, we turn to the configuration with five boundary subregions, requiring that the entanglement wedges of any two adjacent subregions remain unchanged. In the extremal case of the modified spherical geometry, as we will show, the computation of the EWCS reduces to solving a well-defined problem in quantum information theory and holography, namely, determining the upper bound of $S_R$ or $E_P$ under the constraint that certain reduced density matrices are fixed.

\subsubsection*{Configuration with fixed entanglement wedges for three boundary subregions}

\begin{figure}[h] 
\centering 
\includegraphics[width=0.95\textwidth]{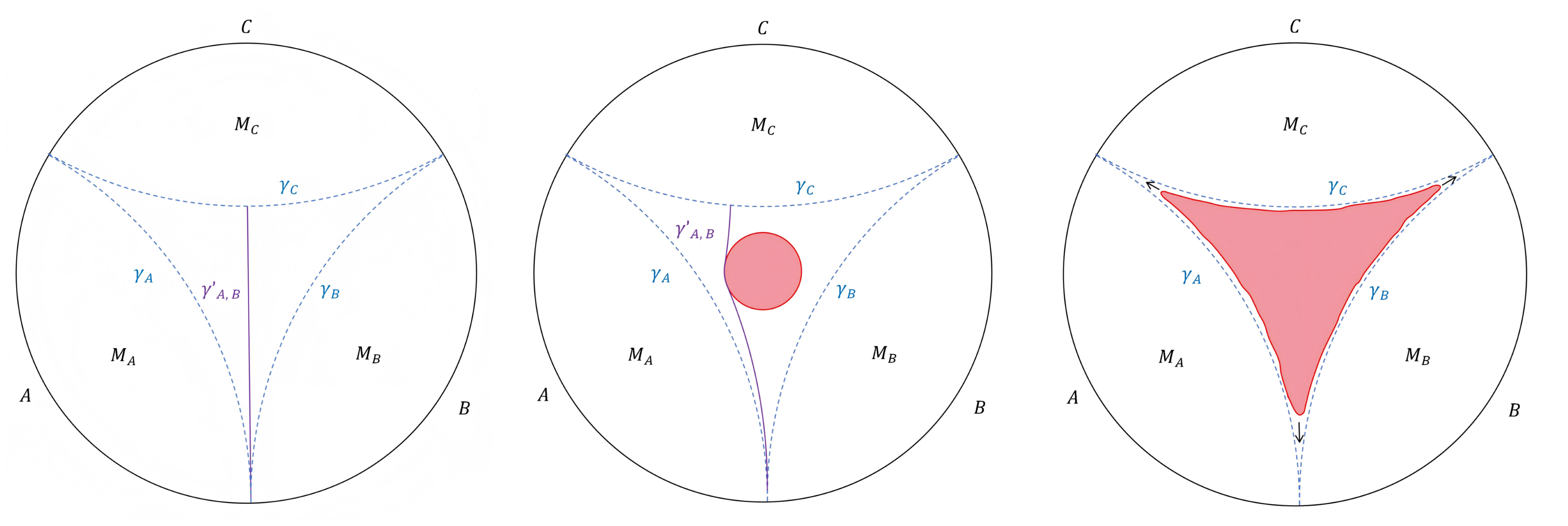}
\caption{$\gamma\prime_{A,B}$ in pure AdS and the spherically IR modified geometry. The three blue dashed lines are the three minimal surfaces homologous to $A$, $B$ and $C$.  We restrict the geometric modification to the exterior of these three minimal surfaces. When the modified IR region has not filled the entire outside region (middle figure), $\gamma\prime_{A,B}$ is pushed outward, and therefore takes a larger value compared to the pure AdS (left figure). As the spherical geometric region gradually fills the whole triangle region enclosed by the three minimal surfaces (right figure), $\gamma\prime_{A,B}$ eventually becomes the shorter one of $\gamma_A$ and $\gamma_B$.} 
\label{EWCS_in_Spherical} 
\end{figure}

\noindent We divide the conformal boundary into three adjacent subregions $A$, $B$, and $C$, whose entanglement wedges $M_A$, $M_B$, and $M_C$ are the regions in the bulk enclosed by the respective minimal surfaces {$\gamma_A$, $\gamma_B$, and $\gamma_C$}, indicated as blue dashed lines in figure \ref{EWCS_in_Spherical}. We modify the geometry only outside $M_A$, $M_B$, and $M_C$, which ensures that the reduced density matrices of $A$, $B$, and $C$ remain unchanged after the bulk geometry modification. Note that although we work in global AdS$_3$ for concreteness, our analysis can in principle be generalized to more general settings. 

$E_W(A:B)$ is proportional to the shortest geodesic $\gamma\prime_{A,B}$ from the boundary point between $A$ and $B$ to $\gamma_C$.  In the extremal case of the modified spherical geometry, since no geodesic penetrates the IR region, the geodesic $\gamma\prime_{A,B}$ is the curve along the boundary of the IR region, as shown in the middle figure of figure \ref{EWCS_in_Spherical}. It is obvious that, after the geometric modification, the length of $\gamma\prime_{A,B}$ becomes larger than its original value in the AdS vacuum.

As mentioned earlier, the IR region need not be restricted to a circular shape; it can take an arbitrary shape. In particular, we can allow the IR region to gradually fill the triangular region outside the three entanglement wedges $M_A$, $M_B$, and $M_C$. As the modified IR region expands toward the three minimal surfaces $\gamma_A$, $\gamma_B$, and $\gamma_C$, the length of $\gamma\prime_{A,B}$ increases gradually throughout this process. When the boundary of the modified IR region coincides with the minimal surfaces, $\gamma\prime_{A,B}$ reduces to the shorter of $\gamma_A$ and $\gamma_B$, at which point $E_W(A:B)$ attains its upper bound at fixed $\rho_A$ and $\rho_B$, and saturates the inequality 
\begin{equation}
   E_W(A:B) \leq \min \{S_A, S_B\}.   
\end{equation}
In holography, this inequality has a clear geometric interpretation: since both $\gamma_A$ and $\gamma_B$ are candidates of the minimal cross section, their lengths must be longer or equal to that of $\gamma'_{A,B}$. Furthermore, we emphasize that this also corresponds to the information-theoretic upper bound of $E_P$ and $S_R/2$ when $\rho_A$ and $\rho_B$ are fixed. This example shows that, modifying IR geometries could lead to boundary states where certain entanglement measures could saturate their information-theoretic bounds, giving rise to states with highly nontrivial and exotic entanglement structures.

\subsubsection*{Configuration with fixed entanglement wedges for overlapping boundary subregions}

\noindent In the calculation above, there is a subtle issue concerning the modified IR region: when we fix the entanglement wedges of non-overlapping boundary subregions, the maximal IR region that can be deformed typically extends all the way to the boundary, thereby inducing geometric changes in the asymptotic boundary area. From the perspective of the quantum marginal problem, if we fix only the reduced density matrices of several non-overlapping boundary subregions, the problem becomes relatively trivial and simple. In such cases, quantities like EWCS often attain their information-theoretic bounds, providing no additional constraints on the full density matrix. The situation changes when we fix the entanglement wedges of overlapping boundary subregions. We illustrate this with a concrete example below.

\begin{figure}[h] 
\centering 
\includegraphics[width=0.95\textwidth]{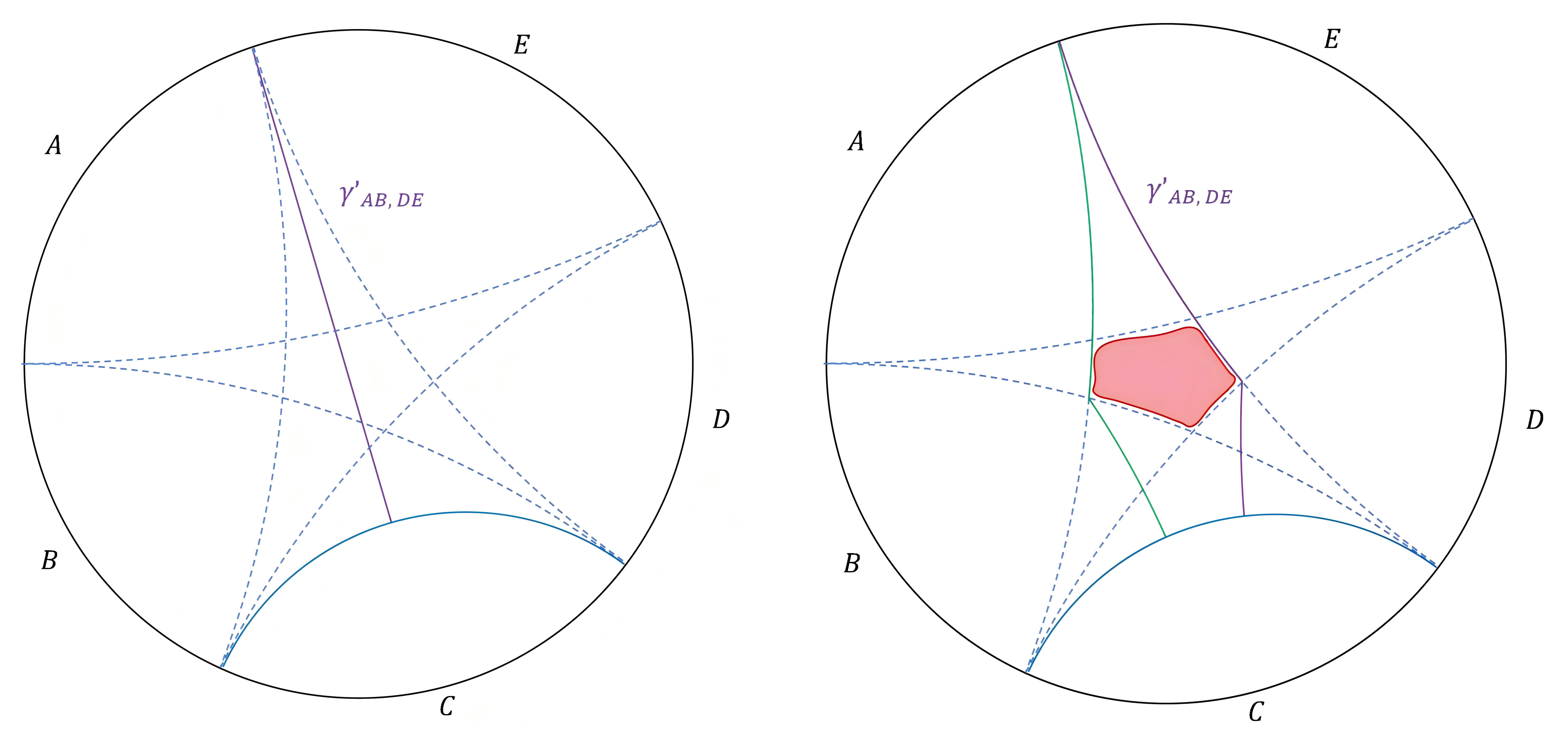} 
\caption{$\gamma\prime_{AB,DE}$ in the pure AdS geometry and the spherically IR modified geometry. We use blue dashed lines to represent the five minimal surfaces homologous to $AB$, $BC$, $CD$, $DE$, and $AE$, and restrict the geometric modification to the pentagonal region enclosed by these minimal surfaces. The blue solid line denotes the minimal surface homologous to $C$. In the left panel, the purple curve represents $\gamma'_{AB,DE}$ in the pure AdS. In the right panel, as the spherical modified IR region gradually fills the pentagonal area enclosed by the minimal surfaces, $\gamma'{AB,DE}$ can no longer penetrate this region and therefore corresponds to the shorter of the two curves — the purple and the green ones. } 
\label{pentagon} 
\end{figure}

As shown in the figure \ref{pentagon}, we divide the whole boundary region into five subregions $A$, $B$, $C$, $D$, and $E$, with the minimal surfaces homologous to $AB$, $BC$, $CD$, $DE$, and $AE$ indicated by blue dashed lines. The subregions whose entanglement wedges are fixed are chosen as: $AB$, $BC$, $CD$, $DE$, and $AE$, while the EWCS that we are going to calculate is $E_W(AB:DE)$. Therefore, we modify the geometry only outside the entanglement wedges of $AB$, $BC$, $CD$, $DE$, and $AE$, namely, within the pentagon in the bulk center enclosed by these five minimal surfaces. Note that in this circumstance, at least five subregions have to be considered, otherwise, there will be no IR region that could be constructed. This choice ensures that the density matrices $\rho_{AB}$, $\rho_{BC}$, $\rho_{CD}$, $\rho_{DE}$, and $\rho_{AE}$ remain unchanged, and the resulting region is purely infrared. {In the spherical extremal case, the EWCS will, in general, no longer coincide trivially with its information-theoretic bounds, and thus we obtain new holographic constraints on the full density matrix.}

The calculation of  $E_W(AB:DE)$ is shown in the right figure in figure \ref{pentagon}. When the boundary of the modified infrared region expands and gradually approaches the five minimal surfaces, $\gamma_{AB,DE}$ is pushed to the boundary of the pentagon, at which point the EWCS attains its extremal value while $\rho_{AB}$, $\rho_{BC}$, $\rho_{CD}$, $\rho_{DE}$, and $\rho_{AE}$ are kept unchanged. We calculate the change of $E_W(AB:DE)$ in this configuration after the IR geometry is 
deformed from the pure AdS into the spherical geometry. {It is convenient to perform this calculation by transforming to planar coordinates. We use $L_A$, $L_B$, $L_C$, $L_D$, and $L_E$ to denote the lengths of the subregions $A$, $B$, $C$, $D$, and $E$ in the planar coordinate and the result is given by}
\begin{equation}\label{simplified result}
\begin{aligned}
\Delta E_W(AB:DE) = \frac{1}{4G_N}\min\Biggl\{&
 \log\!\left[\tfrac{1}{2}\sqrt{\tfrac{L_C(2L_D+L_E)+2(L_D(L_D+L_E)+\sqrt{L_D(L_C+L_D)(L_D+L_E)(L_C+L_D+L_E)})}{L_D(L_C+L_D+L_E)}}\right], \\
& \log\!\left[\tfrac{1}{2}\sqrt{\tfrac{L_C(2L_B+L_A)+2(L_B(L_B+L_A)+\sqrt{L_B(L_C+L_B)(L_B+L_A)(L_C+L_B+L_A)})}{L_B(L_C+L_B+L_A)}}\right]
\Biggr\}.
\end{aligned}
\end{equation}
 The detailed derivation is given in Appendix \ref{A}. Note that $\Delta E_W(AB:DE)$ is an IR quantity with no UV divergence. {As in the spherical extremal case, no geodesics enter the IR region, and the EWCS in this geometry therefore attains its maximal value among all possible IR deformations.} {We thus obtain an upper bound on $E_W(AB:DE)$ in holography, namely
 \begin{equation}
    E_W(AB:DE)\leq \Delta E_W(AB:DE)+E_W(AB:DE)_{\text{pure AdS}},  
 \end{equation}
 which must be satisfied by any full density matrix constructed within this framework. It should be emphasized, however, that this is a holographic upper bound and might not apply to general quantum states. This construction can also be generalized to cases involving more subregions.}

\subsubsection{The hyperbolic case: minimization of EWCS}\label{section EWCShyperbolic}

\noindent In this section, we are going to analyze the lower bound of EWCS under the constraints that certain reduced density matrices are held fixed. Since entanglement measures typically increase in the spherical extremal case, we instead focus on the hyperbolic case. In this type of IR modified geometries, the EWCS under the constraints of unchanged density matrix of boundary subsystems would be smaller than the pure AdS results. Among all possible hyperbolic IR modifications, we could obtain the smallest value of EWCS under the same constraints and this could be conjectured to be the lower bound of EWCS for this quantum marginal problem in holography. 

Similar to the spherical extremal case discussed in Section \ref{EWCS upper bound}, we now replace an IR region with a hyperbolic geometry that ensures the entanglement wedges of certain boundary subregions remain unchanged after the geometric modification. However, in the hyperbolic case, these IR regions are no longer simply the regions outside a set of entanglement wedges. Instead, the IR region that could be modified must be smaller, as the geodesics are pulled toward the center of the bulk by the hyperbolic IR deformation of the geometry. A suitable geometric construction can be employed to determine the IR region that meets this requirement.

\begin{figure}[h] 
\centering 
\includegraphics[width=0.5\textwidth]{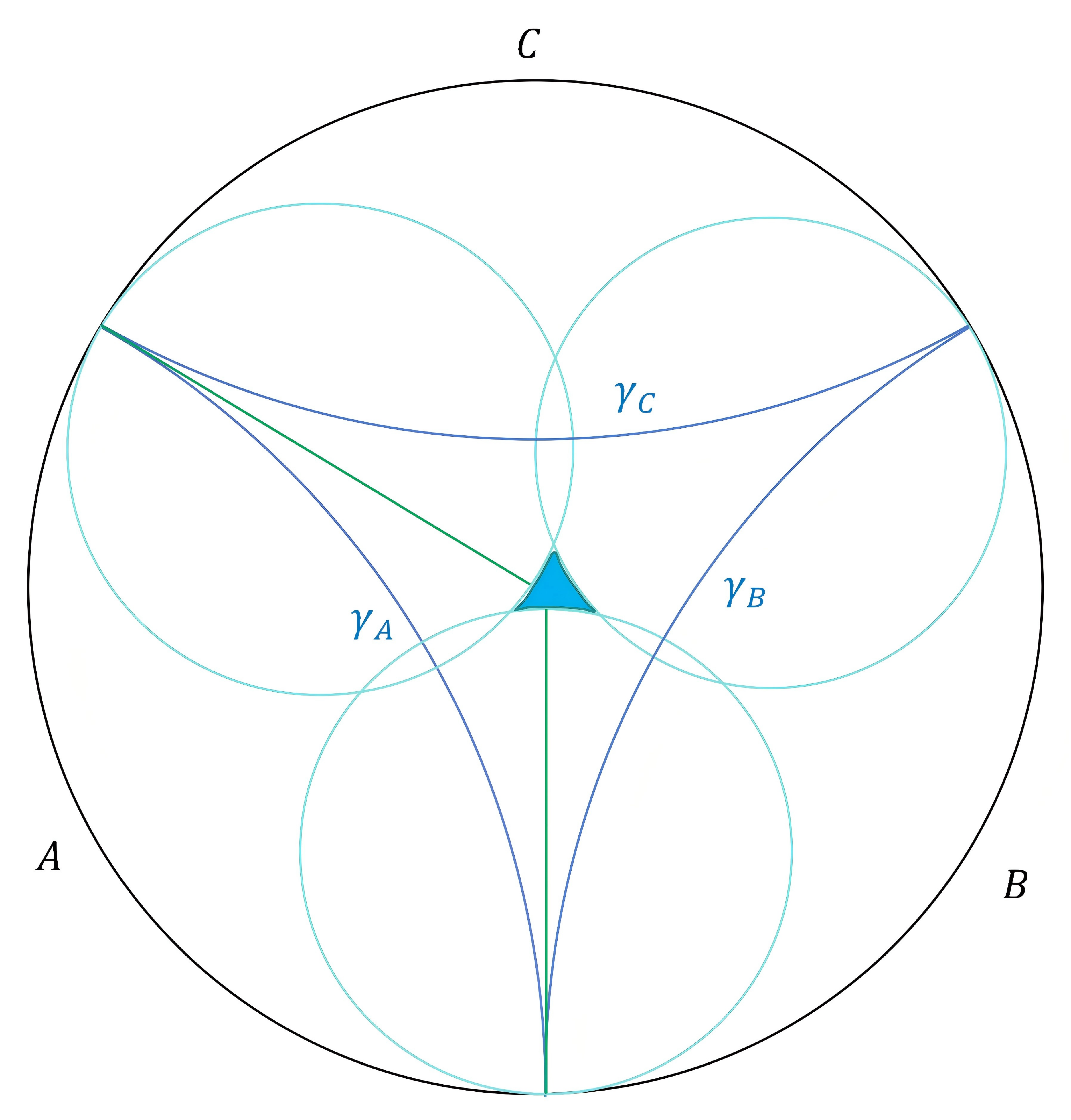} 
\caption{The hyperbolic IR region that preserves the entanglement wedges of $A$, $B$, and $C$. The three horospheres are tangent to {the boundary at} the points separating $A$, $B$, and $C$ and intersect pairwise. The IR region, shown in blue, corresponds to the region enclosed by these three horospheres.} 
\label{HyperbolicIRregion} 
\end{figure}
A horosphere is a circle tangent to the conformal boundary on a Cauchy surface in AdS$_3$. It possesses the property that all geodesics drawn from the tangency point to any other point on the horosphere have equal length. Horospheres are especially useful here as the lengths of the geodesics inside the IR extremal hyperbolic region are zero. As illustrated in figure \ref{HyperbolicIRregion}, we could take the IR region as the region enclosed by three horospheres that are tangent to the boundary at the boundaries of $A$, $B$, and $C$ and intersect pairwise: this configuration ensures that $\gamma_A$ does not intersect the IR region since its length is shorter than that of the two green geodesics, and the same holds for $\gamma_B$ and $\gamma_C$. Consequently, the {RT surfaces and the} entanglement wedges of $A$, $B$, and $C$ are preserved. This geometric construction can be naturally generalized to configurations with more fixed entanglement wedges by introducing additional horospheres that intersect, or are tangent to, their neighboring ones.

Unlike the spherical extremal case, the choice of defining an IR region for the hyperbolic extremal case is not unique because the radii of the horospheres could be adjusted. Therefore, the corresponding values of EWCS are not unique. We could vary the radii of these horospheres while ensuring that they remain pairwise intersecting or tangent, and identify the configuration that minimizes the EWCS. This is conjectured to be the minimum of EWCS in the quantum marginal problem in holography as the geodesic lengths in the IR region are the smallest in the extremal hyperbolic IR geometries among all possible geometries. The results show that, when the entanglement wedges of three subregions $A$, $B$, and $C$ are held fixed, $E_W(A:B)$, $E_W(A:C)$, and $E_W(B:C)$ attain their minimal values when the three horospheres are pairwise tangent. Similarly, when fixing the entanglement wedges of four subregions $A$, $B$, $C$, and $D$, we find that $E_W(A:B)$ also reaches its minimum when three of the four horospheres become pairwise tangent.

Moreover, when $E_W(A:B)$, $E_W(A:C)$, and $E_W(B:C)$ attain their minimal values under the constraint that the entanglement wedges of the three subregions $A$, $B$, and $C$ are held fixed, we find that $g(A:B)=g(B:C)=g(A:C)=0$. Consequently, the boundary CFT state corresponds to a triangle state defined in equation (\ref{triangle state}). This observation can be generalized to multipartite cases, suggesting that one can construct bulk geometries whose dual boundary states take the form of \textit{quadrangle states}, \textit{pentagon states}, and so on, whose definitions will be given later. We will discuss these constructions in detail in Section \ref{multi-MG} and Section \ref{quadrangle pentagon state}.  

\subsubsection*{Minimum of EWCS with fixed entanglement wedges of three boundary subregions}\label{EWCShyperbolic3}

\begin{figure}[h] 
\centering 
\includegraphics[width=0.95\textwidth]{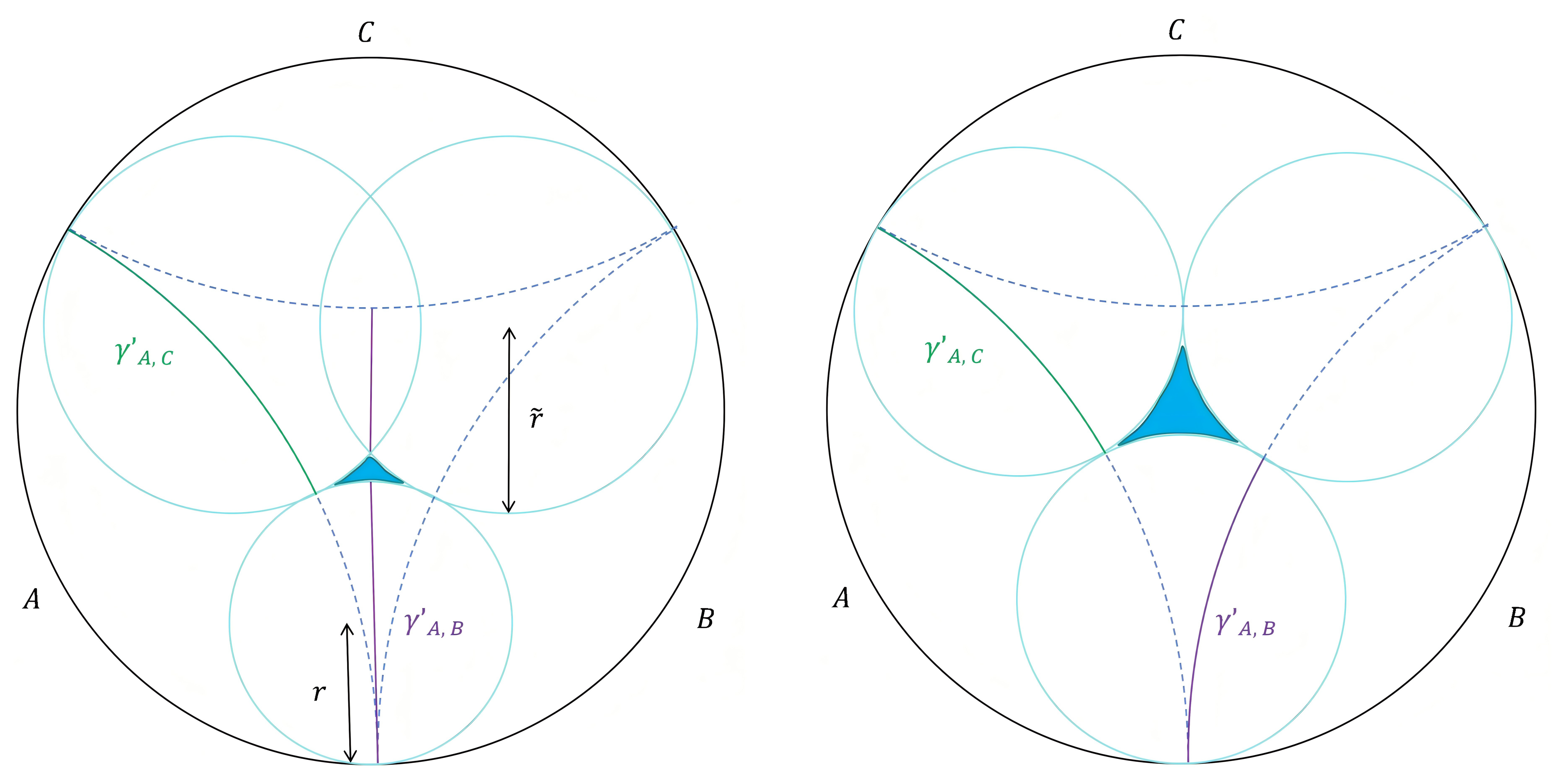} 
\caption{The hyperbolic IR modified geometry that leaves the entanglement wedges of $A$, $B$, and $C$ unchanged. Here $A$, $B$, and $C$ are chosen to be three boundary subregions of equal size. In this setup, the three horospheres must either intersect or be tangent to each other, and the region of the hyperbolic geometry is taken to be the IR region enclosed by these three horospheres. Purple curves represent $\gamma'_{A,B}$, while green curves represents $\gamma'_{A,C}$. When the two upper horospheres intersect, $\gamma'_{A,B}$ consists of two disconnected segments (left figure); when all three horospheres are pairwise tangent, $\gamma'_{A,B}$ becomes a single continuous curve (right figure).} 
\label{three_subregion_hyperbolic} 
\end{figure}
\noindent For the case of three boundary subregions, we only need to fix three points on the boundary to determine $A$, $B$, and $C$. However, due to conformal symmetry, these three points can always be mapped to three equally spaced points on the conformal boundary. For simplicity, we consider the case where the three subregions $A$, $B$, and $C$ are of equal size. As shown in figure \ref{three_subregion_hyperbolic}, we draw three horospheres that are tangent to the boundaries of subregions $A$, $B$, and $C$, respectively, such that each pair of horospheres is either tangent or intersecting to ensure the entanglement wedges of the subsystems stay the same. If the IR region outside these three horospheres is replaced with the extremal hyperbolic geometry, then the entanglement wedges of $A$, $B$, and $C$ remain unchanged after modifying the geometry.

Among the three horospheres, we denote by $r$ the radius of the smallest one. {To minimize the EWCS, the hyperbolic IR region should be taken as large as possible, which corresponds to the configuration in which the other two horospheres are tangent to the smallest one. We denote by $\tilde{r}$ the radius of these two larger horospheres. Since they are tangent to the smallest horospheres, $\tilde{r}$ is given by $\tilde{r}=(3-3r)/(r+3)$.} The parameter $r$ must be in the range from $0\leq r \leq 2\sqrt{3}-3$ for it to be the smallest radius. At $r=2\sqrt{3}-3$ the three horospheres have equal radii and are mutually tangent. The expression for $E_W(A:B)$ in terms of $r$ and $\tilde{r}(r)$ is calculated to be
\begin{equation}\label{results hyperbolic3}
    E_W(A:B) = \frac{1}{4G_N}\left[\log \frac{\sqrt{3} \left(1+\tilde{r}+\sqrt{\tilde{r}^2+6\tilde{r}-3}\right)}{3-\tilde{r}-\sqrt{\tilde{r}^2+6\tilde{r}-3}}\frac{4r}{(1-r)\epsilon}\right],
\end{equation}
where $\epsilon$ is the UV cutoff. The detailed derivation is presented in Appendix \ref{A2}.

\begin{figure}[h] 
\centering 
\includegraphics[width=0.7\textwidth]{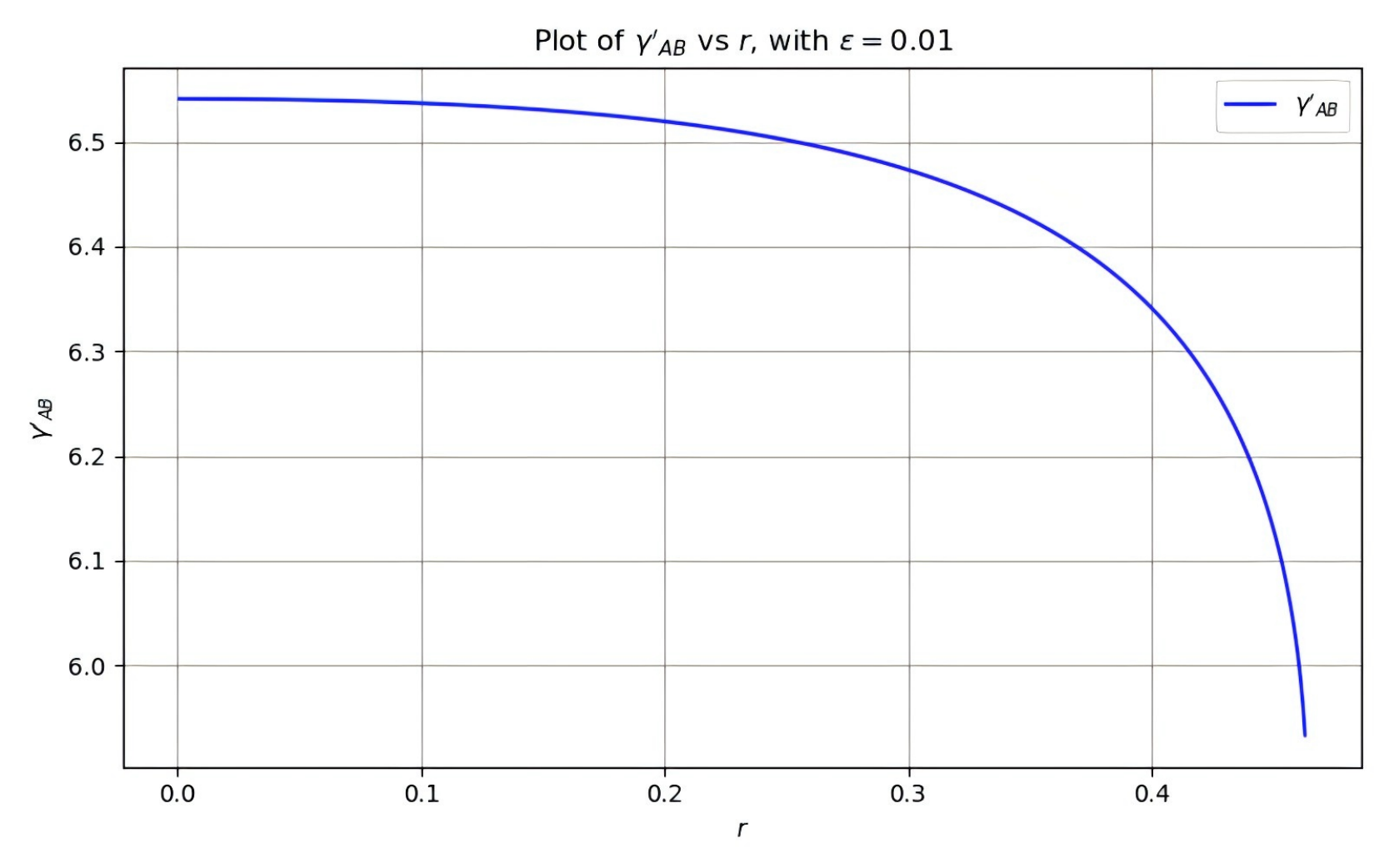} 
\caption{The plot of $\gamma'_{A,B}$ with respect to $r$, where we take the UV cutoff $\epsilon=0.01$. It can be observed that $\gamma'_{A,B}$ decreases monotonically with $r$, reaching its minimum when $r$ takes its maximum value, namely when the three horospheres are pairwise tangent.} 
\label{three_subregion_hyperbolic_gamma} 
\end{figure}

In figure \ref{three_subregion_hyperbolic_gamma} we plot the length of $\gamma'_{A,B}$ as a function of $r$. It is evident that $E_W(A:B)$ attains its minimum when $r$ reaches its maximum, namely when the three horospheres are pairwise tangent. It is worth noting that in this case, {because the three horospheres are tangent to each other,} we have
\begin{equation}
\begin{aligned}
    E_W(A:B)+E_W(A:C)&=S_A,\\
    E_W(A:B)+E_W(B:C)&=S_B,\\
    E_W(A:C)+E_W(B:C)&=S_C.\\
\end{aligned}
\end{equation}
Accordingly, each EWCS can be expressed in terms of the entanglement entropies
\begin{equation}
\begin{aligned}
    E_W(A:B)&=(S_A+S_B-S_C)/2,\\
    E_W(B:C)&=(S_B+S_C-S_A)/2,\\
    E_W(A:C)&=(S_A+S_C-S_B)/2.\\
\end{aligned}
\end{equation}
Thus we obtain
\begin{equation}
    g(A:B)=g(B:C)=g(A:C)=0.
\end{equation}
This implies that the boundary quantum state $\ket{\psi}_{ABC}$ is a triangle state    \cite{Zou_2021}, and therefore takes the form given in equation (\ref{triangle state}). {In this example,  $E_W(A:B)$, $E_W(A:C)$, and $E_W(B:C)$ attain their information-theoretic lower bounds through an appropriate modification of the IR geometry. This confirms the conjecture that, among all possible hyperbolic IR modifications, the minimal EWCS indeed realizes the lower bound under the constraints imposed by the quantum marginal problem. Moreover, this result suggests that the EWCS can saturate its information-theoretic lower bound in holography, and that a holographic construction of a triangle state is possible.}

\subsubsection*{Minimum of EWCS with fixed entanglement wedges of four boundary subregions}\label{EWCShyperbolic4}

\noindent
{In this section, we discuss the minimal value of $E_W(A:B)$ in the hyperbolic modified IR geometry while fixing the entanglement wedges of four subsystems $A$, $B$, $C$, and $D$. We will see that, when the boundary system is divided into more subregions and one considers the quantum marginal problem of minimizing the EWCS between two of them, the resulting minimal value may not coincide trivially with the quantum information-theoretic lower bound.} For the case of four fixed entanglement wedges, we need to fix four points on the boundary to determine subregions $A$, $B$, $C$ and $D$. Due to conformal symmetry, only one degree of freedom remains, as a conformal transformation can always be applied to make $A$ the same size as $C$, and $B$ the same size as $D$, and the total lengths of the four subregions are fixed. The central angle corresponding to interval $A$ and $C$ is {denoted as} $2\theta$. The relation between $\theta$ and the cross-ratio $X(A:C)$ of $A$ and $C$ can be calculated to be
 \begin{equation}
    X(A:C) = \frac{\sin^2\theta}{1-\sin^2{\theta}}.
\end{equation}
\begin{figure}[h] 
\centering 
\includegraphics[width=1\textwidth]{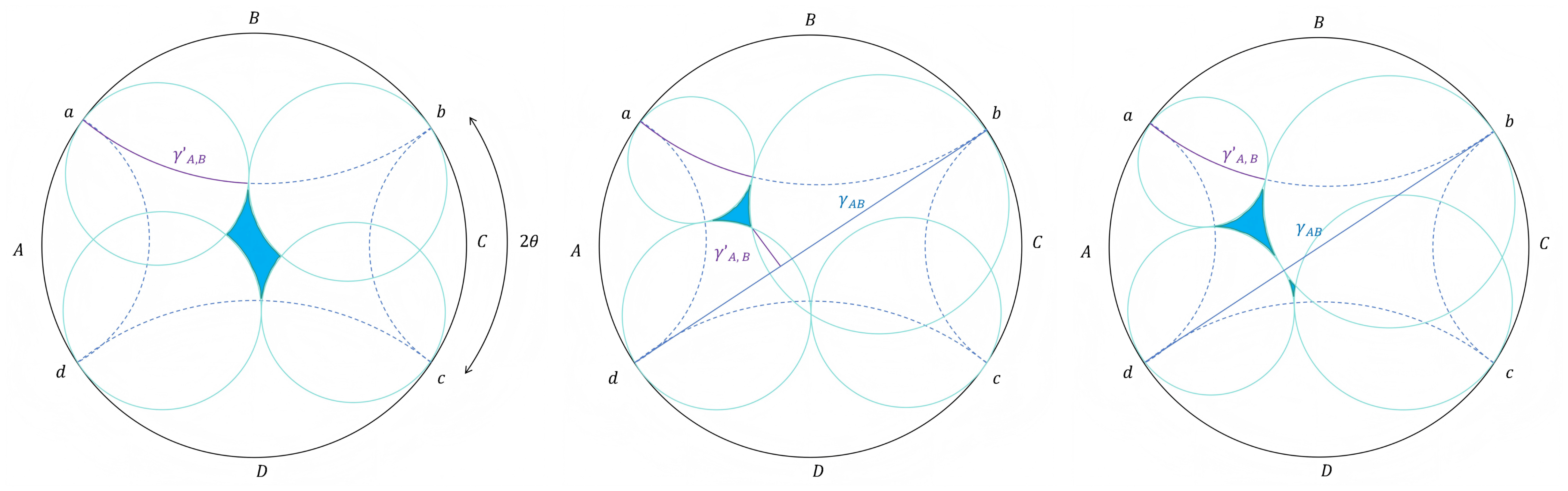} 
\caption{The hyperbolic modified geometry that preserves the entanglement wedges of $A$, $B$, $C$, and $D$. We use blue dashed lines to represent the four minimal surfaces homologous to $A$, $B$, $C$, and $D$. The radii of $\bigodot a$, $\bigodot b$, and $\bigodot c$ are $r$, $r'$ and $r''$, respectively. Any two adjacent horospheres are intersecting or tangent. If $\bigodot b$ intersects $\bigodot d$, then $\gamma'_{A,B}$ consists of 2 parts(middle figure). In the right figure $\bigodot a$, $\bigodot b$, and $\bigodot d$ are pairwise tangent.} 
\label{four_subregion_hyperbolic} 
\end{figure}
We denote by $\bigodot a$, $\bigodot b$, $\bigodot c$, and $\bigodot d$ the horospheres tangent at the points $a$, $b$, $c$, and $d$, with radii $r_a$, $r_b$, $r_c$, $r_d$, respectively. Since adjacent horospheres must intersect or be tangent, we have to demand that 
\begin{equation}
\begin{aligned}
r_b &\geq \frac{\cos^2 \theta - r_a \cos^2 \theta}{\cos^2 \theta + r_a \sin^2 \theta}, \\
r_d &\geq \frac{\sin^2 \theta - r_a \sin^2 \theta}{\sin^2 \theta + r_a \cos^2 \theta},
\end{aligned}
\end{equation}
and both inequalities are saturated when $\gamma'_{A,B}$ reaches its minimum, meaning that $\bigodot a$ is tangent to $\bigodot b$ and $\bigodot d$. Note that $r_c$ does not affect the result. As shown in figure \ref{four_subregion_hyperbolic}, if $\bigodot b$ and circle $\bigodot d$ do not intersect, then $\gamma'_{A,B}$ only contains one segment and the length is entirely determined by $r_a$. If $\bigodot b$ intersects $\bigodot d$ (shown in the middle of figure \ref{four_subregion_hyperbolic}), the length of $\gamma'_{A,B}$ consists of 2 parts. As circle $\bigodot b$ and circle $\bigodot d$ change from being tangent to intersecting, the purely infrared segment of $\gamma'_{A,B}$ increases rapidly in length. Therefore, we have good reason to believe that the total length of $\gamma'_{A,B}$ attains its minimum when circles $\bigodot b$ and $\bigodot d$ are tangent. If we further impose that $\bigodot b$ and $\bigodot d$ are tangent, then we get
\begin{equation}
    r_b+r_d=1,
\end{equation}
which is equivalent to 
\begin{equation}
    r_a = \frac{\sin 2\theta}{2+\sin 2\theta}.
\end{equation}

\begin{figure}[h] 
\centering 
\includegraphics[width=0.7\textwidth]{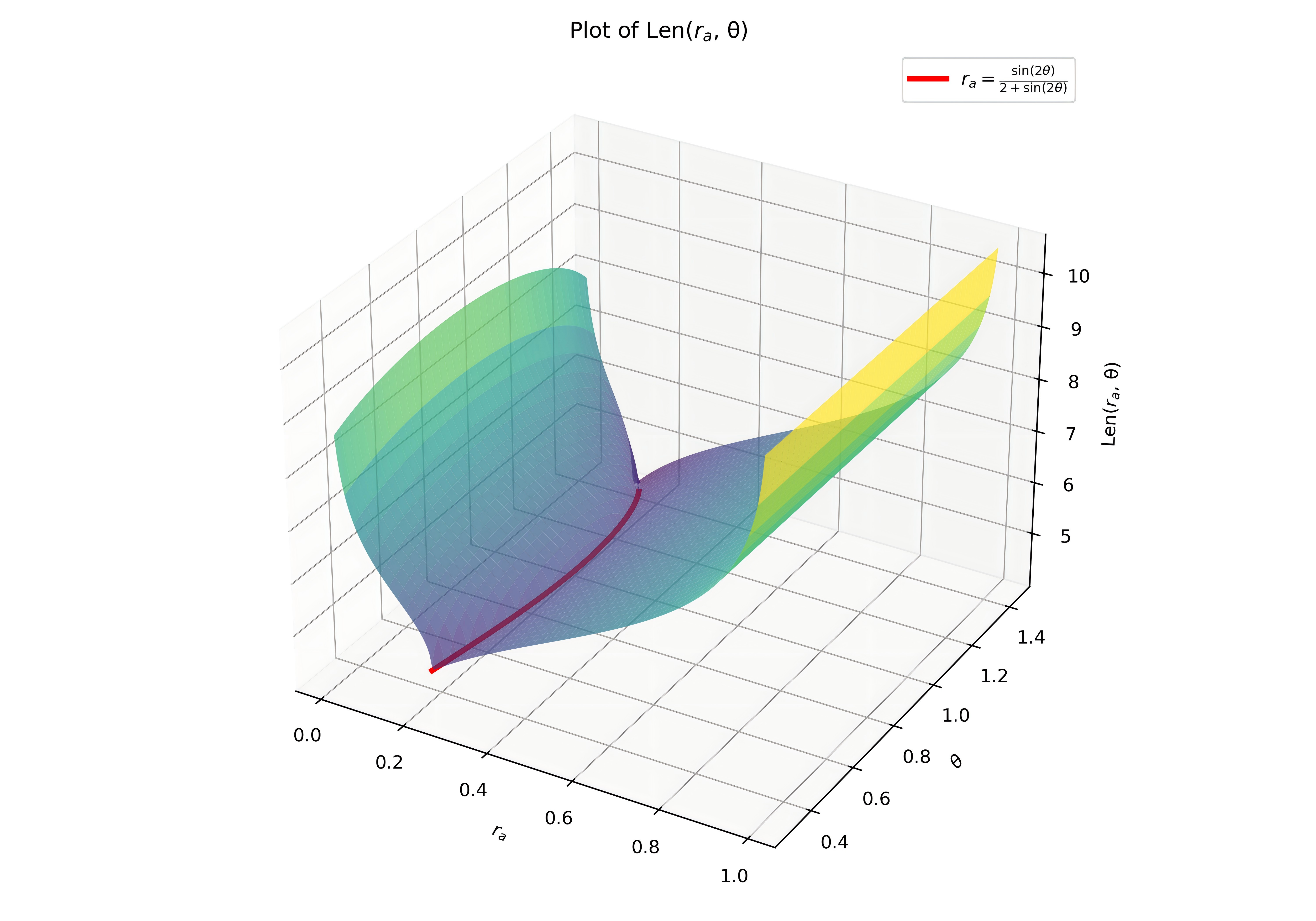} 
\caption{The plot of the length of $\gamma'_{A,B}$, denoted by $\text{Len}(r_a,\theta)$, as a function of $r_a$ and $\theta$ with the UV cutoff $\epsilon=0.01$. The red curve represents the tangency condition between $\bigodot b$ and $\bigodot d$, and it can be observed that the minimum of $\text{Len}(r_a,\theta)$ always lies on this curve.}
\label{plot of Len r theta} 
\end{figure}

We plot the length of $\gamma'_{A,B}$, denoted $\text{Len}(r_a,\theta)$, as a function of $r_a$ and $\theta$ in figure \ref{plot of Len r theta}. It is evident that, for fixed $\theta$ (corresponding to fixed lengths of the boundary subregions), the point where $\text{Len}(r_a,\theta)$ reaches its minimum always lies on the curve $r_a=\frac{\sin 2\theta}{2+\sin 2\theta}$. This agrees with our argument that $\gamma'_{A,B}$ has minimal length when $\bigodot a$, $\bigodot b$, and $\bigodot d$ are pairwise tangent. It then follows that the minimal EWCS is given by
\begin{equation}
    E_W(A:B) = \frac{1}{4G_N}\log\frac{4r_a}{\epsilon(1-r_a)} = \frac{1}{4G_N}\log \left[\frac{2\sin2\theta}{\epsilon} \right] = \frac{1}{4G_N}\log \left[ \frac{4}{\epsilon} \frac{\sqrt{X(A:C)}}{1+X(A:C)}\right],
\end{equation}
where $\epsilon$ is the UV cutoff. {This is claimed to be the minimum of EWCS under the constraint that the density matrices of $A$, $B$, $C$, and $D$ are fixed in holography.}

\subsection{Extremal values of L-entropy and Markov gap}\label{L and M}
\noindent After obtaining the results for EWCS in IR modified geometries, we now analyze the extremal values of the L-entropy and the Markov gap under the constraint that the entanglement wedges of the three boundary subregions $A$, $B$, and $C$ remain invariant after modifying the IR geometry.
\begin{figure}[h] 
\centering 
\includegraphics[width=0.8\textwidth]{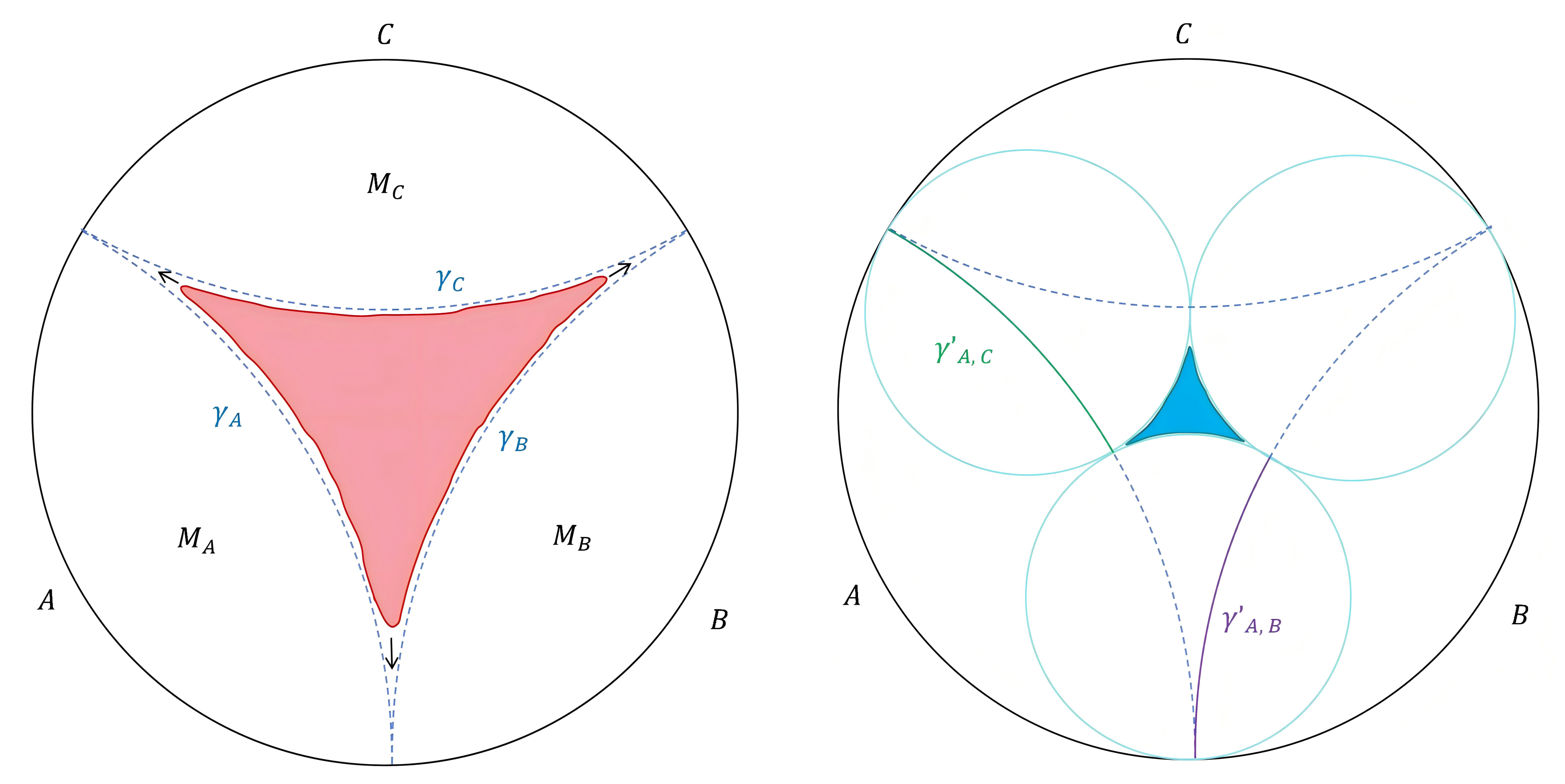} 
\caption{The two configurations in which the EWCS reaches its extremal values with the entanglement wedges of $A$, $B$, and $C$ kept fixed are shown as follows. In the left panel, the region outside the entanglement wedges of $A$, $B$, and $C$ is replaced by the spherical geometry, where $E_W(A:B)$, $E_W(A:C)$, and $E_W(B:C)$ attain their maximal values. In the right panel, the region enclosed by three pairwise tangent horospheres is replaced by the hyperbolic geometry, where $E_W(A:B)$, $E_W(A:C)$, and $E_W(B:C)$ reach their minimal values.}
\label{EWCSextremal}
\end{figure}
In the holographic case, the computations of the L-entropy and the Markov gap both involve the EWCS as well as the entanglement entropies of $A$, $B$, and $C$ (or their linear combinations). Moreover, the L-entropy is negatively correlated with the EWCS, while the Markov gap is positively correlated with it. Therefore, when the entanglement wedges of $A$, $B$, and $C$ are kept fixed, both the L-entropy and the Markov gap attain their extrema when the EWCS reaches its extremal value.

We present in figure \ref{EWCSextremal} the two bulk geometries in which the EWCS reaches its extremal values. In the spherical extremal case, $E_W(A:B)$, $E_W(A:C)$, $E_W(B:C)$ achieve their maximal values,
\begin{equation}
\begin{aligned}
    E_W(A:B)_{\text{max}} &= \min\{S_A,S_B\},\\
    E_W(A:C)_{\text{max}} &= \min\{S_A,S_C\},\\
    E_W(B:C)_{\text{max}} &= \min\{S_B,S_C\}.
\end{aligned}
\end{equation}
In the hyperbolic extremal case, they reach their minimal values,
\begin{equation}
\begin{aligned}
    E_W(A:B)_{\text{min}} &= \frac{1}{2}I(A:B),\\
    E_W(A:C)_{\text{min}} &= \frac{1}{2}I(A:C),\\
    E_W(B:C)_{\text{min}} &= \frac{1}{2}I(B:C).
\end{aligned}
\end{equation}
According to the definitions of the L-entropy and the Markov gap, we can conclude that the L-entropy vanishes in the spherical extremal case, while the Markov gap reaches its maximum. In contrast, in the hyperbolic extremal case, the situation is reversed: the Markov gap vanishes in the hyperbolic extremal case, while the L-entropy reaches its maximum.

{These opposite behaviors confirm our conjecture in Section \ref{EWCS circular} that the two quantities probe distinct types of multipartite entanglement: the Markov gap measures the amount of non-SOTS–type entanglement, whereas the L-entropy quantifies the SOTS-type contribution. Moreover, according to previous results in    \cite{Ju_2024}, in the spherical extremal case, entanglement with correlation length larger than a critical scale is transferred to the longest length scale. As a result, the tripartite entanglement among these subregions becomes more difficult to be transformed into a triangle state via local unitary operations acting on each subsystem, leading to an enhancement of non-SOTS–type entanglement in spherically modified IR geometries. In contrast, in the hyperbolic case, tripartite entanglement with correlation length exceeding the critical scale is truncated, thereby converting non-SOTS–type entanglement into SOTS-type entanglement and rendering the boundary state $\ket{\psi}_{ABC}$ a triangle state. }

\section{Entanglement measures from multi-EWCS for modified IR geometry}\label{section multi EWCS}
\noindent In this section, we study the variation of the multi-EWCS before and after the modification of the bulk IR geometry {in order to analyze the behavior of more-partite entanglement in IR modified geometries}. We begin by briefly reviewing the definition of multi-EWCS and its dual physical quantity in quantum information theory\textemdash the multipartite entanglement of purification (the multi-EoP)    \cite{Umemoto_2018multiEWCS,Bao_2019} in Section \ref{review multiEWCS}. In Section \ref{multiEWCS in modified geometry}, we explore the behavior of multi-EWCS under different bulk geometries. Then in Section \ref{multiEWCS measures}, we begin to discuss the related multipartite entanglement measures and signals. From the behavior of multipartite entanglement measures and signals in different types of bulk geometry, one can extract information about the corresponding entanglement structures of the boundary states. In analyzing the entanglement structure of the boundary state dual to the hyperbolic IR modified geometry, we find that the boundary state exhibits a special class of entanglement structure. 

{As discussed in Section \ref{L and M}, while the reduced density matrices of three boundary subregions $A$, $B$, and $C$ are kept fixed, modifying the hyperbolic geometry allows us to truncate all tripartite entanglement with correlation length greater than the critical scale, thereby rendering $\ket{\psi}_{ABC}$ a triangle state. Motivated by this observation, we conjecture that for a general choice of $n$ boundary subregions, a similar property of the boundary quantum state can be realized through an appropriately modified hyperbolic geometry, and we provide a proof for the cases $n = 4$ and $n=5$ from the perspective of quantum information theory in Section \ref{quadrangle pentagon state}.}

It should be noted that an entanglement signal \cite{balasubramanian2024} is not, strictly speaking, an entanglement measure, since it does not necessarily satisfy the various properties required of a measure as discussed in    \cite{Ma_2024}. Nevertheless, a nonzero signal indicates the presence of genuine tripartite entanglement in the system. In this work, we do not make a distinction between entanglement signals and entanglement measures.

\subsection{Review of multi-EWCS}\label{review multiEWCS}
\begin{figure}[h] 
\centering 
\includegraphics[width=0.5\textwidth]{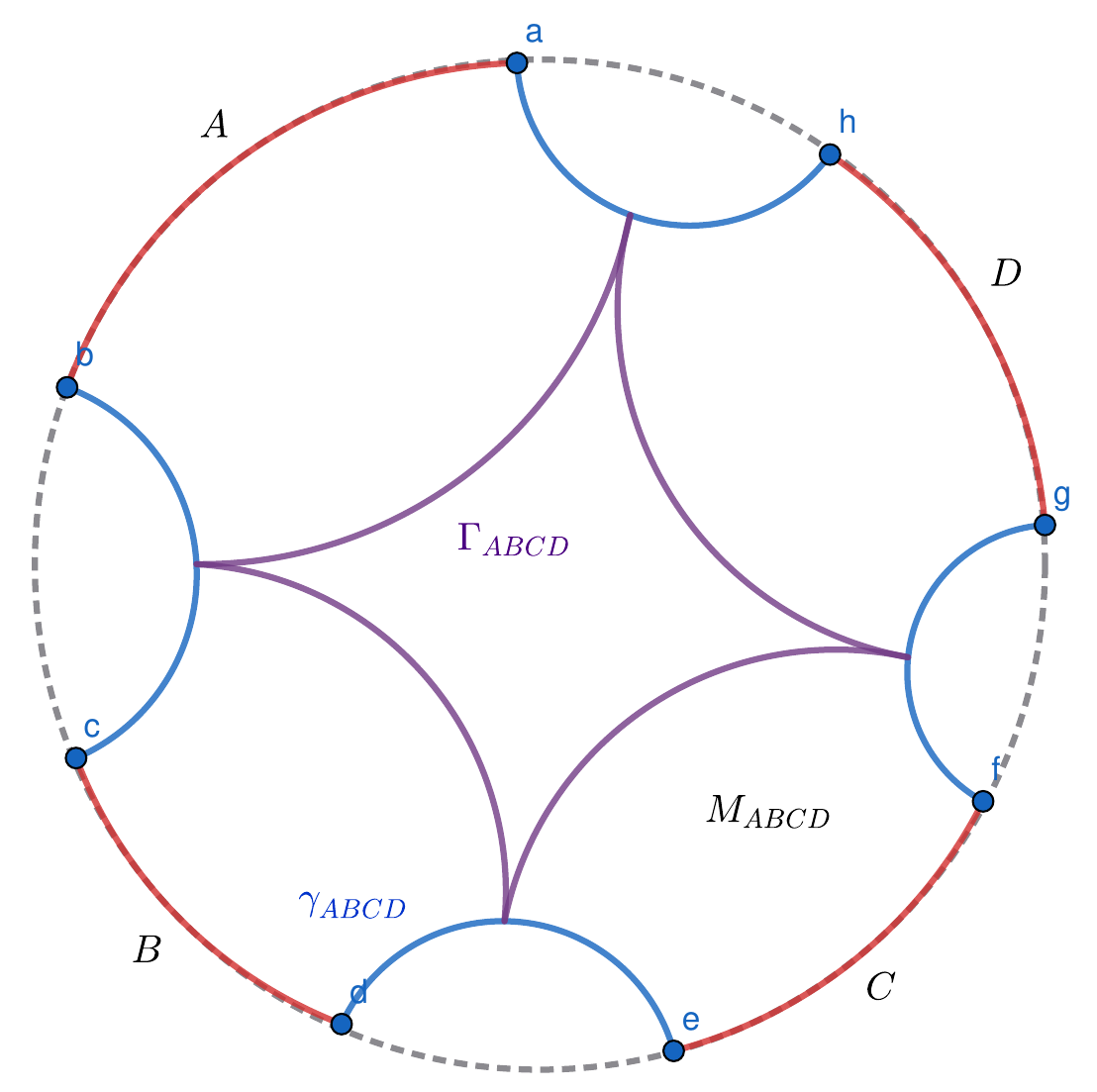} 
\caption{A schematic diagram of $\Gamma_{ABCD}$. The red curves represent the four subregions, $A$, $B$, $C$, and $D$. The blue curve $\gamma_{ABCD}$ represents the minimal surface that is homologous to $ABCD$, and the region they enclose in the bulk is the entanglement wedge of $ABCD$. The purple curves anchored on $\gamma_{ABCD}$ represent $\Gamma_{ABCD}$.} 
\label{definition multiEWCS} 
\end{figure}
\noindent The multi-EWCS is a direct generalization of the bipartite EWCS to the multipartite case. For a collection of non-overlapping subregions on the boundary, $A = A_1 \cup A_2 \cup \dots \cup A_n$, let $\gamma_A$ be the minimal surface homologous to $A$. The region enclosed by $\gamma_A \cup A$ in the bulk is the entanglement wedge, $M_A$, of $A$. We then partition $\gamma_A \cup A$ into $n$ parts, $\gamma_A \cup A = \tilde{A_1} \cup \dots \cup \tilde{A_n}$, such that $A_i \subset \tilde{A_i}$. Treating $M_A$ as a new bulk, $\Gamma_{A_1A_2...A_n}$ is the set of minimal surfaces that are homologous to all $\tilde{A_i}$ and have the minimum area across all such partitions. One explicit example is shown in figure \ref{definition multiEWCS}. The multi-EWCS    \cite{Umemoto_2018multiEWCS, Chu:2019etd, Bao_2019, Bao_2019multiSR} is then defined as 
\begin{equation}
    E_W(A_1:A_2:\dots:A_n) \equiv \frac{\text{Area}(\Gamma_{A_1A_2...A_n})}{4G_N}.
\end{equation}
Substantial evidence indicates that the holographic dual of the multi-EoP is the multi-EWCS    \cite{Umemoto_2018multiEWCS, Bao_2019}. The multi-EoP is defined as
\begin{equation}
    E_P(A_1:A_2:\dots:A_n)\equiv\min_{\ket{\psi}_{A_1A_1'...A_nA_n'}}\sum_iS_{A_iA'_i}.
\end{equation}
Here the minimization is taken over all purifications of $\rho_{A_1A_2...A_n}$. The multipartite entanglement of purification (multi-EoP) is bounded from below by the multipartite mutual information, as well as by bipartite EoP. In particular, Proposition 8 of    \cite{Umemoto_2018multiEWCS} establishes the bound
\begin{equation}\label{lower bound 1}
E_P(A_1:A_2:\dots:A_n)\geq I(A_1:A_2:\dots:A_n),
\end{equation}
where the multipartite mutual information is defined as
\begin{equation}
I(A_1:A_2:\dots:A_n)\equiv \sum_{i=1}^n S_{A_i}-S_{A_1A_2\cdots A_n}.
\end{equation}
We will use this inequality, together with its holographic dual interpretation to prove several of our main results. There exists another lower bound on the multi-EoP, given by
\begin{equation}
E_P(A_1:A_2:\cdots:A_n)\geq \sum_{i=1}^n E_P\left(A_i:A_1\cdots A_{i-1}A_{i+1}\cdots A_n\right).
\end{equation}
For the tripartite case, this inequality reduces to
\begin{equation}\label{lower bound 2}
E_P(A:B:C)\geq E_P(A:BC)+E_P(B:AC)+E_P(C:AB).
\end{equation}
This result corresponds to Proposition 12 of    \cite{Umemoto_2018multiEWCS}, and a holographic proof was subsequently provided in    \cite{Bao_2019}.

\subsection{Multi-EWCS in IR modified geometries}\label{multiEWCS in modified geometry}

\begin{figure}[h] 
\centering 
\includegraphics[width=1\textwidth]{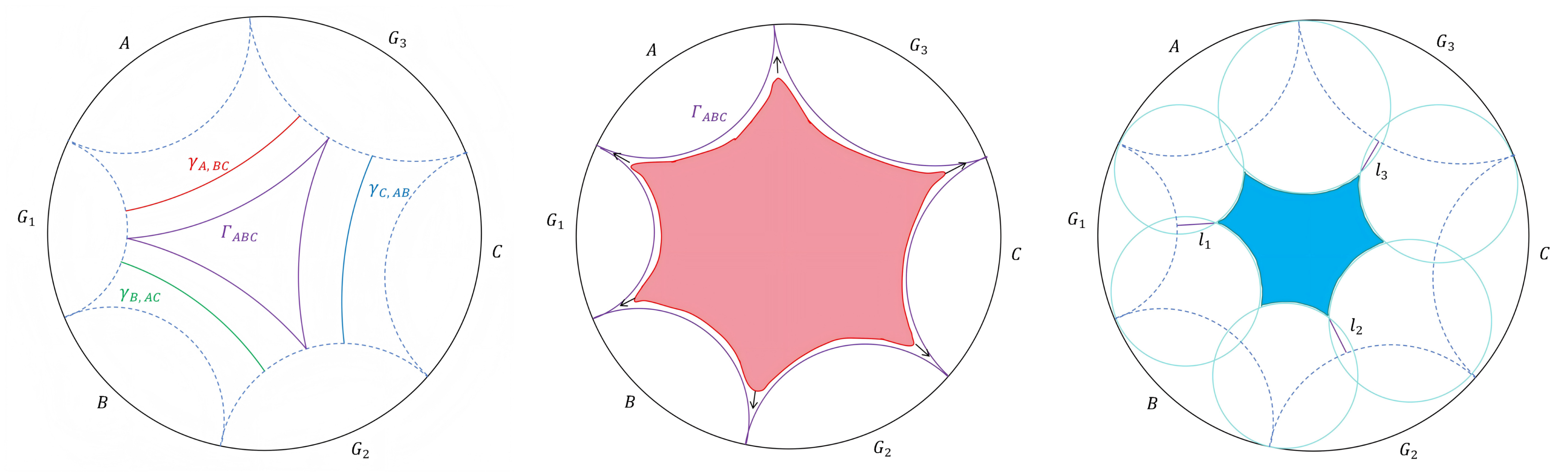} 
\caption{Multi-EWCS in pure AdS, spherically IR modified  geometry, and the hyperbolic IR modified  geometry. In the left figure, the purple curve represents $\Gamma_{ABC}$, while the red, green, and blue curves correspond to $\gamma_{A,BC}$, $\gamma_{B,AC}$, and $\gamma_{C,AB}$, respectively. In the middle figure, we replace the region inside the red region with spherical geometry. As this region gradually approaches the minimal surfaces homologous to $A$, $B$, $C$, $G_1$, $G_2$, and $G_3$, the curve $\Gamma_{ABC}$ coincides with these minimal surfaces, and $\gamma_{A,BC}$, $\gamma_{B,AC}$, and $\gamma_{C,AB}$ (not explicitly shown in the figure) coincide with the minimal surfaces homologous to $A$, $B$, and $C$, respectively. In the right figure, we choose six horospheres tangent to the boundaries of $A$, $B$, $C$, $G_1$, $G_2$, and $G_3$, arranged so that adjacent horospheres are mutually tangent. The curved hexagonal region enclosed by these horospheres is then replaced with the hyperbolic geometry. The curves $l_1$, $l_2$, and $l_3$ denote the shortest geodesics from the boundary of this region to $\gamma_{G_1}$, $\gamma_{G_2}$, and $\gamma_{G_3}$, respectively.} 
\label{figure multiEWCS1} 
\end{figure}
\noindent
We begin by considering three simply connected, non-adjacent subregions—$A$, $B$, and $C$—with the corresponding gap regions denoted as $G_{1,2,3}$. We then move to a configuration where four subregions, $A$, $B$, $C$, and $D$, are adjacent. As shown in figure \ref{figure multiEWCS1}, we select three simply connected boundary subregions $A$, $B$, and $C$, and denote the intervals between them as $G_1$, $G_2$, and $G_3$, which are taken to be sufficiently small so that the entanglement wedge of $ABC$ is connected. In the pure AdS geometry, $\Gamma_{ABC}$ consists of three parts, each larger than $\gamma_{A,BC}$, $\gamma_{B,AC}$, and $\gamma_{C,AB}$, respectively. Therefore, for the pure AdS geometry, we have
\begin{equation}
    E_W(A:B:C) > E_W(A:BC) + E_W(B:AC) + E_W(C:AB).  
\end{equation}
In the extremal case of spherically IR modified geometry, $\Gamma_{ABC}$ coincides with the minimal surfaces homologous to $A$, $B$, $C$, as well as the gap regions $G_1$, $G_2$, and $G_3$, since no geodesic from the boundary enter the IR region. Meanwhile, the curves $\gamma_{A,BC}$, $\gamma_{B,AC}$, and $\gamma_{C,AB}$ coincide with the minimal surfaces homologous to $A$, $B$, and $C$, respectively. Consequently, we obtain
\begin{equation}
E_W(A:B:C) = S_A + S_B + S_C + S_{G_1} + S_{G_2} + S_{G_3},
\end{equation} 
and
\begin{equation}
\begin{aligned}
E_W(A:BC) &= S_A, \
E_W(B:AC) &= S_B, \
E_W(C:AB) &= S_C, \
\end{aligned}
\end{equation}  which are all larger than their values in pure AdS, and are in fact the maximum values for the corresponding quantum marginal problem.

{Then for the hyperbolic IR modified geometry, } as shown in the right figure of figure \ref{figure multiEWCS1}, we replace the curved hexagonal region enclosed by six horospheres with hyperbolic geometry. When the horospheres intersect pairwise and are not too large, we have
\begin{equation}
    E_W(A:B:C) = \frac{\text{Area}(2l_1+2l_2+2l_3)}{4G_N}
\end{equation}
and
\begin{equation}
\begin{aligned}
E_W(A:BC) &= \frac{\text{Area}(l_1+l_3)}{4G_N}, \\
E_W(B:AC) &= \frac{\text{Area}(l_1+l_2)}{4G_N}, \\
E_W(C:AB) &= \frac{\text{Area}(l_2+l_3)}{4G_N}, 
\end{aligned}
\end{equation}
where $l_1$, $l_2$, and $l_3$ denote the shortest geodesics from the boundary of the hexagonal region to $\gamma_{G_1}$, $\gamma_{G_2}$, and $\gamma_{G_3}$, respectively. It is easy to check that, different from the pure AdS case, we now have \begin{equation}
    E_W(A:B:C) = E_W(A:BC) + E_W(B:AC) + E_W(C:AB),  
\end{equation}
which means the inequality (\ref{lower bound 2}) is saturated. {This relation holds irrespective of how large $l_1$, $l_2$, and $l_3$ are. The values of $l_1$, $l_2$, and $l_3$ depend on the specific arrangement of the horospheres. In particular, when all horospheres are mutually tangent, the quantities $E_W(A:B:C)$, $E_W(A:BC)$, $E_W(B:AC)$, and $E_W(C:AB)$ all vanish, thereby attaining their lower bounds under the constraints of the quantum marginal problem.}

Next, we consider a special configuration in which four boundary subregions—$A$, $B$, $C$, and $D$—are adjacent, as illustrated in figure \ref{figure multiEWCS2}. {This setup corresponds to a special case of the above analysis, in which two of the three gap regions shrink to zero.} We modify the IR geometry while keeping the entanglement wedges of these four subregions unchanged. As we will show, under the modified hyperbolic geometry, the corresponding boundary quantum state $\ket{\psi}_{ABCD}$ exhibits a distinctive property: only neighboring subregions share bipartite entanglement, a structure we refer to as a \textit{quadrangle state}.
\begin{figure}[h] 
\centering 
\includegraphics[width=1\textwidth]{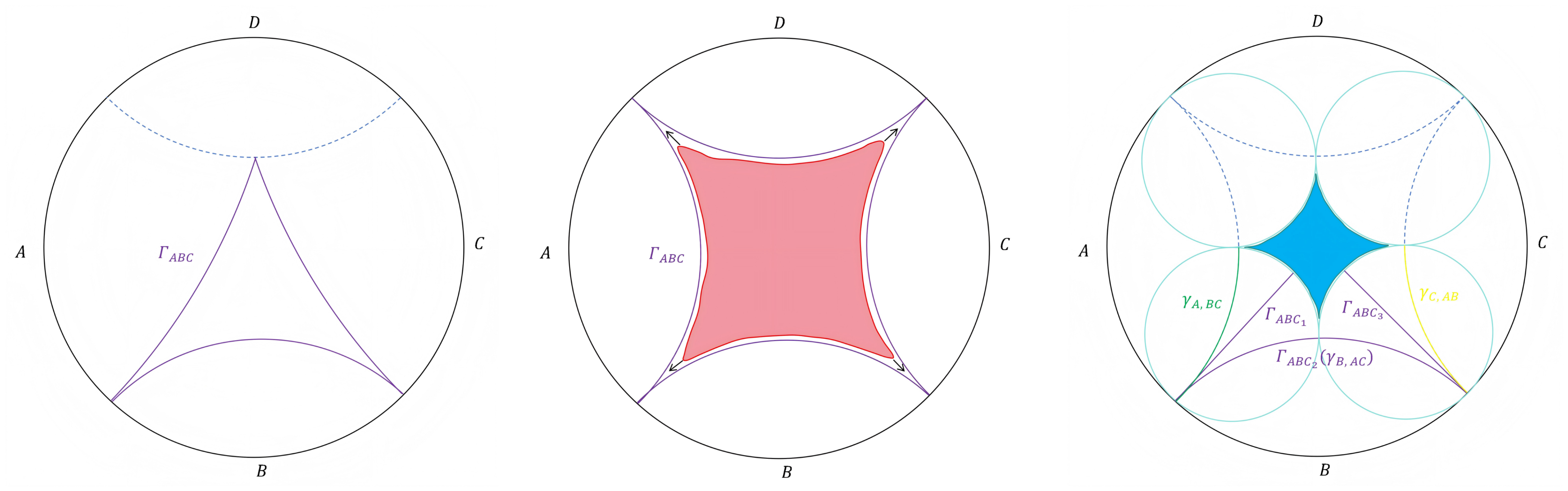}
\caption{Shapes of $\Gamma_{ABC}$ in the pure AdS geometry, the IR spherical modified geometry, and the IR hyperbolic modified geometry when $A$, $B$, $C$ are adjacent. As shown in the middle figure, while keeping the reduced density matrices of the four boundary subregions $A$, $B$, $C$, and $D$ unchanged, when we gradually enlarge the boundary of the spherical geometry region so that it approaches the three minimal surfaces, $\Gamma_{ABC}$(denoted by the purple curve) progressively approaches them. In the case of the modified hyperbolic geometry, we choose the IR region as the region bounded by four horospheres tangent to the boundaries of the subsystems. In this case, $\Gamma_{ABC}$ is likewise represented by the purple curve and it can be divided into three segments, each having the same length as $\gamma_{A,BC}$, $\gamma_{B,AC}$, and $\gamma_{C,AB}$, respectively.} 
\label{figure multiEWCS2} 
\end{figure}

The shapes of $\Gamma_{ABC}$ in the AdS vacuum, the spherically IR modified geometry, and the hyperbolic IR modified geometry when $A$, $B$, $C$ and $D$ are adjacent subregions are shown in figure \ref{figure multiEWCS2}. In the spherical extremal case, $\Gamma_{ABC}$ can be made to coincide with the minimal surfaces homologous to $A$, $B$, $C$, and $D$ by choosing the IR spherical region as large as possible, thereby resulting in
\begin{equation}
    E_W(A:B:C) = S_A+S_B+S_C+S_D,
\end{equation} which is also the upper bound value for this quantum marginal problem.

In the case of the modified hyperbolic geometry, we consider the situation where the cross-ratio of subregions $A$ and $C$ is $X(A:C)=1$. The reason is that in this case, one can always perform a conformal transformation so that the four subregions $A$, $B$, $C$, and $D$ have the same size, which in turn ensures that the four horospheres in figure \ref{figure multiEWCS2} are equal in size and mutually tangent. It is then straightforward to see that, under the modified hyperbolic geometry, inequality 
\begin{equation}
    E_W(A:B:C)\geq I(A:B:C) 
\end{equation}
is saturated. {$E_W(A:B:C)$ therefore reaches the lower bound value in the quantum marginal problem, which is the property associated with IR hyperbolic geometries. Moreover, this property has another important consequence in the behavior of the corresponding state, which we will explain in the next subsections. }
\subsection{Two multipartite entanglement signals related to the multi-EWCS}\label{multiEWCS measures}
\noindent
In this subsection, we investigate the multipartite entanglement structure of the boundary quantum state {dual to the IR modified geometries with the results obtained above}. We begin by examining the behavior of two multipartite entanglement signals associated with the multi-EWCS in the modified geometry. The first is the tripartite entanglement signal introduced in    \cite{bao2025}, and the second is a newly defined multipartite entanglement signal, obtained as a generalization of the quantity $g(A:B)$ introduced in    \cite{Zou_2021}.

\subsubsection{$\Delta_{w}^{(3)}(A:B:C)$ for modified IR geometry}\label{tripartite signal}
\noindent According to    \cite{bao2025}, the holographic tripartite entanglement signal $\Delta_{w}^{(3)}(A:B:C)$ is defined as
\begin{equation}
\Delta_{w}^{(3)}(A:B:C) \equiv E_W(A:B:C) - E_W(A:BC) - E_W(B:AC) - E_W(C:AB)
\end{equation} for a mixed state $ABC$.
The boundary dual of this signal is 
\begin{equation}
\Delta_{p}^{(3)}(A:B:C) \equiv E_P(A:B:C) - E_P(A:BC) - E_P(B:AC) - E_P(C:AB).
\end{equation}
Inequality (\ref{lower bound 2}) ensures that this signal is non-negative. It was shown in    \cite{bao2025} that $\Delta_{p}^{(3)}(A:B:C)=0$ when the state $\rho_{ABC}$ contains only classical correlations and bipartite entanglement among $A$, $B$ and $C$, or when $\rho_{ABC}$ is pure. Consequently, $\Delta_{p}^{(3)}(A:B:C)$ provides a reliable signal of genuine tripartite entanglement for mixed states. We choose $A$, $B$ and $C$ to be three intervals with gap regions $G_1$, $G_2$ and $G_3$ between them. Utilizing results from Section \ref{multiEWCS in modified geometry}, 
in the case of the spherically IR modified geometry, we can make this tripartite entanglement signal attain a large value, which is
\begin{equation}
    \Delta_{w}^{(3)}(A:B:C) = S_{G_1}+S_{G_2}+S_{G_3}>0.
\end{equation}
In contrast, under the hyperbolic IR modified geometry, we find $\Delta_{w}^{(3)}(A:B:C) = 0$

{In the spherical case, the tripartite entanglement signal is nonvanishing, indicating the presence of genuine tripartite entanglement among $A$, $B$, and $C$.}
In contrast, this signal vanishes in the hyperbolic case. {However, this does not imply the absence of tripartite entanglement; rather, there may still exist tripartite entanglement that this signal fails to detect.} In fact, there exists a special  entanglement structure in the hyperbolic case as will be revealed in the next subsection using another signal.  

A similar behavior occurs in the special case when $A$, $B$, $C$, and $D$ are adjacent: under the spherically IR modified geometry we have $\Delta_{w}^{(3)}(A:B:C)=S_D$ while in the hyperbolic case we have $\Delta_{w}^{(3)}(A:B:C)=0$ as well.

\subsubsection{A newly defined multipartite entanglement signal for modified IR geometry}\label{multi-MG}
\noindent We define a multipartite entanglement signal as 
\begin{equation}\label{ME signal g}
    g(A_1:A_2:\cdots :A_n)\equiv E_P(A_1:A_2:\cdots :A_n)-I(A_1:A_2:\cdots :A_n),
\end{equation}
This signal is also a non-negative quantity due to inequality (\ref{lower bound 1}). This definition reduces to the $g(A:B)$ function introduced in    \cite{Zou_2021} when $n=2$. It is straightforward to verify that in the right panel of figure \ref{figure multiEWCS2} and in figure \ref{figure multiEWCS 5}, we have $g(A:B:C)=0$ and $g(A:B:C:D)=0$  for the hyperbolic cases, respectively. Since the vanishing of two partite $g(A:B)$ implies that the pure state $\ket{\psi}_{ABC}$ lacks genuine tripartite entanglement in the perspective of smaller subsystems, it is natural to conjecture that in these two cases, the corresponding boundary states also lack genuine multipartite entanglement in certain sense. 

Another piece of evidence comes from the bipartite entanglement structure studied in    \cite{Ju_2024}, where it was shown that there exists a critical length {beyond which the conditional mutual information is truncated}. In the configuration illustrated in the right figure of figure \ref{figure multiEWCS2}, this critical length coincides with the sizes of the boundary subregions $A$, $B$, $C$, and $D$. Consequently, entanglement beyond this critical scale is eliminated, meaning that only adjacent subregions share bipartite entanglement, while there is no quantum correlation between non-adjacent pairs such as $A$ and $C$ or $B$ and $D$. Therefore, we assert that there exists bipartition $\mathcal{H}_{\alpha}=\mathcal{H}_{\alpha_L}\otimes\mathcal{H}_{\alpha_R}$($\alpha=A,B,C,D$) such that 
\begin{equation}
    \ket{\psi}_{ABCD}=\ket{\psi}_{A_LB_R}\ket{\psi}_{B_LC_R}\ket{\psi}_{C_LD_R}\ket{\psi}_{A_RD_L}.
\end{equation}
 In analogy with the triangle state, we name states with this type of entanglement structure \textit{quadrangle states}. The proof that the boundary state of the right figure of figure \ref{figure multiEWCS2} is a quadrangle state will be provided in Section \ref{quadrangle pentagon state}.
This result can also be generalized to the multipartite case. For instance, as shown in figure \ref{figure multiEWCS 5}, consider five boundary subregions $A, B, C, D, E$, and require that the horospheres tangent to their boundaries are pairwise tangent, while the geometry is modified only within the pentagon enclosed by these horospheres. In the extremal case of the modified hyperbolic geometry, long-range entanglement is truncated, so quantum entanglement should exist only between adjacent subregions, and the boundary CFT state $\ket{\psi}_{ABCDE}$ can be written as
\begin{equation}
    \ket{\psi}_{ABCDE} = \ket{\psi}_{A_LB_R}\ket{\psi}_{B_LC_R}\ket{\psi}_{C_LD_R}\ket{\psi}_{D_LE_R}\ket{\psi}_{A_RE_L}.
\end{equation}
for a suitable bipartition $\mathcal{H}_{\alpha}=\mathcal{H}_{\alpha_L}\otimes\mathcal{H}_{\alpha_R}$($\alpha=A,B,C,D,E$). We name states with this type of entanglement structure \textit{quadrangle states}. We will present the proof of this statement using tools from quantum information theory in Section \ref{quadrangle pentagon state} as well.

\begin{figure}[h] 
\centering 
\includegraphics[width=0.6\textwidth]{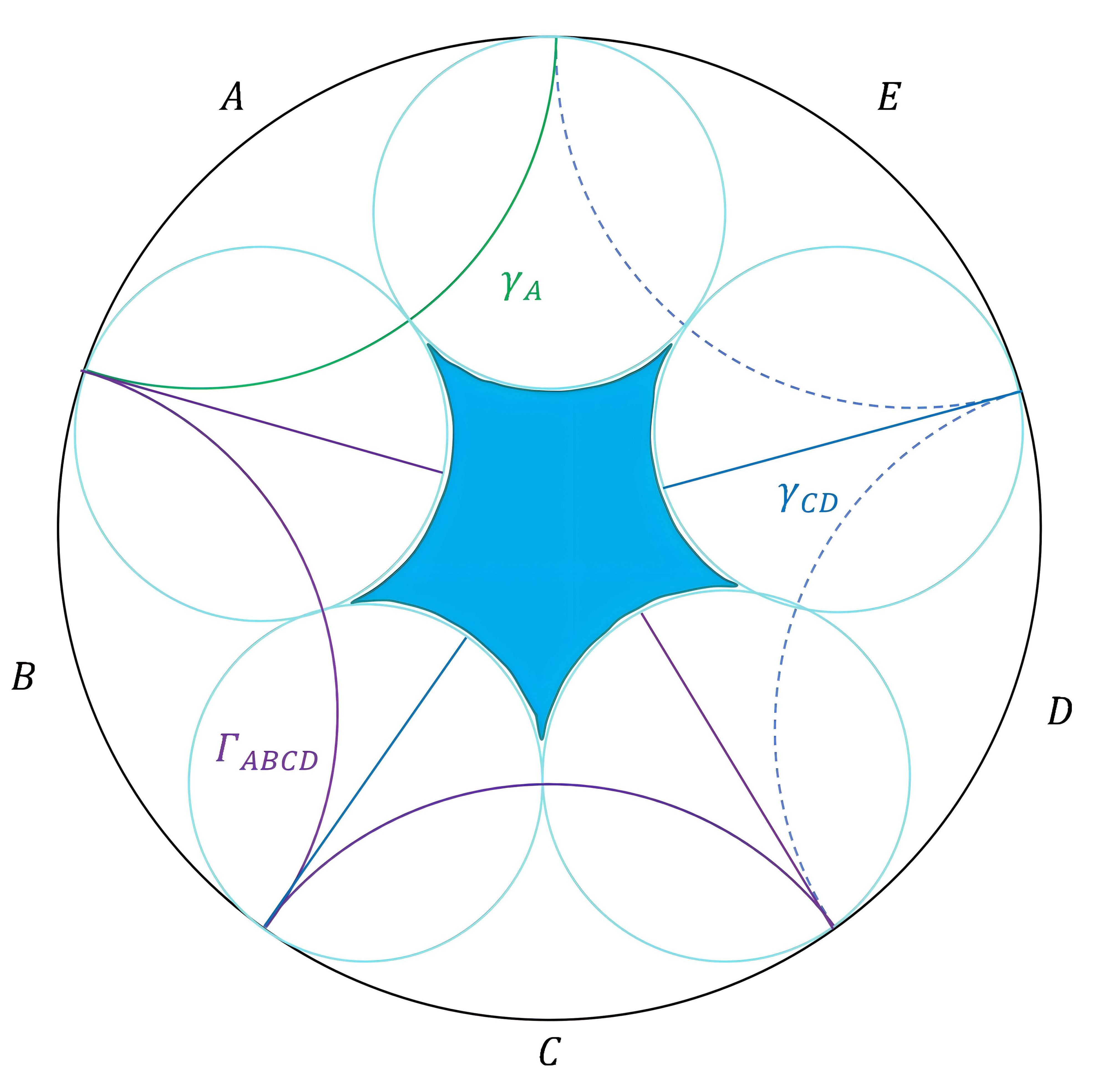} 
\caption{The modified hyperbolic geometry with entanglement wedges of 5 subregions unchanged. The blue curve and the green curve represent the minimal surfaces homologous to $CD$ and $A$, respectively. The purple curve represents $\Gamma_{ABCD}$. By a straightforward calculation, one finds that $g(A:B:C:D)=0$, and the mutual information between non-adjacent subregions vanishes.} 
\label{figure multiEWCS 5} 
\end{figure}

\subsection{Multipartite entanglement structure from hyperbolic IR modified  geometry}\label{quadrangle pentagon state}
\noindent At the end of the previous section, we asserted that in the extremal case of such a deformation, the boundary quantum state can exhibit a distinctive feature: among $n$ boundary subregions, only adjacent ones share bipartite entanglement, while genuine multipartite entanglement is absent. In this section, we provide proofs for the cases of $n=4$ and $n=5$ using methods from quantum information theory, and argue that this conclusion can be naturally extended to systems with more subregions.
\subsubsection{Quadrangle states from the multipartite entanglement signal $g(A:B:C)$ 
}
\label{quadrangle state}

\noindent We first present a theorem on the relationship between the multipartite {entanglement signal (\ref{ME signal g})} and the entanglement structure of a quantum state, which can be regarded as a generalization of theorem 2 in    \cite{Zou_2021}.
\begin{theorem}\label{thm 2}
A state $\ket{\psi}_{ABCD}$ can be written in the form
\begin{equation}\label{eq thm 2}
    \ket{\psi}_{ABCD}=\ket{\psi}_{A_1D_1}\ket{\psi}_{B_1D_2}\ket{\psi}_{C_1D_3}\ket{\psi}_{A_2B_2C_2}
\end{equation}
with an appropriate partition of the local Hilbert spaces $\mathcal{H}_\alpha=\mathcal{H}_{\alpha_1}\otimes\mathcal{H}_{\alpha_2}$ ($\alpha=A,B,C$) and $\mathcal{H}_D = \mathcal{H}_{D_1}\otimes\mathcal{H}_{D_2}\otimes\mathcal{H}_{D_3}$ up to local unitary transformations if and only if $g(A:B:C) = 0$.
\end{theorem}
Although we cannot conclude, as in the tripartite case, that $\ket{\psi}_{ABCD}$ is a quadrangle state from the vanishing of the multipartite entanglement signal, we can assert that there is no nontrivial four-partite entanglement in $\ket{\psi}_{ABCD}$ in the perspective of smaller subsystems. The proof of this theorem is provided in Appendix \ref{B}. In order to further prove, on this basis, that the boundary quantum state $\psi_{ABCD}$ dual to the modified hyperbolic geometry is a quadrangle state, we also need to use the fact that $I(A:C) = I(B:D) = 0$, as well as the following lemma.
\begin{lemma}\label{lemma1 main text}
Let $\rho_{AB}$ be a quantum state on $\mathcal{H}=\mathcal{H}_A\otimes\mathcal{H}_B$. If $I(A:B)=0$, then any purification $\ket{\psi}_{ABC}$ of $\rho_{AB}$ can be decomposed into the form
\begin{equation}
    \ket{\psi}_{ABC} = \ket{\psi}_{AC^{(A)}}\otimes\ket{\psi}_{BC^{(B)}}
\end{equation}
up to local unitary transformations.
\end{lemma}
The proof of this lemma is straightforward. Since $I(A:B)=0$, we have $\rho_{AB}=\rho_A\otimes\rho_B$. We can then purify $\rho_A$ and $\rho_B$ separately, obtaining $\ket{\psi}_{AA'}\otimes\ket{\psi}_{BB'}$, and use the fact that any two purifications of $\rho_{AB}$ differ by a local unitary transformation. The fact that $g(A:B:C)=0$ implies that the boundary state $\ket{\psi}_{ABCD}$ can be expressed in the form of equation (\ref{eq thm 2}). By repeatedly applying inequality $I(A:BC)\geq I(A:B)$, we obtain 
\begin{equation}
\begin{aligned}
    0&=I(B:D)\geq I(B_1:D_2)\geq  0 \\
    0&=I(A:C)\geq I(A_2:C_2)\geq  0, \\
\end{aligned}
\end{equation}
so we have $I(B_1:D_2)=I(A_2:C_2)=0$. According to Lemma \ref{lemma1 main text}, $\ket{\psi}_{B_1D_2}$ and $\ket{\psi}_{A_2B_2C_2}$ can be further decomposed as $\ket{\psi}_{B_1}\otimes\ket{\psi}_{D_2}$ and $\ket{\psi}_{A_2B_2^{(A)}}\otimes\ket{\psi}_{C_2B_2^{(C)}}$, respectively, so that $\ket{\psi}_{ABCD}$ becomes a quadrangle state.

\subsubsection{Pentagon states from the multipartite entanglement signal $g(A:B:C:D)$ }
\label{pentagon state}
\noindent We first generalize Theorem \ref{thm 2} from the previous section to the five-partite case and provide the proof in the appendix. Subsequently, we will use this generalized theorem together with the fact that the mutual information between non-adjacent subregions vanishes to prove that, under the modified geometry shown in figure \ref{figure multiEWCS 5}, the boundary quantum state is a pentagon state.
\begin{theorem}\label{thm 3}
A state $\ket{\psi}_{ABCDE}$ can be written in the form
\begin{equation}\label{equation thm 3}
    \ket{\psi}_{ABCDE}=\ket{\psi}_{A_1E_1}\ket{\psi}_{B_1E_2}\ket{\psi}_{C_1E_3}\ket{\psi}_{D_1E_4}\ket{\psi}_{A_2B_2C_2D_2}
\end{equation}
with an appropriate partition of the local Hilbert spaces $\mathcal{H}_\alpha=\mathcal{H}_{\alpha_1}\otimes\mathcal{H}_{\alpha_2}$ ($\alpha=A,B,C,D$) and $\mathcal{H}_E = \mathcal{H}_{E_1}\otimes\mathcal{H}_{E_2}\otimes\mathcal{H}_{E_3}\otimes\mathcal{H}_{E_4}$ up to local unitary transformations if and only if $g(A:B:C:D) = 0$.
\end{theorem}
If the four-partite entanglement signal $g(A:B:C:D)=0$, then $E$ can be decomposed into four parts, each entangled with $A$, $B$, $C$, and $D$ respectively in a bipartite manner. In addition, there may still exist genuine four-partite entanglement among $A$, $B$, $C$, and $D$, but no nontrivial five-partite entanglement exists in the whole system in the perspective of smaller subsytems. 

By applying Theorem \ref{thm 3} and the fact that the mutual information of non-adjacent subregions vanishes, we can further demonstrate that the boundary state $\ket{\psi}_{ABCDE}$ is a pentagon state. According to inequality $I(A:BC)\geq I(A:B)$ we have
\begin{equation}
\begin{aligned}
    0&=I(B:E)\geq I(B_1:E_2) \geq 0\\
    0&=I(C:E)\geq I(C_1:E_2)\geq 0\\
    0&=I(A:CD) \geq I(A_2:C_2D_2) \geq 0
\end{aligned}
\end{equation}
and consequently obtain $I(B_1:E_2)= I(C_1:E_2)= I(A_2:C_2D_2)=0$. Utilizing Lemma \ref{lemma1 main text}, we can factorize $\ket{\psi}_{ABCDE}$ into
\begin{equation}
    \ket{\psi}_{ABCDE}=\ket{\psi}_{A_1E_1}\ket{\psi}_{B_1}\ket{\psi}_{E_2}\ket{\psi}_{C_1}\ket{\psi}_{E_3}\ket{\psi}_{D_1E_4}\ket{\psi}_{A_2B_2^{(A)}}\ket{\psi}_{B_2^{(CD)}C_2D_2}.
\end{equation}
Furthermore, because $I(B:D) = I(B_2^{(CD)}:D_2) = 0$, the state $\ket{\psi}_{B_2^{(CD)}C_2D_2}$ can be factorized as well. We have
\begin{equation}
    \ket{\psi}_{B_2^{(CD)}C_2D_2} = \ket{\psi}_{B_2^{(CD)}C_2^{(B)}}\ket{\psi}_{C_2^{(D)}D_2}.
\end{equation}
This series of steps successfully decomposes the full state $\ket{\psi}_{ABCDE}$ into a form where entanglement exists exclusively between adjacent subregions, which is a defining characteristic of a pentagon state.

From the above analysis, we demonstrated that, in the hyperbolic extremal case, when the horospheres defining the IR region are made tangent to their neighboring ones, the boundary quantum state exhibits a distinctive feature that we refer to as \textit{the polygon state}, a configuration in which only adjacent subregions share bipartite entanglement. We provided rigorous proof of this result for the cases with $n=4$ and $n=5$ boundary subregions, and we expect that the same conclusion holds for arbitrary $n$.

\section{Entanglement measures from multi-entropy for modified IR geometry}\label{multi entropy}
\noindent In this section, we focus on the behavior of the multi-entropy    \cite{Gadde_2022} under modified IR geometries. We begin by reviewing the definition of multi-entropy and its holographic dual. For a $\mathrm{q}$-partite pure state $\ket{\psi}_{A_1A_2\cdots A_{\mathrm{q}}}$, the \textit{n}-th Rényi $\mathrm{q}$-partite multi-entropy is defined by
\begin{equation}
    S_n^{(\mathrm{q})}(A_1:\cdots A_{\mathrm{q}})
    \equiv
    \frac{1}{1-n}\frac{1}{n^{\mathrm{q}-2}}\log \frac{Z_n^{(\mathrm{q})}}{\left(Z_1^{(\mathrm{q})}\right)^{n^{\mathrm{q}-1}}},
\end{equation}
\begin{equation}
    Z_n^{(\mathrm{q})}\equiv \bra{\psi}^{\otimes n^{\mathrm{q}-1}}\Sigma_1(g_1)\Sigma_2(g_2)\cdots \Sigma_\mathrm{q}(g_\mathrm{q})\ket{\psi}^{\otimes n^{\mathrm{q}-1}},
\end{equation}
where $\Sigma_i\left(g_\mathrm{k}\right)$ are the twist operators that implement the permutation action of $g_\mathrm{k}$ on indices of reduced density matrices for subsystem $A_\mathrm{k}$. Here, the permutation of $g_\mathrm{k}$ is defined by the discrete translation along the $\mathrm{k}$-th direction on a $(\mathrm{q} - 1)$-dimensional hypercube of length $n$, that is
\begin{equation}
    g_{\mathrm{k}}\cdot \left(x_1,x_2,\cdots,x_\mathrm{k},\cdots,x_{\mathrm{q}-1}\right)=\left(x_1,x_2,\cdots,x_\mathrm{k}+1,\cdots,x_{\mathrm{q}-1}\right),\qquad \mathrm{k}=1,2,\cdots \mathrm{q}-1,
\end{equation}
\begin{equation}
    g_{\mathrm{q}}\cdot \left(x_1,x_2,\cdots,\cdots,x_{\mathrm{q}-1}\right)=
    \left(x_1,x_2,\cdots,\cdots,x_{\mathrm{q}-1}\right).
\end{equation}
Here $\left(x_1,x_2,\cdots,\cdots,x_{\mathrm{q}-1}\right)$ denotes an integer lattice point on the $(\mathrm{q}-1)$-dimensional hypercube of length $n$ with the periodic boundary condition. The multi-entropy is defined by 
\begin{equation}
    S^{(\mathrm{q})}\left(A_1:A_2:\cdots A_{\mathrm{q}}\right)\equiv \lim_{n\to 1}S_n^{(\mathrm{q})}\left(A_1:A_2:\cdots A_{\mathrm{q}}\right).
\end{equation}
The multi-entropy serves as the multipartite generalization of entanglement entropy. When we take $\mathrm{q}=2$, it reduces to the standard von Neumann entropy of a bipartite system. Some other properties of the multi-entropy such as additivity can be found in    \cite{Penington_2023, Gadde_2024}.

The holographic dual of the multi-entropy for $\mathrm{q}$ subregions is conjectured to be a minimal bulk web    \cite{Gadde_2022}, $\mathcal{W}$, which satisfies the following two conditions. First, $\mathcal{W}$ is anchored to the boundaries of all boundary subregions. Second, $\mathcal{W}$ contains sub-webs that are homologous to each of the boundary subregions. Two explicit examples are shown in figure \ref{Wdefinition}.
\begin{figure}[h] 
\centering 
\includegraphics[width=0.85\textwidth]{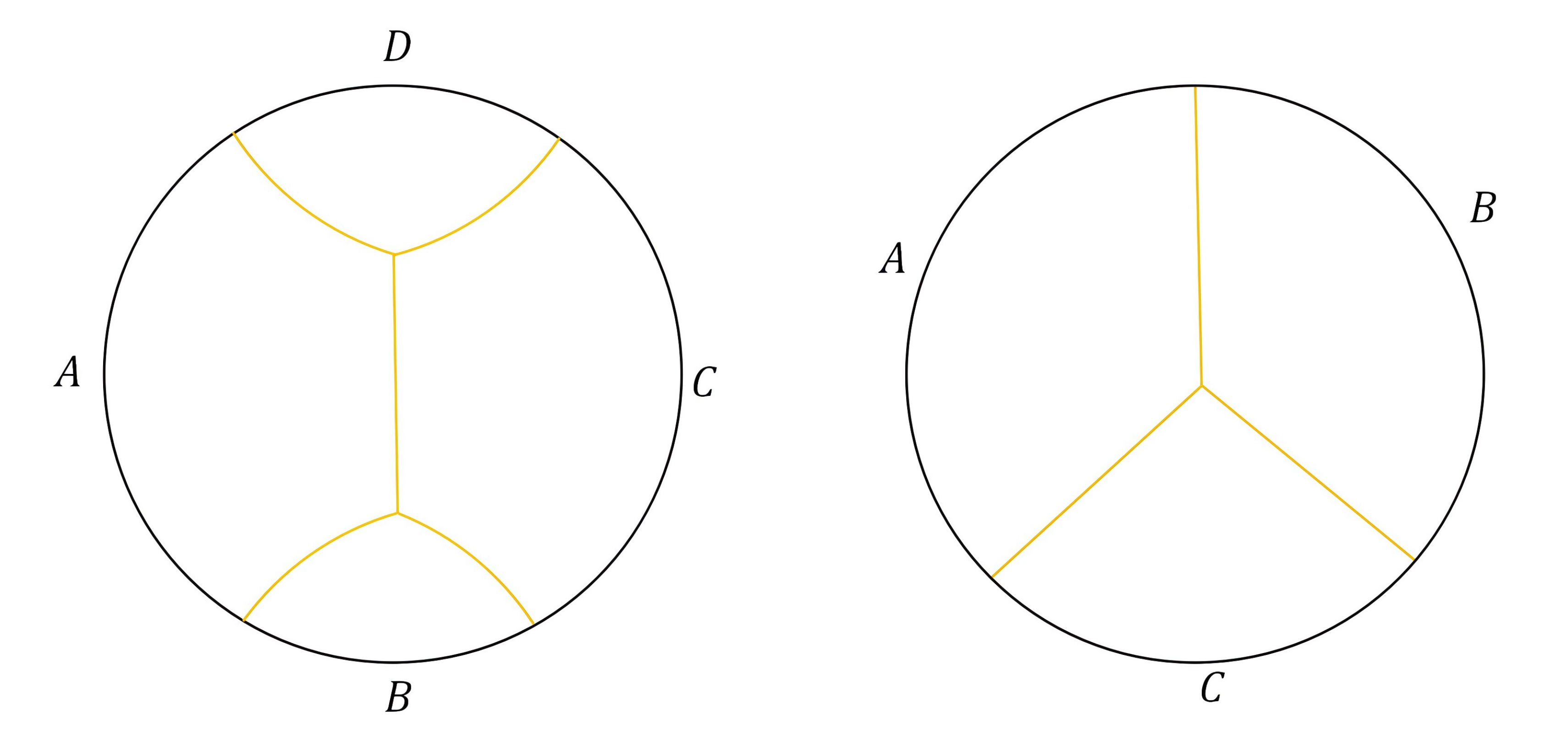} 
\caption{Holographic duals of $S^{(4)}(A:B:C:D)$ (left) and $S^{(3)}(A:B:C)$ (right). The two bulk webs, shown in yellow, are anchored to the boundaries of the corresponding subregions. Moreover, $\mathcal{W}$ contains sub-webs that are homologous to each of these subregions.} 
\label{Wdefinition} 
\end{figure}

In this section, we mainly focus on two tripartite entanglement measures for a pure state $\ket{\psi}_{ABC}$ built from the multi-entropy. The first is $\kappa$, which is studied in    \cite{liu2024} on 2d CFTs, and is defined as
\begin{equation}
    \kappa \equiv S^{(3)}(A:B:C)-\frac{1}{2}\left(S_A+S_B+S_C\right).
\end{equation}
This quantity coincides with the genuine multi-entropy $\text{GM}^{(\mathrm{q})}\left(A_1:A_2:\cdots A_{\mathrm{q}}\right)$ defined in    \cite{iizuka2025GM, iizuka2025MGM} when $\mathrm{q}=3$. It was shown in    \cite{iizuka2025GM, iizuka2025MGM} that $\kappa$, or equivalently $\mathrm{GM}^{(3)}$, vanishes for both separable states and triangle states. Consequently, it is regarded as a measure of genuine tripartite entanglement\footnote{It is worth noting that different criteria for genuine tripartite entanglement is adopted in    \cite{iizuka2025GM, iizuka2025MGM} and    \cite{Basak2025}. A key distinction is that triangle states have vanishing $\kappa$ and are therefore regarded as lacking genuine tripartite entanglement in    \cite{iizuka2025GM, iizuka2025MGM}, whereas their L-entropy remains nonzero. But $\kappa$ could be nonzero for SOTS, such as the GHZ state. \label{f6}}. The other measure we will consider is one that we define as 
\begin{equation}
    \Upsilon \equiv \min \left\{S_A+S_B, S_A+S_C, S_B+S_C\right\} - S^{(3)}(A:B:C).
\end{equation}
In holography, $\Upsilon$ is non-negative due to the inequality
\begin{equation}\label{inequality upsilon}
    S^{(3)}(A:B:C) \leq \min \left\{S_A+S_B, S_A+S_C, S_B+S_C\right\}.
\end{equation}
This inequality follows from the fact that $\gamma_A + \gamma_B$, $\gamma_A + \gamma_C$, and $\gamma_B + \gamma_C$ are all valid candidates for $\mathcal{W}$, and thus their lengths must be shorter than or equal to that of $\mathcal{W}$. However, we emphasize that this inequality does not necessarily hold for general {non-holographic} quantum states. Although the physical meaning of $\Upsilon$ remains uncertain, the toy model of the modified geometry offers a framework from which we can conjecture its possible interpretation, as discussed in Section \ref{entanglement structure}.

Similar to the previous sections, we choose the IR region in such a way that the entanglement wedges of certain boundary subregions remain unchanged. We mainly consider the case of three adjacent subregions on the boundary and show that, in the extremal cases of modified spherical and hyperbolic geometries, the tripartite multi-entropy reaches its upper and lower bounds, respectively. By analyzing the entanglement structure of the boundary CFT under the modified hyperbolic geometry together with the behavior of this measure, we propose a conjecture about the physical interpretation of $\Upsilon$, namely, that it quantifies the weakest bipartite entanglement among the three subregions.

\subsection{Multi-entropy for the modified spherical geometry}
\begin{figure}[h] 
\centering
\includegraphics[width=0.95\textwidth]{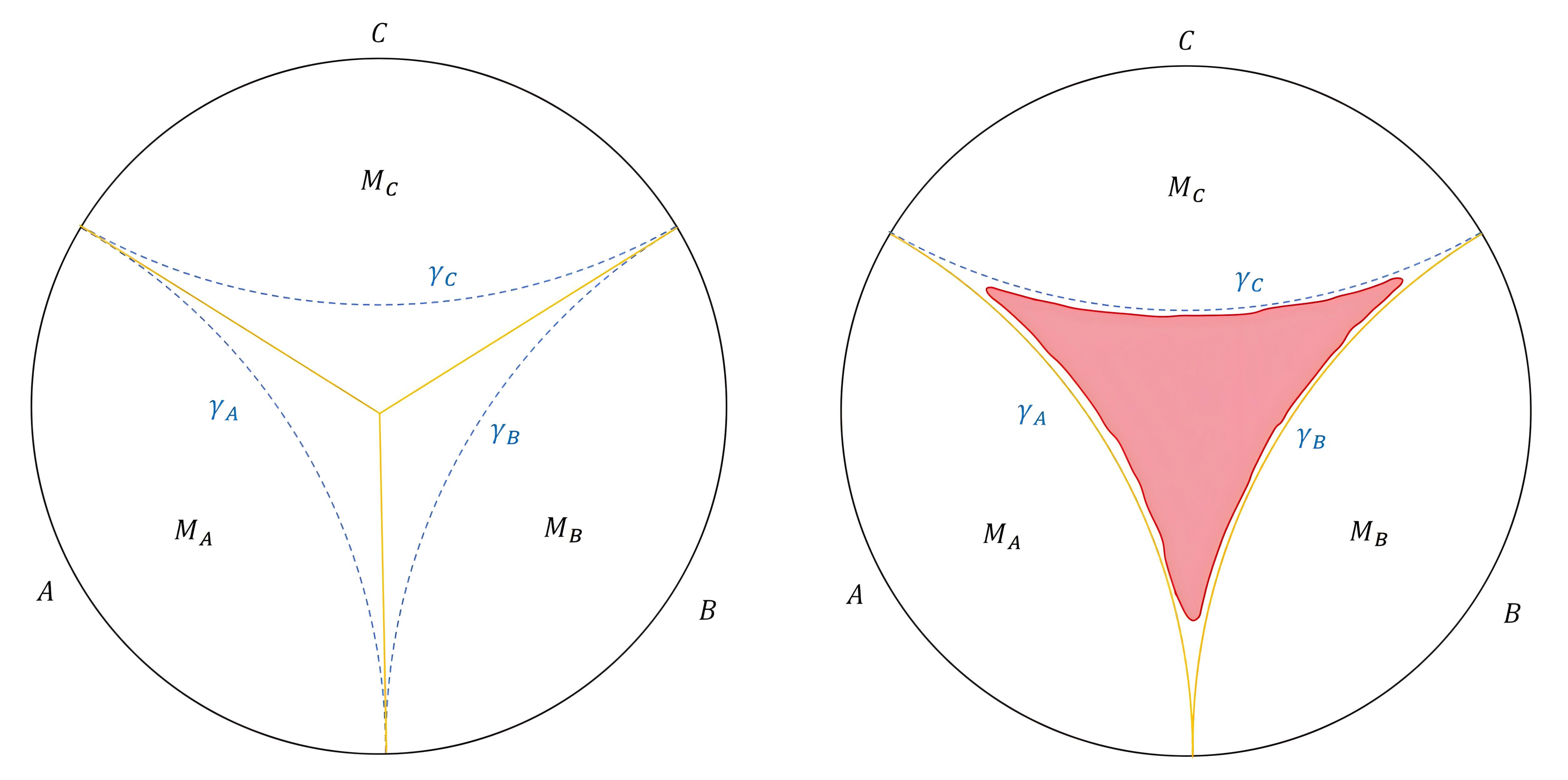} 
\caption{$\mathcal{W}$ in the AdS vacuum (left figure) and the modified spherical geometry (right figure).  Similarly, we modify the geometry only in the region outside the entanglement wedges $M_A$, $M_B$, and $M_C$. As the spherical geometry gradually fills this region, $S^{(3)}(A:B:C)$ becomes the smallest among $S_A+S_B$, $S_B+S_C$, and $S_A+S_C$.} 
\label{MEvacuumspherical} 
\end{figure}
\noindent We first examine the behavior of the multi-entropy under spherically IR modified geometries. Since in the spherical extremal case, geodesics cannot penetrate the IR region, the length of $\mathcal{W}$ increases monotonically as the chosen IR region grows. As shown in the right figure of figure \ref{MEvacuumspherical}, once the modified IR region nearly fills the area enclosed by the three minimal surfaces, $S^{(3)}(A:B:C)$ reaches its upper bound value under the condition that the density matrices of $A$, $B$, and $C$ remain unchanged.

Therefore, $\kappa$ saturates the upper bound value for this quantum marginal problem as well. On the other hand, inequality (\ref{inequality upsilon}) is saturated so that  $\Upsilon$ saturates its lower bound  value in this case.

\subsection{Multi-entropy for the modified hyperbolic geometry}\label{MEhyperbolic}
\begin{figure}[h] 
\centering 
\includegraphics[width=0.5\textwidth]{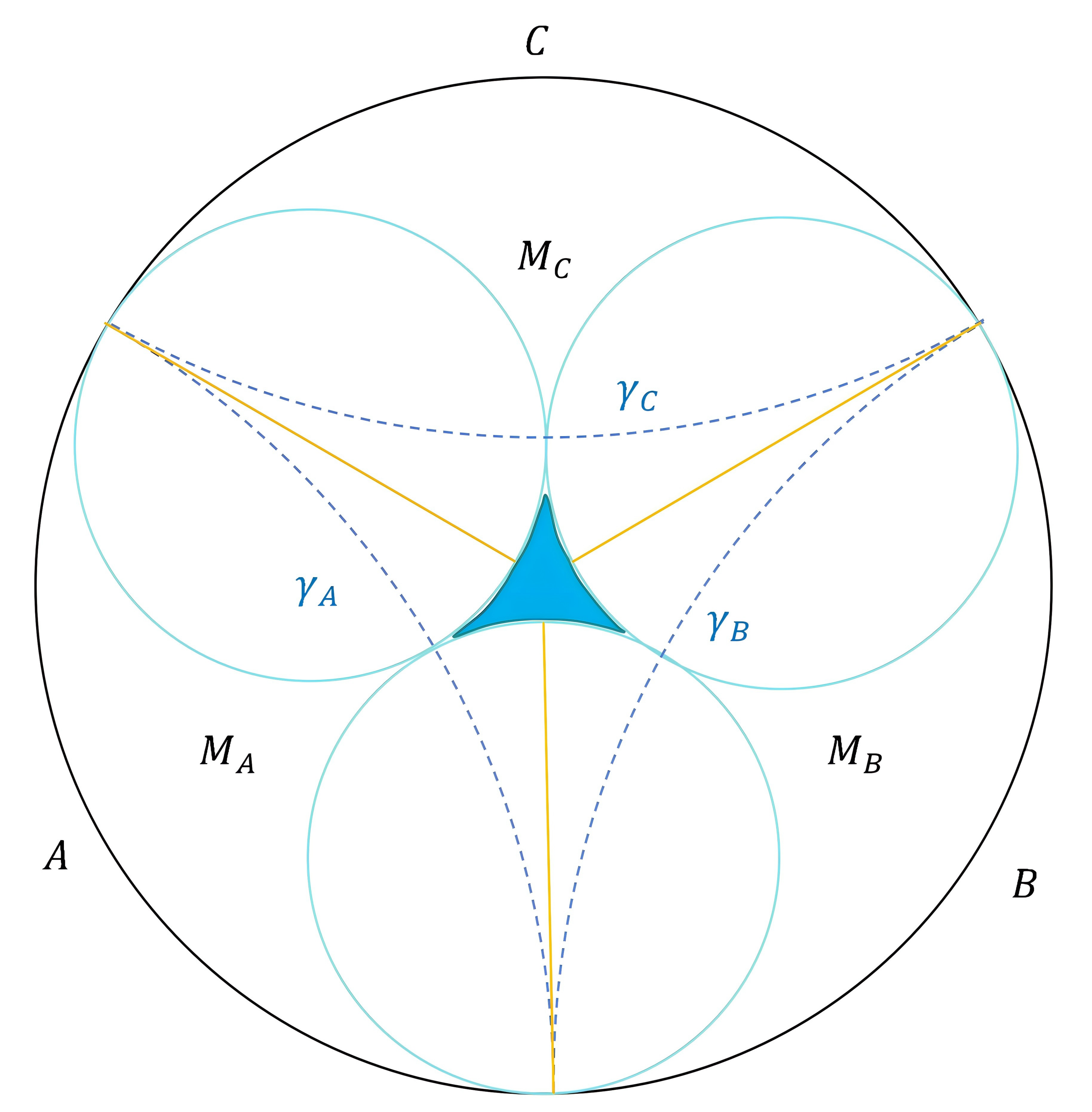} 
\caption{$\mathcal{W}$ in the modified hyperbolic geometry. The three horospheres are pairwise tangent, and the hyperbolic geometric region is taken to be the IR region enclosed by these three circles. In this case, we have $S^{(3)}(A:B:C) = \frac{1}{2}(S_A+S_B+S_C)$}. 
\label{MEvacuumhyperbolic} 
\end{figure}
\noindent Similar to Section \ref{section EWCShyperbolic}, we choose the IR region as the area enclosed  by the three horospheres tangent to the boundaries of subregions $A$, $B$ and $C$, respectively. Since the geodesic length inside the IR region vanishes in this case, $\mathcal{W}$ takes the shape of the yellow web as shown in figure \ref{MEvacuumhyperbolic}. Accordingly, we have
\begin{equation}
    \kappa \equiv S^{(3)}(A:B:C) - \frac{1}{2}(S_A+S_B+S_C) = 0
\end{equation}

According to the convention adopted in    \cite{iizuka2025GM, iizuka2025MGM} (see footnote \ref{f6}), we may conclude that, under the hyperbolic IR-modified geometry, there is no genuine tripartite entanglement among $A$, $B$, and $C$. This is consistent with our earlier conclusion based on the Markov gap: the boundary state dual to the hyperbolic IR-modified geometry is a triangle state, for which $\kappa$ naturally vanishes.

\subsection{Physical {interpretation} of the entanglement measure $\Upsilon$}\label{entanglement structure}
\noindent In this section, building on our understanding of the entanglement structure of boundary quantum states under different geometries, we propose a conjecture regarding the physical meaning of $\Upsilon$. Without loss of generality, we assume that $S_A+S_B$ is the least among $S_A+S_B$, $S_A+S_C$, and $S_B+S_C$ in the following discussion. It should be noted that
\begin{equation}
    \Upsilon + \kappa = \frac{1}{2} I(A:B).
\end{equation}
Considering the hyperbolic extremal case where the boundary state is a triangle state, $\kappa$ vanishes, implying that $\Upsilon$ coincides with the mutual information. Moreover, in a triangle state we have $I(A:B)=I(A_L:B_R)$, since
\begin{equation}
\begin{aligned}
    I(A:B) &= S_A+S_B-S_{AB}\\
    &= S_A + S_B - S_C\\
    &= S_{A_L}+S_{A_R}+S_{B_L}+S_{B_R}-S_{C_L}-S_{C_R}\\
    &= S_{A_L}+S_{B_R} = I(A_L:B_R).
\end{aligned}
\end{equation}
Here we have used $S_A=S_{A_L}+S_{A_R}$, $S_B=S_{B_L}+S_{B_R}$, $S_C=S_{C_L}+S_{C_R}$, $S_{C_L}=S_{A_R}$, and $S_{C_R}=S_{B_L}$. 
Therefore, in this case, we can say that $\Upsilon = \frac{1}{2}I(A_L:B_R)$ measures the weakest genuine bipartite entanglement among the three subsystems $A$, $B$, and $C$.  
In the general case, the subsystems $A$, $B$, and $C$ may share genuine tripartite entanglement. To quantify the weakest bipartite entanglement among them, one must subtract the contribution from tripartite entanglement in $I(A:B)$. For instance, $\Upsilon$ vanishes for the GHZ state because all of $I(A:B)$ comes from genuine tripartite entanglement. Although purely GHZ-like entanglement is forbidden in holographic states    \cite{balasubramanian2025,Hayden_2021}, we suspect similar behaviors to occur in holographic cases, making it natural to remove this type of contribution. Moreover, note that in the general case where $\kappa \ne 0$,
\begin{equation}
\Upsilon = \frac{1}{2}I(A:B) - \kappa.
\end{equation}
Hence, we conjecture that $\Upsilon$ continues to serve as a measure of the weakest bipartite entanglement among $A$, $B$, and $C$, even in the general case. An illustration for this interpretation is shown in figure \ref{upentanglement}.

\begin{figure}[h] 
\centering 
\includegraphics[width=0.95\textwidth]{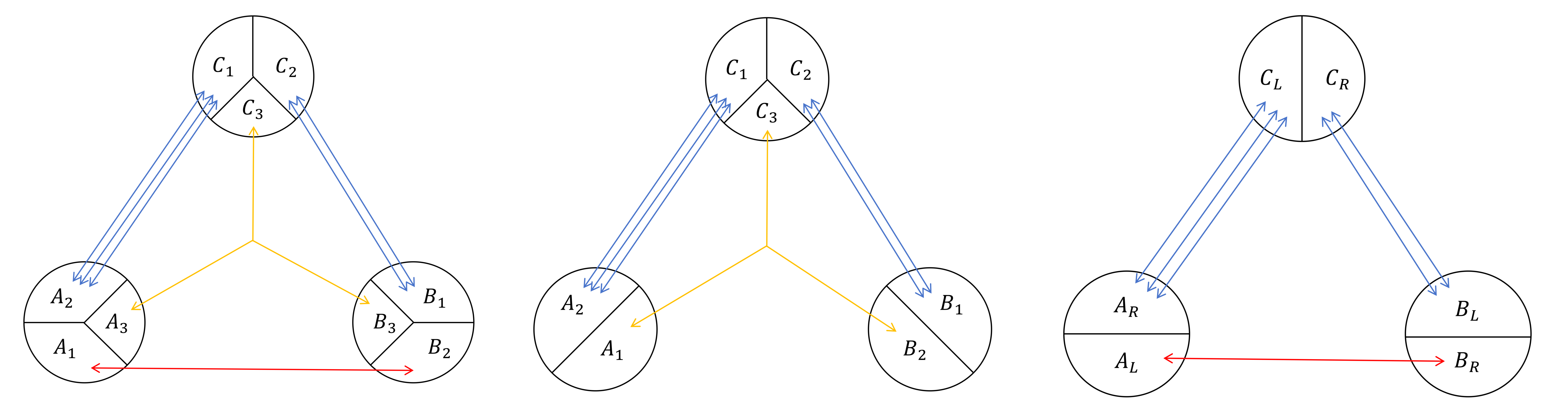} 
\caption{The multipartite entanglement structures of the boundary CFT under different geometries: the left, middle, and right figures correspond to the pure AdS, the spherically IR modified geometry, and the hyperbolic one, respectively. The yellow arrows in the figure represent genuine tripartite entanglement, which vanishes in the extremal case of the modified hyperbolic geometry. The blue arrows denote the relatively stronger bipartite entanglement shared among $A$, $B$, and $C$, while the weakest bipartite entanglement is measured by $\Upsilon$, shown as the red arrows. In the extremal case of the modified spherical geometry, $\Upsilon$ reduces to zero.} 
\label{upentanglement} 
\end{figure}

\section{Conclusion and Discussion}\label{sec7}
\noindent In this work, we investigate the relation between the holographic multipartite entanglement structures and the bulk geometry, especially focusing on the relation between the long scale multipartite entanglement and the bulk IR geometry.
We employ the two toy models proposed in    \cite{Ju_2024} to change the IR geometries in two opposite ways in an AdS$_3$ geometry and study the behavior of various multipartite entanglement measures in the modified geometries. Making use of a range of well-studied measures, we are able to determine the multipartite entanglement structure of the boundary field theory in the two modified geometries. 
Moreover, starting from the entanglement structure revealed by the modified geometries, we can also learn more about what type of entanglement structures is captured by the measures that we have considered.
 
First, we examine the variation of the EWCS in the two opposite types of IR modified geometries. We analyze two multipartite entanglement measures related to the bipartite EWCS, namely the L-entropy and the Markov gap, and find that they exhibit distinctly different behaviors. This strongly suggests that they detect different types of tripartite entanglement: the Markov gap detects non-SOTS type tripartite entanglement while the L-entropy can detect the SOTS type entanglement. The Markov gap increases in the spherically IR modified geometry while it decreases in the hyperbolic case. This is consistent with the behavior found for CMI in    \cite{Ju_2024}, \ie, non-SOTS type tripartite entanglement increases at long scales in the spherically IR modified geometry while decreasing in the hyperbolic case.
 
Specifically, we  modify the geometry in the largest possible IR region while ensuring that the entanglement wedges of specific boundary subregions remain unchanged, thereby preserving the density matrices before and after the geometric modification. In this case, we find that the EWCS can reach the upper bound for this quantum marginal problem subject to the constraint that certain boundary density matrices remain invariant in the extremal case of the spherically IR modified  geometry. On the contrary, in the extremal case of the hyperbolic modification, the EWCS can achieve the corresponding minimal value. Moreover, we find that in the modified hyperbolic geometry it is possible for the Markov gap to vanish, which implies that the boundary quantum state is a triangle state.

Next, we investigate the behavior of multipartite EWCS in modified IR geometries and compute two multipartite entanglement signals, $\Delta_{w}^{(3)}$, as well as $g(A:B:C)$, a multipartite generalization of the Markov gap, for the two types of IR modified geometries. From the behavior of $\Delta_{w}^{(3)}$, we observe that multipartite entanglement increases relative to the pure AdS in the spherically IR modified geometry, while we can make it vanish in the hyperbolic IR modified geometry. Through quantum information analysis, we demonstrate that in the modified hyperbolic geometry, one can obtain boundary quantum states in which multiple subregions share only bipartite entanglement without exhibiting multipartite entanglement. This important result aligns with our expectations, as previous studies have shown that under modified hyperbolic geometry, long-range entanglement is transferred to the critical length. Consequently, it is possible that only adjacent subregions share bipartite entanglement in the hyperbolic case.

Finally, we investigate the behavior of the multi-entropy under the two types of modified geometries and analyze two associated entanglement measures, $\kappa$ and $\Upsilon$. The physical meaning of $\kappa$ is relatively clear, as it is regarded as a measure of {genuine tripartite} entanglement. We find that in the spherical and hyperbolic IR geometries, the multi-entropy reaches its upper and lower bounds, respectively, with $\Upsilon$ and $\kappa$ vanishing respectively. In particular, in the hyperbolic case, $\kappa$ can be made to vanish, once again confirming the feature that hyperbolic deformations truncate tripartite entanglement at longer scales. We also find a physical interpretation for $\Upsilon$ from its behavior in different geometries.

We summarize in table \ref{tab:summary} the behavior of the main entanglement measures discussed in this work in the two IR modified geometries. It can be observed that the L-entropy exhibits behavior opposite to that of other tripartite entanglement measures\footnote{It should be noted that $\Upsilon$ is not regarded as a tripartite entanglement measure. Although it exhibits the same behavior as the L-entropy in the two types of IR-modified geometries, it does not probe the same type of entanglement as the L-entropy.}, suggesting that it likely probes a completely different type of entanglement. In conclusion, changing the IR geometry serves as valuable tools for testing and clarifying the physical meaning of various entanglement measures, and could reveal the relation between bulk geometry and boundary entanglement structures. In this setup, considering more possible configurations would give us more information on this relation and we will leave it for future work.
\begin{table}[h]
\centering
\begin{tabular}{|c|c|c|c|}
\hline
 & Measure & Spherical & Hyperbolic \\
\hline
\multirow{2}{*}{from EWCS} 
 & L-entropy   & $\downarrow$ & $\uparrow$ \\
 & Markov gap  & $\uparrow$   & $\downarrow$ \\
\hline
\multirow{2}{*}{from multi-EWCS} 
 & $\Delta_{w}^{(3)}$         & $\uparrow$ & $\downarrow$ \\
 & $g(A:B:C)$ & $\uparrow$ & $\downarrow$ \\
\hline
\multirow{2}{*}{from multi-entropy} 
 & $\kappa$   & $\uparrow$ & $\downarrow$ \\
 & $\Upsilon$ & $\downarrow$ & $\uparrow$ \\
\hline
\end{tabular}
\caption{The change for the values of the entanglement measures considered in this work in the two types of IR modified geometries. In all cases, the boundary subregions are chosen to be simply connected, adjacent, and of equal size. We find that the L-entropy exhibits behavior opposite to that of the other tripartite (or multipartite) entanglement measures, reflecting the different entanglement structures that it detects.}
\label{tab:summary}
\end{table}

\section*{Acknowledgement}

We thank Bart{\l}omiej Czech, Xiantong Chen, Xuanting Ji, and Wen-Bin Pan for their valuable discussions.
This project is supported in part by the National Natural Science Foundation of China under Grant No. 12035016.

\appendix

\section{Extremal values of EWCS with fixed entanglement wedges for five subregions}\label{A}
\noindent In this appendix, we provide the derivation of equation (\ref{simplified result}), which corresponds to the maximal possible variation of EWCS under the condition that the entanglement wedges of certain pairs of adjacent boundary subregions remain unchanged. We first present and prove two lemmas that serve to simplify the following calculation.
\begin{figure}[h] 
\centering 
\includegraphics[width=0.95\textwidth]{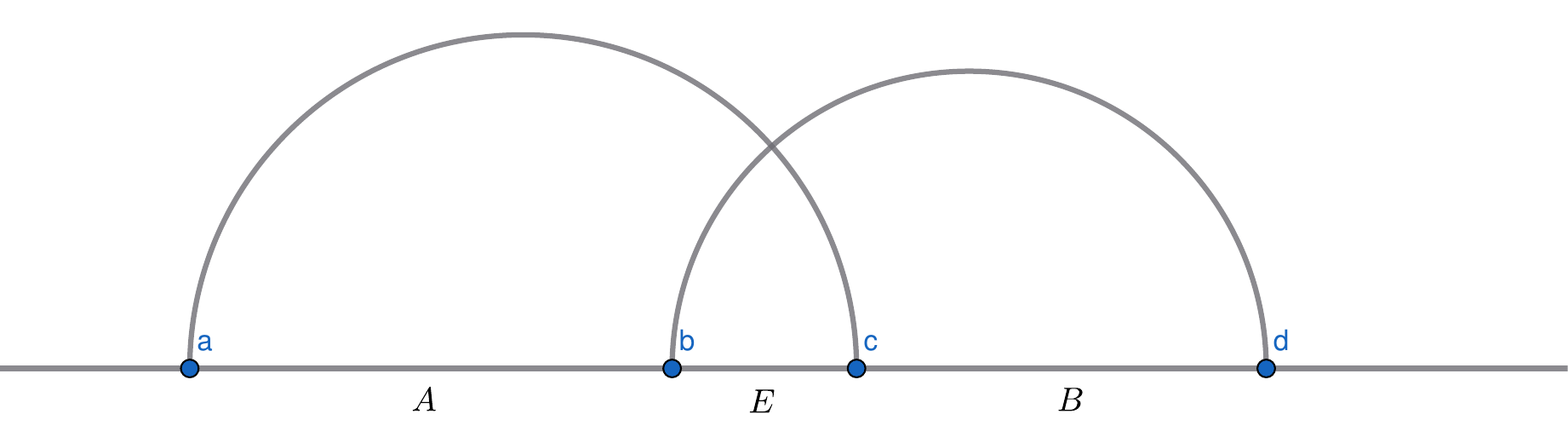} 
\caption{Illustration of Lemma \ref{lemma2}. The minimal surfaces homologous to $AE$ and $BE$ are perpendicular to each other.} 
\label{Fig:lemma2} 
\end{figure}
\begin{lemma}\label{lemma2}
As shown in figure \ref{Fig:lemma2}, if the minimal surfaces homologous to the boundary subregions $AE$ and $BE$ are perpendicular, then the entanglement wedges of $A$ and $B$ are at the critical point between being connected and disconnected, we have
\begin{equation}
    S_{A} + S_{B} = S_{ABE} + S_{E}.
\end{equation}
In the Poincaré disk model, this is equivalent to
\begin{equation}
    L_AL_B=L_{ABE}L_E.
\end{equation}
\end{lemma}

\begin{lemma}\label{lemma3}
As shown in figure \ref{Fig:lemma3}, if the minimal surfaces homologous to $AE_1B$, $E_1BE_2$ and $BE_2C$ intersect at the same point, then the entanglement wedge of $ABC$ is at the critical point between being completely disconnected and completely connected. We have
\begin{equation}
    S_{A} + S_{B} +S_C = S_{AE_1BE_2C} + S_{E_1}+S_{E_2},
\end{equation}
and this is equivalent to 
\begin{equation}
    L_AL_BL_C=L_{AE_1BE_2C}L_{E_1}L_{E_2}
\end{equation}
in the Poincaré disk model.
\end{lemma}

\begin{figure}[h] 
\centering 
\includegraphics[width=0.95\textwidth]{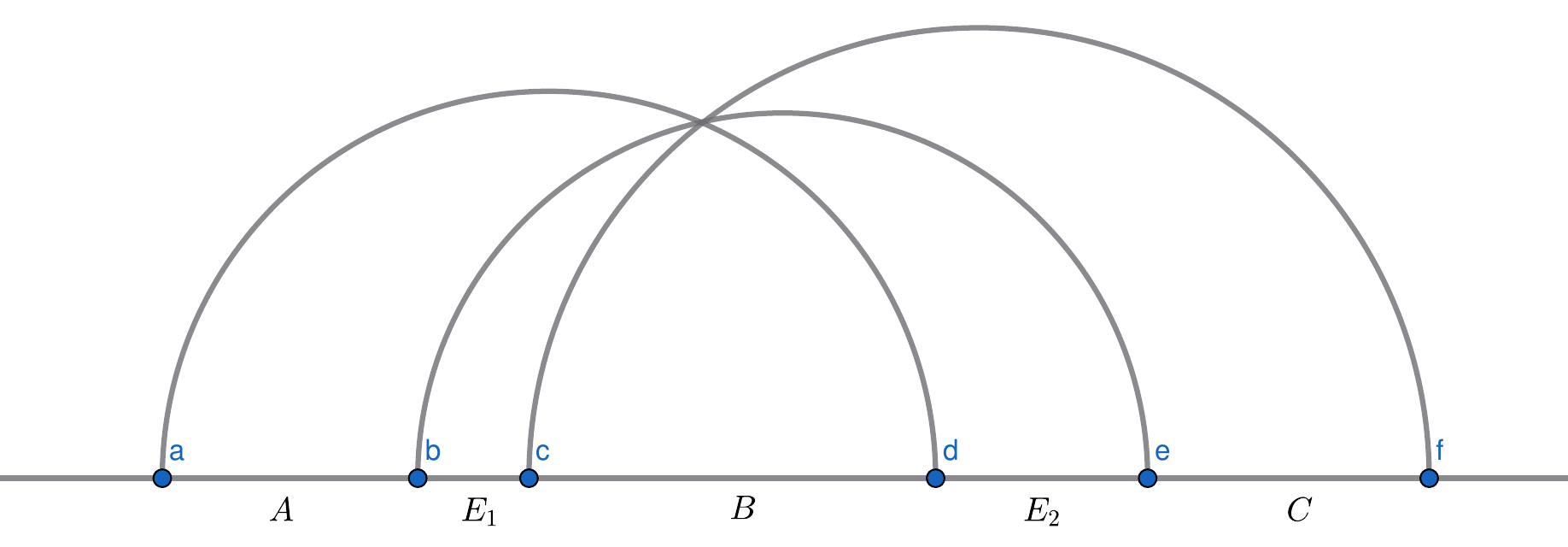} 
\caption{Illustration of Lemma \ref{lemma3}. The minimal surfaces homologous to $AE_1B$, $E_1BE_2$, and $BE_2C$ intersect at the same point.} 
\label{Fig:lemma3} 
\end{figure}
We can exploit the conformal symmetry of hyperbolic geometry to prove these two lemmas. For Lemma \ref{lemma2}, we can always perform a special conformal transformation such that subsystems $A$ and $E$ have the same length (as shown in figure \ref{lemma2proof}). Since conformal transformations preserve angles, $\gamma_{BE}$ remains perpendicular to $\gamma_{AE}$, and therefore it must be a straight line perpendicular to the boundary. The right boundary of subregion $B$ is located at infinity, and in this case it is evident that 
\begin{equation}\label{lemma2equation}
    S_A+S_B-S_E-S_{ABE}=0.
\end{equation}

\begin{figure}[h] 
\centering 
\includegraphics[width=0.95\textwidth]{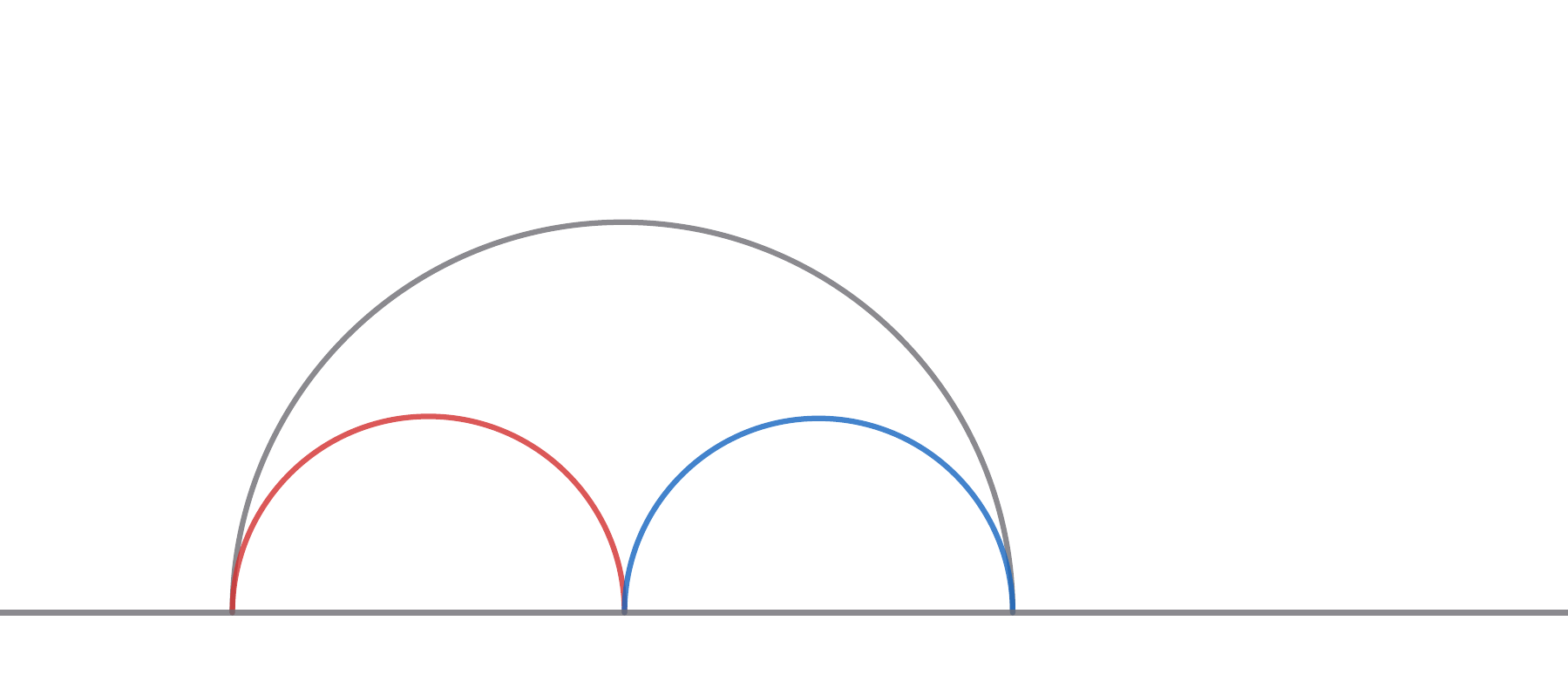} 
\caption{Schematic diagram of the proof of Lemma \ref{lemma2}. After applying a special conformal transformation, the right endpoint of subregion $B$ is mapped to infinity, and subregions $A$ and $E$ are of equal size. The curves $\gamma_{ABE}$, $\gamma_{BE}$, and $\gamma_{B}$ are straight lines perpendicular to the conformal boundary. At this point, the sum of the lengths of the red curves equals the sum of the lengths of the blue curves.} 
\label{lemma2proof} 
\end{figure}

\begin{figure}[h] 
\centering 
\includegraphics[width=0.95\textwidth]{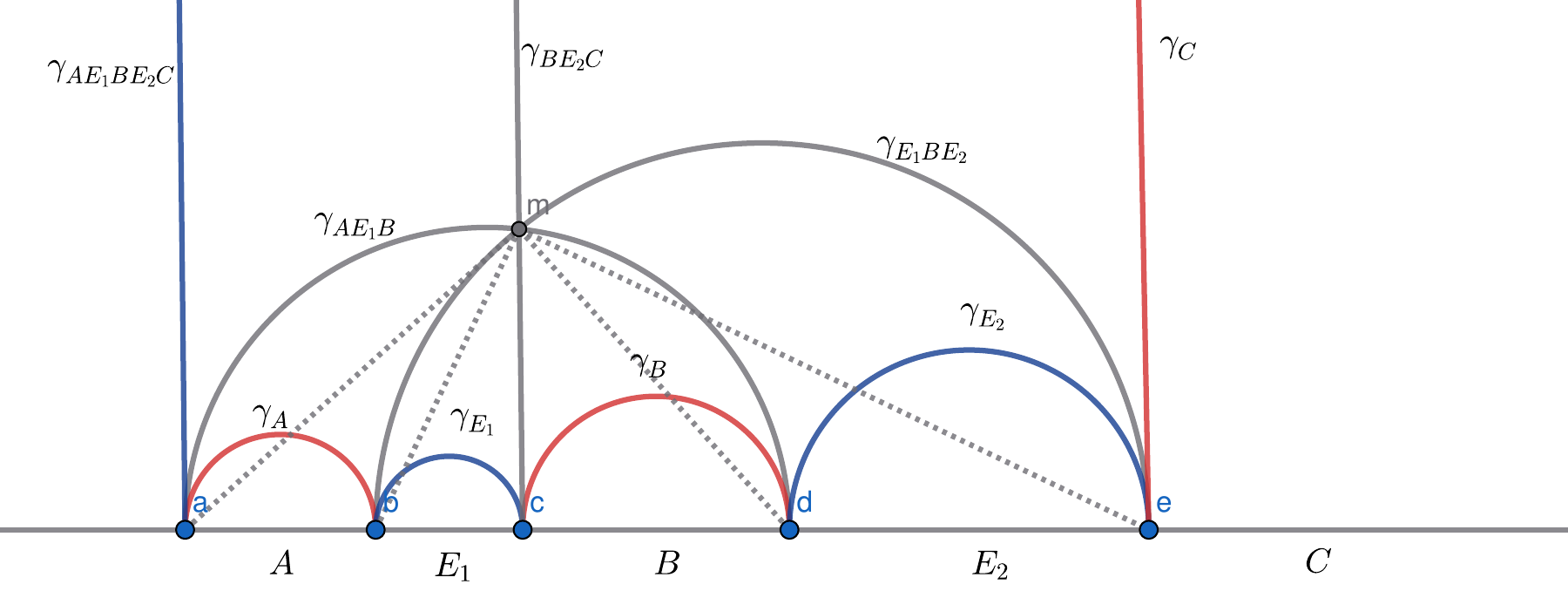}
\caption{Schematic diagram of the proof of Lemma \ref{lemma3}. After applying a special conformal transformation, the right endpoint of subregion $C$ is mapped to infinity, and the curves $\gamma_{AE_1BE_2C}$, $\gamma_{BE_2C}$, and $\gamma_{C}$ become straight lines perpendicular to the conformal boundary. It can be shown that the sum of the lengths of the red curves equals the sum of the lengths of the blue curves.} 
\label{lemma3proof} 
\end{figure}
For Lemma \ref{lemma3}, we can similarly apply a special conformal transformation to map the right boundary of subregion $C$ to infinity. Denoting by $m$ the common intersection point of the three minimal surfaces $\gamma_{AE_1B}$, $\gamma_{E_1BE_2}$ and $\gamma_{BE_2C}$, we can construct similar triangles by connecting $ma$, $mb$, $md$, and $me$. Using this elementary geometric method, we obtain
\begin{equation}
\frac{L_{E_1}}{\overline{mc}}=\frac{\overline{mc}}{L_B+L_{E_2}}
\end{equation}
and
\begin{equation}
\frac{L_{B}}{\overline{mc}}=\frac{\overline{mc}}{L_A+L_{E_1}}.
\end{equation}
Where $\overline{mc}$ is the Euclidean distance between point $m$ and point $c$. Combining these two relations we can derive that
\begin{equation}
    L_AL_B=L_{E_1}L_{E_2}
\end{equation}
Moreover, in the Poincaré half-plane model, the length of a geodesic is proportional to the logarithm of the length of the boundary subregion homologous to it, so we have
\begin{equation}\label{lemma3equation}
    S_A+S_B+S_C-S_{E_1}-S_{E_2}-S_{AE_1BE_2C}=0.
\end{equation}
Since special conformal transformations do not alter the values of infrared quantities, equation (\ref{lemma2equation}) and equation (\ref{lemma3equation}) still hold prior to the special conformal transformation. Thus, we have completed the proof of Lemma \ref{lemma2} and Lemma \ref{lemma3}.

We next introduce a trick for computing the length of a geodesic segment.
\begin{figure}[h] 
\centering 
\includegraphics[width=0.80\textwidth]{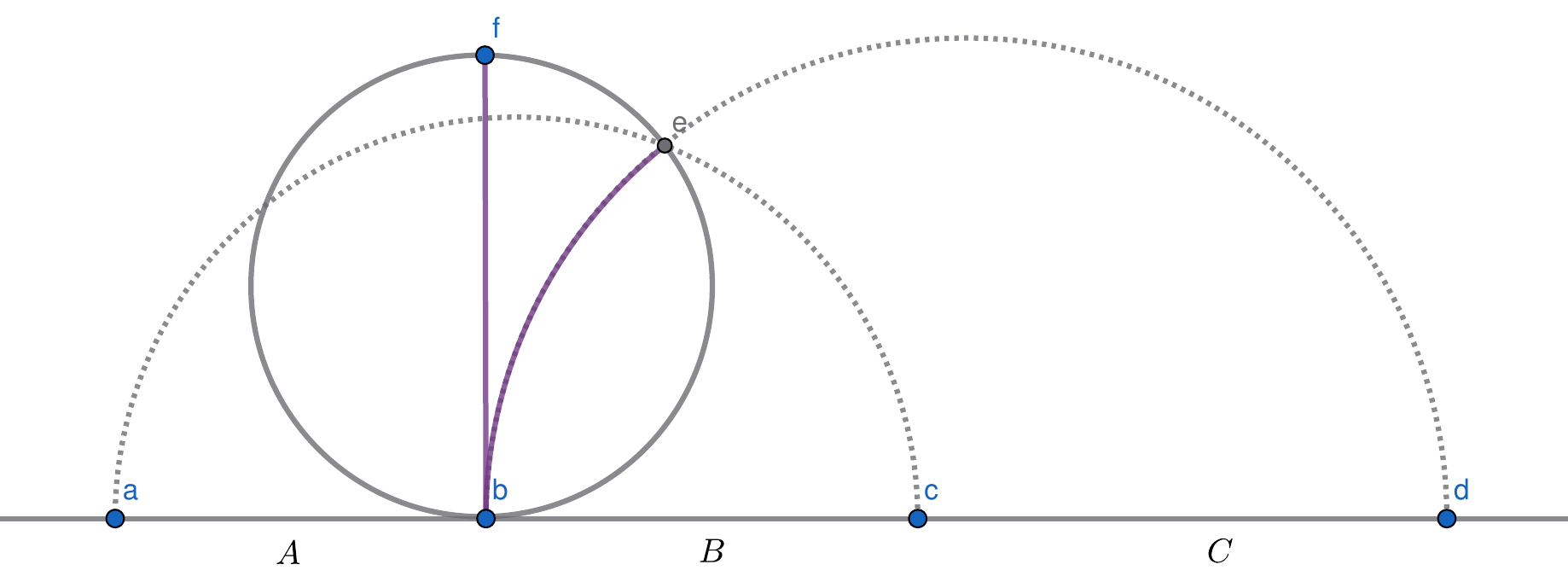} 
\caption{The horosphere is tangent at point $b$ and passes through point $e$. The purple curve $bf$ is the "diameter" of this horosphere. Due to the properties of the horosphere, the two purple geodesics $\overset{\frown}{bf}$ and $\overset{\frown}{be}$ have the same length. } 
\label{geodesic segment} 
\end{figure}
As shown in figure \ref{geodesic segment}, the computation of the length of $\overset{\frown}{be}$ can be reduced to that of $\overset{\frown}{bf}$, while the length of $\overset{\frown}{bf}$ 
\begin{equation}
    \overset{\frown}{bf} = \int_{\epsilon}^{d}\frac{dz}{z} = \log\frac{d}{\epsilon}
\end{equation}
Here $d$ is the Euclidean length of $\overset{\frown}{bf}$, $\epsilon$ is the UV cutoff. From elementary geometry, we can obtain that  
\begin{equation}
    d = (L_B+L_C)\sqrt{\frac{L_AL_B}{L_C(L_A+L_B+L_C)}}.
\end{equation}
With the help of Lemma \ref{lemma2}, Lemma \ref{lemma3} and this trick, we can compute the variation of the $E_W(AB:DE)$ before and after the geometric modification.

\begin{figure}[h] 
\centering 
\includegraphics[width=0.5\textwidth]{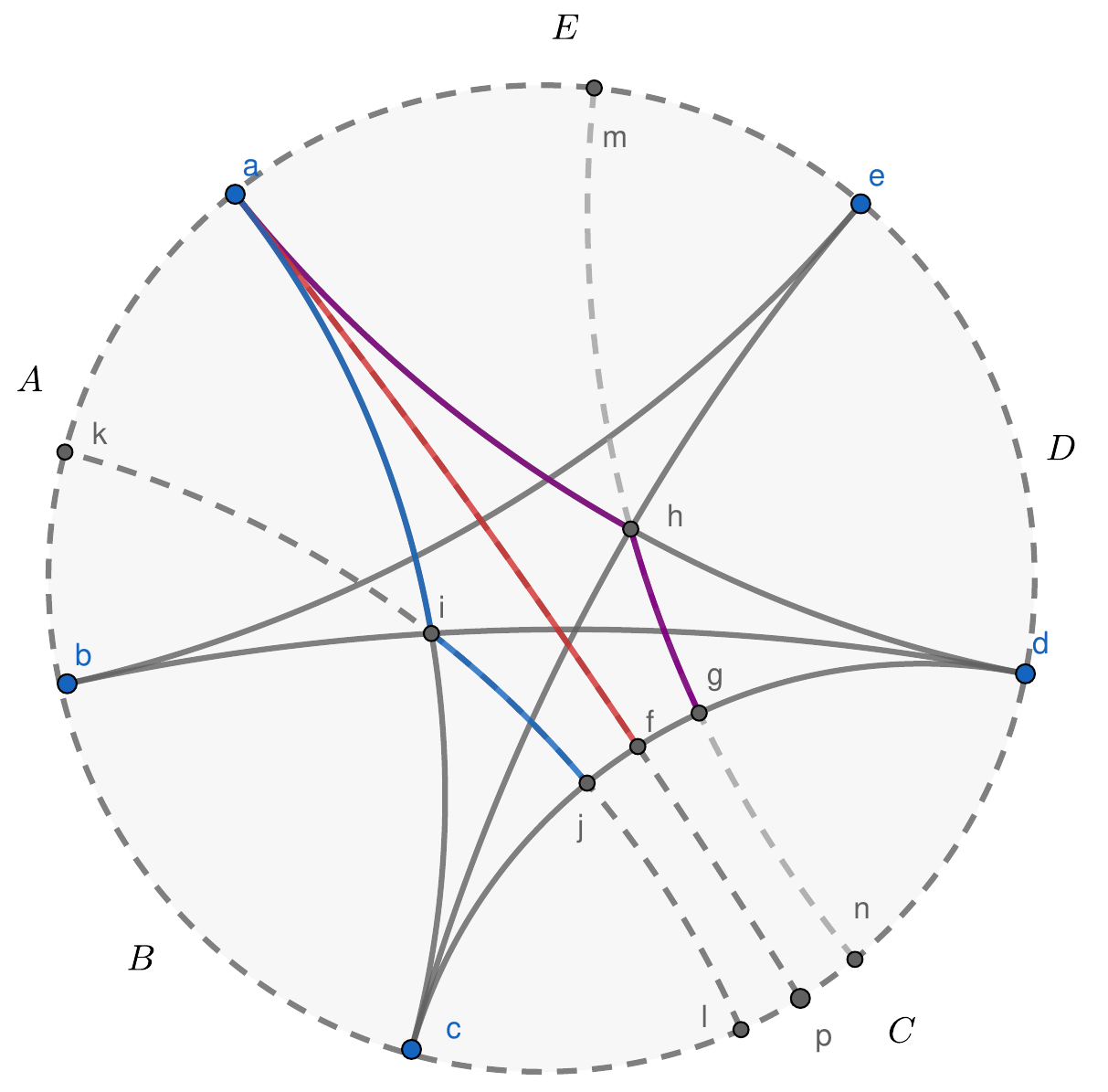} 
\caption{The shape of $\gamma'_{AB,DE}$ before and after the geometric modification in the Poincaré disk. We choose the IR region as the purely IR region enclosed by the five minimal surfaces in the bulk. Here, the red curves represent $\gamma'_{AB,DE}$ before the geometric modification, while the shorter of the purple and blue curves represents $\gamma'_{AB,DE}$. The geodesics $\overset{\frown}{af}$, $\overset{\frown}{hg}$ and $\overset{\frown}{ij}$ are both perpendicular to the geodesic $\overset{\frown}{cd}$. } 
\label{pentagonDisk}
\end{figure}
As shown in figure \ref{pentagonDisk}, before the geometric modification, $\gamma_{AB,DE}$ is the red curve perpendicular to the geodesic $\gamma_C$; after the modification, $\gamma_{AB,DE}$ becomes the shorter of the purple and blue curve. By applying Lemma \ref{lemma2} and Lemma \ref{lemma3}, we can determine the positions of these geodesics and subsequently compute their lengths. In the following discussion we assume that the length of $\overset{\frown}{gh}+\overset{\frown}{ah}$ is less than that of $\overset{\frown}{ai}+\overset{\frown}{ij}$.
\begin{figure}[H] 
\centering 
\includegraphics[width=0.95\textwidth]{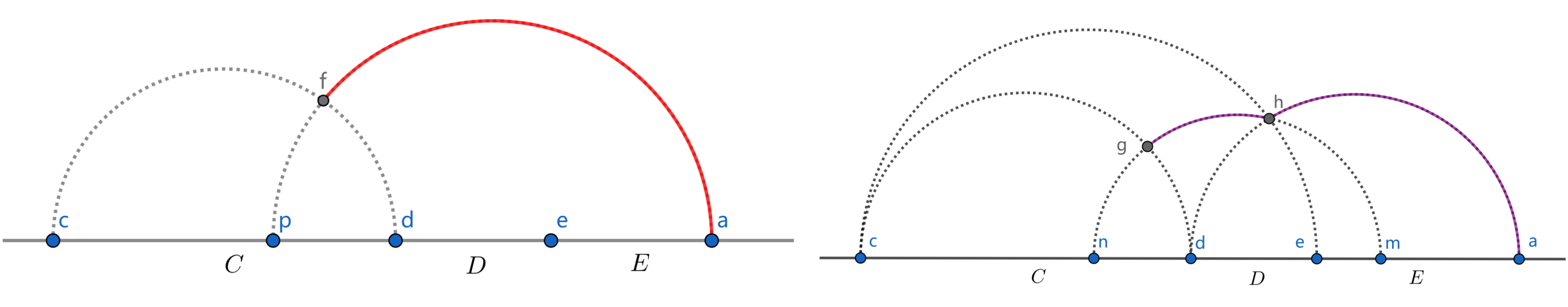} 
\caption{The shape of $\gamma'_{AB,DE}$ before and after the geometric modification in the Poincaré half plane. We assume that the geodesic $\overset{\frown}{af}$, upon extension, intersects the conformal boundary at point $p$, while the geodesic $\overset{\frown}{gh}$, upon extension, intersects the conformal boundary at points $m$ and $n$.} 
\label{pentagonHalfplane} 
\end{figure}
Since geodesic $\overset{\frown}{gh}$ is perpendicular to $\gamma_C$, and it intersects $\gamma_{CD}$ and $\gamma_{DE}$ at the same point $h$, we have
\begin{equation}
\overline{cn}\cdot\overline{dm}=\overline{cm}\cdot\overline{nd}
\end{equation}
and
\begin{equation}
\overline{cn}\cdot\overline{de}\cdot\overline{ma}=\overline{ca}\cdot\overline{nd}\cdot\overline{em}.
\end{equation}
From these two equations, we obtain
\begin{equation}
\overline{em}= \frac{-L_D(L_C+L_D)+\sqrt{L_D(L_C+L_D)(L_D+L_E)(L_C+L_D+L_E)}}{L_C+2L_D+L_E},    
\end{equation}
\begin{equation}
    \overline{cn}=\frac{L_C(L_C+L_D+\overline{em})}{L_C+2\overline{em}+2L_D}.
\end{equation}
This thus determines the positions of points $m$ and $n$. In a similar manner, we can solve for
\begin{equation}
    \overline{cp}=\frac{L_C(L_C+L_D+L_E)}{L_C+2L_D+2L_E}.
\end{equation}
The variation of $E_W(AB:DE)$ is
\begin{equation}
    \Delta E_W(AB:DE)=\frac{1}{4G_N}\left(\overset{\frown}{gh}+\overset{\frown}{ah}-\overset{\frown}{af}\right).
\end{equation}
By employing the trick we introduced earlier, we can express the lengths of these three geodesics as
\begin{equation}
\begin{aligned}
    \overset{\frown}{af} &= 2\log \frac{L_C+L_D+L_E-\overline{cp}}{\epsilon} - \log \frac{d(\overline{cp},L_C-\overline{cp},L_D+L_E)}{\epsilon},\\
    \overset{\frown}{ah} &= 2\log \frac{L_D+L_E}{\epsilon} - \log \frac{d(L_C,L_D,L_E)}{\epsilon},\\
    \overset{\frown}{gh} &= 2\log \frac{L_C+L_D+\overline{em}-\overline{cn}}{\epsilon} - \log \frac{d(\overline{cn},L_C-\overline{cn},L_D+\overline{em})}{\epsilon}\\
    &-\log \frac{d(L_E-\overline{em},L_D+\overline{em},L_C-\overline{cn})}{\epsilon}.
\end{aligned}
\end{equation}
Here we have defined $d(x,y,z)\equiv (y+z)\sqrt{\frac{xy}{z(x+y+z)}}$. The final simplified result is
\begin{equation}
\resizebox{0.90\hsize}{!}{$\displaystyle 
   \Delta E_W(AB:DE) =  \frac{1}{4G_N}\log\left[\frac{1}{2}\sqrt{\frac{L_C(2L_D+L_E)+2(L_D(L_D+L_E)+\sqrt{L_D(L_C+L_D)(L_D+L_E)(L_C+L_D+L_E)})}{L_D(L_C+L_D+L_E)}}\right].
   $}
\end{equation}
Moreover, if the blue curve is shorter than the purple curve, we only have to substitute $L_E$ and $L_D$ to $L_A$ and $L_B$ respectively to get the final result. 

\section{Calculation of lower bound of EWCS under the modified hyperbolic geometry}\label{A2}
\noindent In this appendix, we give the detailed derivation of equation (\ref{results hyperbolic3}). As shown in figure \ref{three_subregion_hyperbolic}, among the three horospheres, we denote by $r$ the radius of the smallest horosphere. Since the other two must be at least tangent to it, the radius of the other two horospheres is at least
\begin{equation}
    \tilde{r} = \frac{3-3r}{r+3} 
\end{equation}
The coordinates of the intersection points of the two larger horospheres are
\begin{equation}
    \left(0, \frac{1-\tilde{r}\pm \sqrt{\tilde{r}^2+6\tilde{r}-3}}{2}\right)
\end{equation}
The length of $\gamma'_{A,B}$ consists of two terms. The first term is 
\begin{equation}
    \log \frac{\sqrt{3} \left(1+\tilde{r}+\sqrt{\tilde{r}^2+6\tilde{r}-3}\right)}{3-\tilde{r}-\sqrt{\tilde{r}^2+6\tilde{r}-3}},
\end{equation}
which goes to zero when three horospheres are pairwise tangent, and the second term is a UV divergent term  
\begin{equation}
    \log \frac{4r}{(1-r)\epsilon}.
\end{equation}
Accordingly, the value of $E_W(A:B)$ is proportional to the sum of these two terms. The expression of $E_W(A:B)$ in terms of $r$ and $\tilde{r}(r)$ is given by
\begin{equation}
    E_W(A:B) = \frac{1}{4G_N}\left[\log \frac{\sqrt{3} \left(1+\tilde{r}+\sqrt{\tilde{r}^2+6\tilde{r}-3}\right)}{3-\tilde{r}-\sqrt{\tilde{r}^2+6\tilde{r}-3}}\frac{4r}{(1-r)\epsilon}\right]
\end{equation}
with $r$ ranging from $0$ to $2\sqrt{3}-3$.

\section{Entanglement structure of quantum states with vanishing multi Markov gap}\label{B}
\noindent This appendix provides the proof of Theorem \ref{thm 2} and Theorem \ref{thm 3} presented in the main text. We begin by noting that the multi EWCS is dual to the multipartite entanglement of purification. The multipartite entanglement of purification for $\rho_{ABC}$ is defined as 
\begin{equation}
    \Delta_P(\rho_{A_1A_2...A_n}) \equiv \min_{\ket{\psi}_{A_1A'_1A_2A'_2...A_nA'_n}}(S_{A_1A'_1}+S_{A_2A'_2}+...+S_{A_nA'_n}).
\end{equation}
where the minimization is taken over all possible purifications of $\rho_{A_1A_2...A_n}$. In addition, we will also make use of a result known as the quantum Markov property   \cite{Hayden2004}, stated in the following theorem.
\begin{theorem}\label{quantum Markov property}
Let $\rho_{ABC}$ be a quantum state on $\mathcal{H}_A\otimes\mathcal{H}_B\otimes\mathcal{H}_C$. Then $I(A:C|B)=0$ if and only if there exists a decomposition of $\mathcal{H}_B$, 
\begin{equation}
    \mathcal{H}_B = \bigoplus_i\mathcal{H}_{B_L}^i\otimes \mathcal{H}_{B_R}^i
\end{equation}
such that
\begin{equation}
    \rho_{ABC}=\sum_i q_i\rho_{AB_L}^i\otimes\rho_{B_RC}^i
\end{equation}
In particular, for a pure state $\ket{\psi}_{ABC}$ satisfying $I(A:C|B)=0$, there exists a bipartition $\mathcal{H}_B=\mathcal{H}_{B_L}\otimes\mathcal{H}_{B_R}$ such that $\ket{\psi}_{ABC}=\ket{\psi}_{AB_L}\otimes \ket{\psi}_{B_RC}$.
\end{theorem}

We first derive that the state $\ket{\psi}_{ABCD}$ can be written in the form
\begin{equation}
    \ket{\psi}_{ABCD}=\ket{\psi}_{A_1D_1}\ket{\psi}_{B_1D_2}\ket{\psi}_{C_1D_3}\ket{\psi}_{A_2B_2C_2}
\end{equation}
from the condition $g(A:B:C)=0$. Consider the optimal purification $\ket{\psi}_{AA'BB'CC'}$ of $\rho_{ABC}$, and rename $A'$, $B'$, and $C'$ as $D_1$, $D_2$, $D_3$, respectively. Using $g(A:B:C)=\Delta(A:B:C)-I(A:B:C)=0$, we obtain
\begin{equation}
S_A+S_B+S_C = S_{AD_1}+S_{BD_2}+S_{CD_3}+S_{ABC}.
\end{equation}
The above expression can be rewritten as
\begin{equation}
I(D_1:BD_2|A) + I(D_2:AC|B) + I(D_1:C|ABD_2) + I(D_3:AD_1BD_2|C) = 0
\end{equation}
or
\begin{equation}
I(D_3:BD_2|C) + I(D_2:AC|B) + I(D_3:A|BCD_2) + I(D_1:BCD_2D_3|A) = 0
\end{equation}
or
\begin{equation}
I(D_1:CD_3|A) + I(D_3:AC|C) + I(D_1:C|ACD_3) + I(D_2:AD_1CD_3|B) = 0.
\end{equation}
Since conditional mutual informations are nonnegative, each of them must vanish. That is 
\begin{equation}
I(D_3:AD_1BD_2|C) = I(D_1:BCD_2D_3|A) = I(D_2:AD_1CD_3|B) = 0.
\end{equation}
Moreover, conditional mutual information is monotonic under the partial trace of a subsystem. Indeed,
\begin{equation}
\begin{aligned}
I(A:BC|E)-I(A:B|E)&=S_{AE}+S_{BCE}-S_{ABCE}-S_{E}-S_{AE}-S_{BE}+S_{ABE}+S_E\\
&=S_{BCE}+S_{ABE}-S_{ABCE}-S_{BE}\\
&=I(A:C|BE)\geq 0.
\end{aligned}
\end{equation}
which implies
\begin{equation}
    \begin{aligned}
        0&=I(D_1:BCD_2D_3|A) \geq I(D_1:BD_2C_2|A) \geq 0,\\
        0&=I(D_2:AD_1CD_3|B)\geq I(D_2:A_2C_2|B)\geq 0.
    \end{aligned}
\end{equation}
Using the conditions 
\begin{equation}
    I(D_3:AD_1BD_2|C) = I(D_1:BD_2C_2|A) = I(D_2:A_2C_2|B) = 0
\end{equation}
together with Theorem \ref{quantum Markov property}, We can iteratively decompose $\ket{\psi}_{AD_1BD_2CD_3}$ as
\begin{equation}
\begin{aligned}
&\ket{\psi}_{AD_1BD_2CD_3} = \ket{\psi}_{AD_1BD_2C_2}\ket{\psi}_{C_1D_3} \\
&= \ket{\psi}_{A_2BD_2C_2}\ket{\psi}_{A_1D_1}\ket{\psi}_{C_1D_3} = \ket{\psi}_{A_1D_1} \ket{\psi}_{B_1D_2}\ket{\psi}_{C_1D_3}\ket{\psi}_{A_2B_2C_2}.
\end{aligned}
\end{equation} 
Finally, since all purifications of $\rho_{ABC}$ differ only by a local unitary transformation, $\ket{\psi}_{ABCD}$ can be written in the above form up to a local unitary transformation.

The proof of the converse direction of this theorem is relatively straightforward. If $\ket{\psi}_{ABCD} = \ket{\psi}_{A_1D_1} \ket{\psi}_{B_1D_2}\ket{\psi}_{C_1D_3}\ket{\psi}_{A_2B_2C_2}$, then
\begin{equation}
 \begin{aligned}
g(A:B:C) &\leq \frac{1}{2}(S_{AD_1}+S_{BD_2}+S_{CD_3}-S_A-S_B-S_C+S_{ABC}) \\
&= \frac{1}{2}(S_{A_2}+S_{B_2}+S_{C_2}-S_{A_1}-S_{A_2}-S_{B_1}-S_{B_2}-S_{C_1}-S_{C_2}+S_D) \\
&= 0.
\end{aligned}   
\end{equation}
And since $g(A:B:C)$ is non-negative so it must vanish. Here we have used $S_{AD_1}=S_{A_2}$, $S_{BD_2}=S_{B_2}$, $S_{CD_3}=S_{C_2}$, $S_{A/B/C}=S_{A_1/B_1/C_1}+S_{A_2/B_2/C_2}$, and $S_{ABC}=S_{D}=S_{D_1}+S_{D_2}+S_{D_3}=S_{A_1}+S_{B_1}+S_{C_1}$.

For Theorem \ref{thm 3}, the proof proceeds in a similar manner. If $g(A:B:C:D)=0$, and we consider the optimal purification $\ket{\psi}_{AA'BB'CC'DD'}$ of $\rho_{ABCD}$, renaming $A'$, $B'$, $C'$, and $D'$ as $E_1$, $E_2$, $E_3$ and $E_4$ we have 
\begin{equation}\label{4partiteMarkov}
    S_{A}+S_B+S_C+S_D=S_{AE_1}+S_{BE_2}+S_{CE_3}+S_{DE_4} + S_{ABCD}.
\end{equation}
The above expression can be rewritten as
\begin{equation}
\begin{aligned}
    &I(E_1:BE_2CE_3DE_4|A)+I(E_2:ACD|B)+I(E_3:ABE_2DE_4|C)\\
    &+I(E_4:ABE_2C|D)=0.
\end{aligned}
\end{equation}
By the non-negativity of the conditional mutual information, we obtain
\begin{equation}
I(E_1:BE_2CE_3DE_4|A)=0    
\end{equation}
Moreover, since $A$, $B$, $C$, and $D$ exhibit cyclic symmetry in equation (\ref{4partiteMarkov}), we can also deduce that
\begin{equation}
\begin{aligned}
    I(E_2:AE_1CE_3DE_4|B)=I(E_3:AE_1BE_2DE_4|C)=I(E_4:AE_1BE_2CE_3|D)=0.
\end{aligned}
\end{equation}
Furthermore, using the monotonicity of the conditional mutual information under partial trace, we obtain
\begin{equation}
\begin{aligned}
    0&=I(E_2:AE_1CE_3DE_4|B)\geq I(E_2:A_2CE_3DE_4|B)\geq 0,\\
    0&=I(E_3:AE_1BE_2DE_4|C)\geq I(E_3:A_2B_2DE_4|C)\geq 0,\\
    0&=I(E_4:AE_1BE_2CE_3|D)\geq I(E_4:A_2B_2C_2|D)\geq 0.
\end{aligned}
\end{equation}
Combining 
\begin{equation}
\begin{aligned}
&I(E_1:BE_2CE_3DE_4|A)=I(E_2:A_2CE_3DE_4|B)\\
&=I(E_3:A_2B_2DE_4|C)=I(E_4:A_2B_2C_2|D)=0
\end{aligned}
\end{equation}
with the quantum Markov property, the state $\ket{\psi}_{AE_1BE_2CE_3DE_4}$ can be successively factorized into
\begin{equation}
    \begin{aligned}
       &\ket{\psi}_{AE_1BE_2CE_3DE_4}=\ket{\psi}_{A_1E_1}\ket{\psi}_{A_2BE_2CE_3DE_4}=\ket{\psi}_{A_1E_1}\ket{\psi}_{B_1E_2}\ket{\psi}_{A_2B_2CE_3DE_4}\\
       &=\ket{\psi}_{A_1E_1}\ket{\psi}_{B_1E_2}\ket{\psi}_{C_1E_3}\ket{\psi}_{A_2B_2C_2DE_4}\\ 
       &= \ket{\psi}_{A_1E_1}\ket{\psi}_{B_1E_2}\ket{\psi}_{C_1E_3}\ket{\psi}_{D_1E_4}\ket{\psi}_{A_2B_2C_2D_2}.
    \end{aligned}
\end{equation}
Conversely, if the state $\ket{\psi}_{ABCDE}$ can be decomposed into the above form, then by a straightforward calculation we obtain
\begin{equation}
\begin{aligned}
    g(A:B:C:D)&\leq \frac{1}{2}(S_{AE_1}+S_{BE_2}+S_{CE_3}+S_{DE_4}-S_A-S_B-S_C+S_{ABC})\\
    &=\frac{1}{2}(S_{A_2}+S_{B_2}+S_{C_2}+S_{D_2}-S_{A_1}-S_{A_2}-S_{B_1}-S_{B_2}\\
    &-S_{C_1}-S_{C_2}-S_{D_1}-S_{D_2}+S_{E_1E_2E_3E_4})=0.
\end{aligned}
\end{equation}
Here we have used $S_{AE_1}=S_{A_2}$, $S_{BE_2}=S_{B_2}$, $S_{CE_3}=S_{C_2}$, $S_{DE_3}=S_{D_2}$ $S_{A/B/C/D}=S_{A_1/B_1/C_1/D_1}+S_{A_2/B_2/C_2/D_2}$, and $S_{ABCD}=S_{E}=S_{E_1}+S_{E_2}+S_{E_3}+S_{E_4}=S_{A_1}+S_{B_1}+S_{C_1}+S_{D_1}$.

\bibliography{main}

@article{Hayden2004,
   author = {Hayden, Patrick and Jozsa, Richard and Petz, Dénes and Winter, Andreas},
   title = {Structure of States Which Satisfy Strong Subadditivity of Quantum Entropy with Equality},
   journal = {Communications in Mathematical Physics},
   volume = {246},
   number = {2},
   pages = {359-374},
   ISSN = {1432-0916},
   DOI = {10.1007/s00220-004-1049-z},
   url = {https://doi.org/10.1007/s00220-004-1049-z},
   year = {2004},
   type = {Journal Article}
}

@article{Christandl2004,
   author = {Christandl, Matthias and Winter, Andreas},
   title = {“Squashed entanglement”: An additive entanglement measure},
   journal = {Journal of Mathematical Physics},
   volume = {45},
   number = {3},
   pages = {829-840},
   ISSN = {0022-2488},
   DOI = {10.1063/1.1643788},
   url = {https://doi.org/10.1063/1.1643788},
   year = {2004},
   type = {Journal Article}
}

@article{Yang2008,
   title={An Additive and Operational Entanglement Measure: Conditional Entanglement of Mutual Information},
   volume={101},
   ISSN={1079-7114},
   url={http://dx.doi.org/10.1103/PhysRevLett.101.140501},
   DOI={10.1103/physrevlett.101.140501},
   number={14},
   journal={Physical Review Letters},
   publisher={American Physical Society (APS)},
   author={Yang, Dong and Horodecki, Michał and Wang, Z. D.},
   year={2008},
   month=sep }

@inproceedings{schilling2014,
    author = "Schilling, Christian",
    title = "{The Quantum Marginal Problem}",
    eprint = "1404.1085",
    archivePrefix = "arXiv",
    primaryClass = "quant-ph",
    month = "4",
    year = "2014"
}

@article{schilling2015,
       author = {{Schilling}, Christian},
        title = "{Quantum Marginal Problem and its Physical Relevance}",
      journal = {arXiv e-prints},
         year = 2015,
        month = jul,
          eid = {arXiv:1507.00299},
        pages = {arXiv:1507.00299},
          doi = {10.48550/arXiv.1507.00299},
archivePrefix = {arXiv},
       eprint = {1507.00299},
 primaryClass = {quant-ph},
       adsurl = {https://ui.adsabs.harvard.edu/abs/2015arXiv150700299S}
}

@article{Umemoto_2018multiEWCS,
   title={Entanglement of purification for multipartite states and its holographic dual},
   volume={2018},
   ISSN={1029-8479},
   url={http://dx.doi.org/10.1007/JHEP10(2018)152},
   DOI={10.1007/jhep10(2018)152},
   number={10},
   journal={Journal of High Energy Physics},
   publisher={Springer Science and Business Media LLC},
   author={Umemoto, Koji and Zhou, Yang},
   year={2018},
   month=oct }

@article{Umemoto_2018,
   title={Entanglement of purification through holographic duality},
   volume={14},
   ISSN={1745-2481},
   url={http://dx.doi.org/10.1038/s41567-018-0075-2},
   DOI={10.1038/s41567-018-0075-2},
   number={6},
   journal={Nature Physics},
   publisher={Springer Science and Business Media LLC},
   author={Umemoto, Koji and Takayanagi, Tadashi},
   year={2018},
   month=mar, pages={573–577} }

@article{Bao:2018fso,
    author = "Bao, Ning and Chatwin-Davies, Aidan and Remmen, Grant N.",
    title = "{Entanglement of Purification and Multiboundary Wormhole Geometries}",
    eprint = "1811.01983",
    archivePrefix = "arXiv",
    primaryClass = "hep-th",
    doi = "10.1007/JHEP02(2019)110",
    journal = "JHEP",
    volume = "02",
    pages = "110",
    year = "2019"
}

@article{Bao_2019,
   title={Conditional and multipartite entanglements of purification and holography},
   volume={99},
   ISSN={2470-0029},
   url={http://dx.doi.org/10.1103/PhysRevD.99.046010},
   DOI={10.1103/physrevd.99.046010},
   number={4},
   journal={Physical Review D},
   publisher={American Physical Society (APS)},
   author={Bao, Ning and Halpern, Illan F.},
   year={2019},
   month=feb }

@article{Bao_2019multiSR,
   title={Multipartite reflected entropy},
   volume={2019},
   ISSN={1029-8479},
   url={http://dx.doi.org/10.1007/JHEP10(2019)102},
   DOI={10.1007/jhep10(2019)102},
   number={10},
   journal={Journal of High Energy Physics},
   publisher={Springer Science and Business Media LLC},
   author={Bao, Ning and Cheng, Newton},
   year={2019},
   month=oct }

@article{dutta2019,
    author = "Dutta, Souvik and Faulkner, Thomas",
    title = "{A canonical purification for the entanglement wedge cross-section}",
    eprint = "1905.00577",
    archivePrefix = "arXiv",
    primaryClass = "hep-th",
    doi = "10.1007/JHEP03(2021)178",
    journal = "JHEP",
    volume = "03",
    pages = "178",
    year = "2021"
}

@article{Xie2021,
  title = {Triangle Measure of Tripartite Entanglement},
  author = {Xie, Songbo and Eberly, Joseph H.},
  journal = {Phys. Rev. Lett.},
  volume = {127},
  issue = {4},
  pages = {040403},
  numpages = {6},
  year = {2021},
  month = {Jul},
  publisher = {American Physical Society},
  doi = {10.1103/PhysRevLett.127.040403},
  url = {https://link.aps.org/doi/10.1103/PhysRevLett.127.040403}
}

@article{Maldacena:1997re,
    author = "Maldacena, Juan Martin",
    title = "{The Large N limit of superconformal field theories and supergravity}",
    eprint = "hep-th/9711200",
    archivePrefix = "arXiv",
    reportNumber = "HUTP-97-A097, HUTP-98-A097",
    doi = "10.4310/ATMP.1998.v2.n2.a1",
    journal = "Adv. Theor. Math. Phys.",
    volume = "2",
    pages = "231--252",
    year = "1998"
}

@article{Ji:2025vks,
    author = "Ji, Xuanting and Ju, Xin-Xiang and Sun, Ya-Wen and Wang, Yuan-Tai and Zhou, He-Lin",
    title = "{Holographic geometry/real-space entanglement correspondence and metric reconstruction}",
    eprint = "2505.08534",
    archivePrefix = "arXiv",
    primaryClass = "hep-th",
    doi = "10.1007/JHEP09(2025)081",
    journal = "JHEP",
    volume = "09",
    pages = "081",
    year = "2025"
}

@article{Ju:2024kuc,
    author = "Ju, Xin-Xiang and Pan, Wen-Bin and Sun, Ya-Wen and Wang, Yuan-Tai and Zhao, Yang",
    title = "{More on the upper bound of holographic n-partite information}",
    eprint = "2411.19207",
    archivePrefix = "arXiv",
    primaryClass = "hep-th",
    doi = "10.1007/JHEP03(2025)184",
    journal = "JHEP",
    volume = "03",
    pages = "184",
    year = "2025"
}

@article{Ju:2024hba,
    author = "Ju, Xin-Xiang and Pan, Wen-Bin and Sun, Ya-Wen and Zhao, Yang",
    title = "{Holographic multipartite entanglement from the upper bound of $n$-partite information}",
    eprint = "2411.07790",
    archivePrefix = "arXiv",
    primaryClass = "hep-th",
    month = "11",
    year = "2024"
}

@book{Rangamani:2016dms,
    author = "Rangamani, Mukund and Takayanagi, Tadashi",
    title = "{Holographic Entanglement Entropy}",
    eprint = "1609.01287",
    archivePrefix = "arXiv",
    primaryClass = "hep-th",
    reportNumber = "YITP-16-106",
    doi = "10.1007/978-3-319-52573-0",
    publisher = "Springer",
    volume = "931",
    year = "2017"
}

@article{Akers_2020,
   title={Entanglement wedge cross sections require tripartite entanglement},
   volume={2020},
   ISSN={1029-8479},
   url={http://dx.doi.org/10.1007/JHEP04(2020)208},
   DOI={10.1007/jhep04(2020)208},
   number={4},
   journal={Journal of High Energy Physics},
   publisher={Springer Science and Business Media LLC},
   author={Akers, Chris and Rath, Pratik},
   year={2020},
   month=apr }

@article{Maldacena:2013xja,
    author = "Maldacena, Juan and Susskind, Leonard",
    title = "{Cool horizons for entangled black holes}",
    eprint = "1306.0533",
    archivePrefix = "arXiv",
    primaryClass = "hep-th",
    doi = "10.1002/prop.201300020",
    journal = "Fortsch. Phys.",
    volume = "61",
    pages = "781--811",
    year = "2013"
}

@article{VanRaamsdonk:2010pw,
    author = "Van Raamsdonk, Mark",
    title = "{Building up spacetime with quantum entanglement}",
    eprint = "1005.3035",
    archivePrefix = "arXiv",
    primaryClass = "hep-th",
    doi = "10.1142/S0218271810018529",
    journal = "Gen. Rel. Grav.",
    volume = "42",
    pages = "2323--2329",
    year = "2010"
}

@article{VanRaamsdonk:2009ar,
    author = "Van Raamsdonk, Mark",
    title = "{Comments on quantum gravity and entanglement}",
    eprint = "0907.2939",
    archivePrefix = "arXiv",
    primaryClass = "hep-th",
    month = "7",
    year = "2009"
}

@article{Swingle:2009bg,
    author = "Swingle, Brian",
    title = "{Entanglement Renormalization and Holography}",
    eprint = "0905.1317",
    archivePrefix = "arXiv",
    primaryClass = "cond-mat.str-el",
    doi = "10.1103/PhysRevD.86.065007",
    journal = "Phys. Rev. D",
    volume = "86",
    pages = "065007",
    year = "2012"
}

@article{Ju:2025tgg,
    author = "Ju, Xin-Xiang and Sun, Ya-Wen and Zhao, Yang",
    title = "{Upper bound of holographic entanglement entropy combinations}",
    eprint = "2505.11059",
    archivePrefix = "arXiv",
    primaryClass = "hep-th",
    doi = "10.1007/JHEP09(2025)085",
    journal = "JHEP",
    volume = "09",
    pages = "085",
    year = "2025"
}

@article{Ryu_2006,
	doi = {10.1103/physrevlett.96.181602},
  
	url = {https://doi.org/10.1103%2Fphysrevlett.96.181602},
  
	year = 2006,
	month = {may},
  
	publisher = {American Physical Society ({APS})},
  
	volume = {96},
  
	number = {18},
  
	author = {Shinsei Ryu and Tadashi Takayanagi},
  
	title = {Holographic Derivation of Entanglement Entropy from the anti{\textendash}de Sitter Space/Conformal Field Theory Correspondence},
  
	journal = {Physical Review Letters}
}

@article{Czech:2012bh,
    author = "Czech, Bartlomiej and Karczmarek, Joanna L. and Nogueira, Fernando and Van Raamsdonk, Mark",
    title = "{The Gravity Dual of a Density Matrix}",
    eprint = "1204.1330",
    archivePrefix = "arXiv",
    primaryClass = "hep-th",
    doi = "10.1088/0264-9381/29/15/155009",
    journal = "Class. Quant. Grav.",
    volume = "29",
    pages = "155009",
    year = "2012"
}

@article{Wall:2012uf,
    author = "Wall, Aron C.",
    title = "{Maximin Surfaces, and the Strong Subadditivity of the Covariant Holographic Entanglement Entropy}",
    eprint = "1211.3494",
    archivePrefix = "arXiv",
    primaryClass = "hep-th",
    doi = "10.1088/0264-9381/31/22/225007",
    journal = "Class. Quant. Grav.",
    volume = "31",
    number = "22",
    pages = "225007",
    year = "2014"
}

@article{Harper:2021uuq,
    author = "Harper, Jonathan",
    title = "{Hyperthreads in holographic spacetimes}",
    eprint = "2107.10276",
    archivePrefix = "arXiv",
    primaryClass = "hep-th",
    reportNumber = "BRX-TH-6688",
    doi = "10.1007/JHEP09(2021)118",
    journal = "JHEP",
    volume = "09",
    pages = "118",
    year = "2021"
}

@article{Harper:2022sky,
    author = "Harper, Jonathan",
    title = "{Perfect tensor hyperthreads}",
    eprint = "2205.01140",
    archivePrefix = "arXiv",
    primaryClass = "hep-th",
    reportNumber = "BRX-TH-6705",
    doi = "10.1007/JHEP09(2022)239",
    journal = "JHEP",
    volume = "09",
    pages = "239",
    year = "2022"
}

@article{Zou:2022nuj,
    author = "Zou, Yijian and Shi, Bowen and Sorce, Jonathan and Lim, Ian T. and Kim, Isaac H.",
    title = "{Modular Commutators in Conformal Field Theory}",
    eprint = "2206.00027",
    archivePrefix = "arXiv",
    primaryClass = "cond-mat.str-el",
    doi = "10.1103/PhysRevLett.129.260402",
    journal = "Phys. Rev. Lett.",
    volume = "129",
    number = "26",
    pages = "260402",
    year = "2022"
}

@article{Hernandez-Cuenca:2023iqh,
    author = "Hern{\'a}ndez-Cuenca, Sergio and Hubeny, Veronika E. and Jia, Hewei Frederic",
    title = "{Holographic entropy inequalities and multipartite entanglement}",
    eprint = "2309.06296",
    archivePrefix = "arXiv",
    primaryClass = "hep-th",
    reportNumber = "MIT-CTP/5610",
    doi = "10.1007/JHEP08(2024)238",
    journal = "JHEP",
    volume = "08",
    pages = "238",
    year = "2024"
}

@article{Hayden_2021,
   title={The Markov gap for geometric reflected entropy},
   volume={2021},
   ISSN={1029-8479},
   url={http://dx.doi.org/10.1007/JHEP10(2021)047},
   DOI={10.1007/jhep10(2021)047},
   number={10},
   journal={Journal of High Energy Physics},
   publisher={Springer Science and Business Media LLC},
   author={Hayden, Patrick and Parrikar, Onkar and Sorce, Jonathan},
   year={2021},
   month=oct }

@article{Zou_2021,
   title={Universal Tripartite Entanglement in One-Dimensional Many-Body Systems},
   volume={126},
   ISSN={1079-7114},
   url={http://dx.doi.org/10.1103/PhysRevLett.126.120501},
   DOI={10.1103/physrevlett.126.120501},
   number={12},
   journal={Physical Review Letters},
   publisher={American Physical Society (APS)},
   author={Zou, Yijian and Siva, Karthik and Soejima, Tomohiro and Mong, Roger S. K. and Zaletel, Michael P.},
   year={2021},
   month=mar }

@article{Gadde_2022,
   title={New multipartite entanglement measure and its holographic dual},
   volume={106},
   ISSN={2470-0029},
   url={http://dx.doi.org/10.1103/PhysRevD.106.126001},
   DOI={10.1103/physrevd.106.126001},
   number={12},
   journal={Physical Review D},
   publisher={American Physical Society (APS)},
   author={Gadde, Abhijit and Krishna, Vineeth and Sharma, Trakshu},
   year={2022},
   month=dec }

@article{Penington_2023,
   title={Fun with replicas: tripartitions in tensor networks and gravity},
   volume={2023},
   ISSN={1029-8479},
   url={http://dx.doi.org/10.1007/JHEP05(2023)008},
   DOI={10.1007/jhep05(2023)008},
   number={5},
   journal={Journal of High Energy Physics},
   publisher={Springer Science and Business Media LLC},
   author={Penington, Geoff and Walter, Michael and Witteveen, Freek},
   year={2023},
   month=may }

@article{Headrick:2014cta,
    author = "Headrick, Matthew and Hubeny, Veronika E. and Lawrence, Albion and Rangamani, Mukund",
    title = "{Causality \& holographic entanglement entropy}",
    eprint = "1408.6300",
    archivePrefix = "arXiv",
    primaryClass = "hep-th",
    reportNumber = "DCPT-14-33, BRX-TH-6284",
    doi = "10.1007/JHEP12(2014)162",
    journal = "JHEP",
    volume = "12",
    pages = "162",
    year = "2014"
}

@article{Chu:2019etd,
    author = "Chu, Jinwei and Qi, Runze and Zhou, Yang",
    title = "{Generalizations of Reflected Entropy and the Holographic Dual}",
    eprint = "1909.10456",
    archivePrefix = "arXiv",
    primaryClass = "hep-th",
    doi = "10.1007/JHEP03(2020)151",
    journal = "JHEP",
    volume = "03",
    pages = "151",
    year = "2020"
}

@article{DeWolfe:2020vjp,
    author = "DeWolfe, Oliver and Levin, Joshua and Smith, Graeme",
    title = "{Multipartite optimized correlation measures and holography}",
    eprint = "2007.11587",
    archivePrefix = "arXiv",
    primaryClass = "hep-th",
    doi = "10.1103/PhysRevD.102.066001",
    journal = "Phys. Rev. D",
    volume = "102",
    number = "6",
    pages = "066001",
    year = "2020"
}

@article{Bao:2015bfa,
    author = "Bao, Ning and Nezami, Sepehr and Ooguri, Hirosi and Stoica, Bogdan and Sully, James and Walter, Michael",
    title = "{The Holographic Entropy Cone}",
    eprint = "1505.07839",
    archivePrefix = "arXiv",
    primaryClass = "hep-th",
    reportNumber = "CALT-TH-2015-020, IPMU15-0074, SLAC-PUB-16294, SU-ITP-15-08",
    doi = "10.1007/JHEP09(2015)130",
    journal = "JHEP",
    volume = "09",
    pages = "130",
    year = "2015"
}

@article{Balasubramanian:2013lsa,
    author = "Balasubramanian, Vijay and Chowdhury, Borun D. and Czech, Bartlomiej and de Boer, Jan and Heller, Michal P.",
    title = "{Bulk curves from boundary data in holography}",
    eprint = "1310.4204",
    archivePrefix = "arXiv",
    primaryClass = "hep-th",
    doi = "10.1103/PhysRevD.89.086004",
    journal = "Phys. Rev. D",
    volume = "89",
    number = "8",
    pages = "086004",
    year = "2014"
}

@article{Fujita:2011fp,
    author = "Fujita, Mitsutoshi and Takayanagi, Tadashi and Tonni, Erik",
    title = "{Aspects of AdS/BCFT}",
    eprint = "1108.5152",
    archivePrefix = "arXiv",
    primaryClass = "hep-th",
    reportNumber = "IPMU-11-0136, MIT-CTP-4289",
    doi = "10.1007/JHEP11(2011)043",
    journal = "JHEP",
    volume = "11",
    pages = "043",
    year = "2011"
}

@article{Takayanagi:2011zk,
    author = "Takayanagi, Tadashi",
    title = "{Holographic Dual of BCFT}",
    eprint = "1105.5165",
    archivePrefix = "arXiv",
    primaryClass = "hep-th",
    reportNumber = "IPMU11-0091",
    doi = "10.1103/PhysRevLett.107.101602",
    journal = "Phys. Rev. Lett.",
    volume = "107",
    pages = "101602",
    year = "2011"
}

@article{Ju:2025mvz,
    author = "Ju, Xin-Xiang and Liu, Bo-Hao and Sun, Ya-Wen and Zhao, Yang",
    title = "{Modular Hamiltonian of holographic time band states}",
    eprint = "2504.13739",
    archivePrefix = "arXiv",
    primaryClass = "hep-th",
    month = "4",
    year = "2025"
}

@article{Gadde_2023,
   title={Towards a classification of holographic multi-partite entanglement measures},
   volume={2023},
   ISSN={1029-8479},
   url={http://dx.doi.org/10.1007/JHEP08(2023)202},
   DOI={10.1007/jhep08(2023)202},
   number={8},
   journal={Journal of High Energy Physics},
   publisher={Springer Science and Business Media LLC},
   author={Gadde, Abhijit and Krishna, Vineeth and Sharma, Trakshu},
   year={2023},
   month=aug }

@article{Ju:2023bjl,
    author = "Ju, Xin-Xiang and Pan, Wen-Bin and Sun, Ya-Wen and Wang, Yuan-Tai",
    title = "{Generalized Rindler Wedge and Holographic Observer Concordance}",
    eprint = "2302.03340",
    archivePrefix = "arXiv",
    primaryClass = "hep-th",
    month = "2",
    year = "2023"
}

@article{Gadde_2024,
   title={Monotonicity conjecture for multi-party entanglement. Part I},
   volume={2024},
   ISSN={1029-8479},
   url={http://dx.doi.org/10.1007/JHEP02(2024)025},
   DOI={10.1007/jhep02(2024)025},
   number={2},
   journal={Journal of High Energy Physics},
   publisher={Springer Science and Business Media LLC},
   author={Gadde, Abhijit and Jain, Shraiyance and Krishna, Vineeth and Kulkarni, Harshal and Sharma, Trakshu},
   year={2024},
   month=feb }

@article{Ju_2024,
   title={Entanglement structures from modified IR geometry},
   volume={2024},
   ISSN={1029-8479},
   url={http://dx.doi.org/10.1007/JHEP07(2024)181},
   DOI={10.1007/jhep07(2024)181},
   number={7},
   journal={Journal of High Energy Physics},
   publisher={Springer Science and Business Media LLC},
   author={Ju, Xin-Xiang and Lai, Teng-Zhou and Liu, Bo-Hao and Pan, Wen-Bin and Sun, Ya-Wen},
   year={2024},
   month=jul }

@article{Gadde:2025csh,
    author = "Gadde, Abhijit and Jain, Shraiyance and Kulkarni, Harshal",
    title = "{Multi-partite entanglement monotones}",
    eprint = "2406.17447",
    archivePrefix = "arXiv",
    primaryClass = "quant-ph",
    reportNumber = "TIFR/TH/24-10",
    month = "6",
    year = "2025"
}

@article{liu2024,
    author = "Liu, Bowei and Zhang, Junjia and Ohyama, Shuhei and Kusuki, Yuya and Ryu, Shinsei",
    title = "{Multiwavefunction overlap and multientropy for topological ground states in (2+1) dimensions}",
    eprint = "2410.08284",
    archivePrefix = "arXiv",
    primaryClass = "cond-mat.str-el",
    reportNumber = "RIKEN-iTHEMS-Report-24 KYUSHU-HET-297",
    doi = "10.1103/tjcf-yryh",
    journal = "Phys. Rev. B",
    volume = "112",
    number = "12",
    pages = "125160",
    year = "2025"
}

@article{Ma_2024,
   title={Multipartite entanglement measures: A review},
   ISSN={2667-3258},
   url={http://dx.doi.org/10.1016/j.fmre.2024.03.031},
   DOI={10.1016/j.fmre.2024.03.031},
   journal={Fundamental Research},
   publisher={Elsevier BV},
   author={Ma, Mengru and Li, Yinfei and Shang, Jiangwei},
   year={2024},
   month=jun }

@article{balasubramanian2024,
    author = "Balasubramanian, Vijay and Kang, Monica Jinwoo and Murdia, Chitraang and Ross, Simon F.",
    title = "{Signals of multiparty entanglement and holography}",
    eprint = "2411.03422",
    archivePrefix = "arXiv",
    primaryClass = "hep-th",
    doi = "10.1007/JHEP06(2025)068",
    journal = "JHEP",
    volume = "06",
    pages = "068",
    year = "2025"
}

@article{Gadde:2024taa,
    author = "Gadde, Abhijit and Harper, Jonathan and Krishna, Vineeth",
    title = "{Multi-invariants and bulk replica symmetry}",
    eprint = "2411.00935",
    archivePrefix = "arXiv",
    primaryClass = "hep-th",
    reportNumber = "YITP-24-130, LCTP-24-19",
    doi = "10.1007/JHEP06(2025)116",
    journal = "JHEP",
    volume = "06",
    pages = "116",
    year = "2025"
}

@article{Yuan_2025,
   title={Reflected Multientropy and Its Holographic Dual},
   volume={135},
   ISSN={1079-7114},
   url={http://dx.doi.org/10.1103/76vs-rxcs},
   DOI={10.1103/76vs-rxcs},
   number={9},
   journal={Physical Review Letters},
   publisher={American Physical Society (APS)},
   author={Yuan, Ma-Ke and Li, Mingyi and Zhou, Yang},
   year={2025},
   month=aug }

@article{Basak2025,
    author = "Basak, Jaydeep Kumar and Malvimat, Vinay and Yoon, Junggi",
    title = "{A New Genuine Multipartite Entanglement Measure: from Qubits to Multiboundary Wormholes}",
    eprint = "2411.11961",
    archivePrefix = "arXiv",
    primaryClass = "hep-th",
    month = "11",
    year = "2024"
}

@article{iizuka2025GM,
    author = "Iizuka, Norihiro and Nishida, Mitsuhiro",
    title = "{Genuine multientropy and holography}",
    eprint = "2502.07995",
    archivePrefix = "arXiv",
    primaryClass = "hep-th",
    doi = "10.1103/714c-byxq",
    journal = "Phys. Rev. D",
    volume = "112",
    number = "2",
    pages = "026011",
    year = "2025"
}

@article{iizuka2025MGM,
    author = "Iizuka, Norihiro and Lin, Simon and Nishida, Mitsuhiro",
    title = "{More on genuine multientropy and holography}",
    eprint = "2504.16589",
    archivePrefix = "arXiv",
    primaryClass = "hep-th",
    doi = "10.1103/x76v-mr6n",
    journal = "Phys. Rev. D",
    volume = "112",
    number = "6",
    pages = "066014",
    year = "2025"
}

@article{iizuka2025,
    author = "Iizuka, Norihiro and Miyata, Akihiro and Nishida, Mitsuhiro",
    title = "{Multipartite Markov gaps and entanglement wedge multiway cuts}",
    eprint = "2507.15262",
    archivePrefix = "arXiv",
    primaryClass = "hep-th",
    doi = "10.1007/JHEP10(2025)148",
    journal = "JHEP",
    volume = "10",
    pages = "148",
    year = "2025"
}

@article{Jiang:2025iet,
    author = "Jiang, Libo and Liu, Yan",
    title = "{The holographic dual of the GHZ state}",
    eprint = "2508.17898",
    archivePrefix = "arXiv",
    primaryClass = "hep-th",
    month = "8",
    year = "2025"
}

@article{bao2025,
    author = "Bao, Ning and Furuya, Keiichiro and Naskar, Joydeep",
    title = "{Tripartite Entanglement Signal from Multipartite Entanglement of Purification}",
    eprint = "2509.08209",
    archivePrefix = "arXiv",
    primaryClass = "hep-th",
    month = "9",
    year = "2025"
}

@article{balasubramanian2025,
    author = "Balasubramanian, Vijay and Kang, Monica Jinwoo and Cummings, Charlie and Murdia, Chitraang and Ross, Simon F.",
    title = "{Purely GHZ-like entanglement is forbidden in holography}",
    eprint = "2509.03621",
    archivePrefix = "arXiv",
    primaryClass = "hep-th",
    month = "9",
    year = "2025"
}

@article{Yuan:2025dgx,
    author = "Yuan, Ma-Ke and Li, Mingyi and Zhou, Yang",
    title = "{Multi-entropy from Linking in Chern-Simons Theory}",
    eprint = "2510.18408",
    archivePrefix = "arXiv",
    primaryClass = "hep-th",
    month = "10",
    year = "2025"
}

@article{Akella:2025owv,
    author = "Akella, Sriram",
    title = "{Tripartite entanglement in the HaPPY code is not holographic}",
    eprint = "2510.08520",
    archivePrefix = "arXiv",
    primaryClass = "hep-th",
    month = "10",
    year = "2025"
}

@article{ahn2025,
    author = "Ahn, Byoungjoon and Basak, Jaydeep Kumar and Kim, Keun-Young and Koo, Gwon Bin and Malvimat, Vinay and Yoon, Junggi",
    title = "{Probing the Hierarchy of Genuine Multipartite Entanglement with Generalized Latent Entropy}",
    eprint = "2510.19922",
    archivePrefix = "arXiv",
    primaryClass = "hep-th",
    month = "10",
    year = "2025"
}

@article{Balasubramanian:2025jhq,
    author = "Balasubramanian, Vijay and Jiang, Hanzhi and Ross, Simon F.",
    title = "{Time Evolution of Multi-Party Entanglement Signals}",
    eprint = "2511.16729",
    archivePrefix = "arXiv",
    primaryClass = "hep-th",
    month = "11",
    year = "2025"
}
\bibliographystyle{JHEP}

\end{document}